\documentclass[12pt]{article}
\usepackage[margin=0.9in]{geometry}
\usepackage{authblk}
\usepackage[utf8]{inputenc}
\usepackage{amsmath}
\usepackage{amssymb}
\usepackage{amsfonts}
\usepackage{mathtools}
\usepackage{gensymb}
\usepackage{enumitem}
\usepackage{graphicx}
\usepackage{scrextend}
\usepackage{blindtext}
\usepackage{caption}
\usepackage{url}
\usepackage{comment}
\usepackage{subfig}
\usepackage{circuitikz}
\usepackage{listings}
\usepackage{color}
 \usepackage{microtype}
\usepackage{graphicx}
\usepackage{booktabs}
\usepackage{hyperref}% for professional tables
\renewcommand{\leq}{\leqslant}
\renewcommand{\ge}{\geqslant}

\newcommand {\matl}{\left[ \begin{matrix}}
\newcommand {\matr}{\end{matrix}\right]}
\newcommand {\Exp}{ \mathbb E }

\newcommand{\kl}{D_{\operatorname{KL}}}

\newcommand{\tv}{D_{\operatorname{TV}}}

\DeclareMathOperator*{\Expw}{\Exp}

\usepackage{bbm}

\usepackage{algorithm}% http://ctan.org/pkg/algorithms
\usepackage{algpseudocode}% http://ctan.org/pkg/algorithmicx

\usepackage{mathrsfs}
\usepackage{amsmath}
\usepackage{amssymb}
\usepackage{mathtools}
\usepackage{amsthm}
\usepackage{parskip}
\usepackage[capitalize,noabbrev]{cleveref}
\usepackage{dsfont}
\theoremstyle{plain}
\newtheorem{theorem}{Theorem}[section]
\newtheorem{proposition}[theorem]{Proposition}
\newtheorem{lemma}[theorem]{Lemma}
\newtheorem{corollary}[theorem]{Corollary}
\theoremstyle{definition}
\newtheorem{definition}[theorem]{Definition}
\newtheorem{assumption}[theorem]{Assumption}
\theoremstyle{remark}
\newtheorem{remark}[theorem]{Remark}
\definecolor{codegreen}{rgb}{0,0.6,0}
\definecolor{codegray}{rgb}{0.5,0.5,0.5}
\definecolor{codepurple}{rgb}{0.58,0,0.82}
\definecolor{backcolour}{rgb}{0.95,0.95,0.92}
 
\lstdefinestyle{mystyle}{
    backgroundcolor=\color{backcolour},   
    commentstyle=\color{codegreen},
    keywordstyle=\color{magenta},
    numberstyle=\tiny\color{codegray},
    stringstyle=\color{codepurple},
    basicstyle=\footnotesize,
    breakatwhitespace=false,         
    breaklines=true,                 
    captionpos=b,                    
    keepspaces=true,                 
    numbers=left,                    
    numbersep=5pt,                  
    showspaces=false,                
    showstringspaces=false,
    showtabs=false,                  
    tabsize=2
}
 
\date{\today}
\allowdisplaybreaks 

\usepackage{color}

\usepackage{bm}

\usepackage{cleveref}
\crefname{assumption}{assumption}{assumptions}
\Crefname{assumption}{Assumption}{Assumptions}

\crefname{theorem}{theorem}{theorems}
\Crefname{theorem}{Theorem}{Theorems}

\crefname{figure}{fig.}{}
\Crefname{figure}{Fig.}{}

\title{\textbf{Huber-robust likelihood ratio tests \\ for composite nulls and alternatives}}
\author[1]{Aytijhya Saha}
\author[2]{Aaditya Ramdas}

\affil[1]{Massachusetts Institute of Technology.  \texttt{aytijhya@mit.edu}}
\affil[2]{Carnegie Mellon University. \texttt{aramdas@cmu.edu}}
\begin{document}
\date{}

\date{}

\maketitle

\begin{abstract}
We propose an e-value based framework for testing arbitrary composite nulls against composite alternatives, when an $\epsilon$ fraction of the data can be arbitrarily corrupted. Our tests are inherently sequential, being valid at arbitrary data-dependent stopping times, but they are new even for fixed sample sizes, giving type-I error control without any regularity conditions. We first prove that least favourable distribution (LFD) pairs, when they exist, yield optimal e-values for testing arbitrary composite nulls against composite alternatives. Then we show that if an LFD pair exists for some composite null and alternative, then the LFDs of Huber’s $\epsilon$-contamination or total variation (TV) neighborhoods around that specific pair form the optimal LFD pair for the corresponding robustified composite hypotheses. Furthermore, where LFDs do not exist, we develop new robust composite tests for general settings.
Our test statistics are a nonnegative supermartingale under the (robust) null, even under a sequentially adaptive (non-i.i.d.) contamination model where the conditional distribution of each observation given the past data lies within an $\epsilon$ TV ball of some distribution in the original composite null. When LFDs exist, our supermartingale grows to infinity exponentially fast under any distribution in the ($\epsilon$ TV-corruption of the) alternative at the optimal rate. When LFDs do not exist, we provide an asymptotic growth rate analysis, showing that as $\epsilon \to 0$, the exponent converges to the corresponding Kullback–Leibler divergence, recovering the classical optimal non-robust rate.
Simulations validate the theory and demonstrate reasonable practical performance.
\end{abstract}

\section{Introduction}

In statistical hypothesis testing, the assumption that hypothesized models are perfectly specified is often far from reality. 
Real-world data rarely conforms perfectly to our idealized models, making it crucial to develop robust testing methodologies that can withstand small --- and potentially adversarial --- deviations from the idealized models.

Let $\mathcal{M}$ denote the set of all probability distributions over some measurable space $(\Omega, \mathcal{A})$. We consider both batch and sequential tests, but using the latter as a way of getting to the former.
In the former setting, we observe a batch of data $X_1,\dots,X_n$ from some unknown distribution $Q \in \mathcal M$. In the latter setting, we sequentially observe data points $X_1,X_2,\cdots$ from $Q$. 

Given $\alpha \in (0,1)$, this paper will consider the general problem of designing a level-$\alpha$ test for the null hypothesis
$Q \in H_0$ vs.\  the alternative hypothesis that $Q \in H_1$, for some given sets $H_0, H_1 \subset \mathcal M$, when an $\epsilon$ fraction of the data can be arbitrarily corrupted by an adversary or, more generally, when the true data distribution lies within an $\epsilon$ neighborhood of the hypothesized models. 
%We will formalize our corruption model later, first using total variation (TV) balls and then allowing sequentially adaptive corruptions. Ignoring the sequential adaptivity,  the effective problem becomes one of testing $Q \in H_0^\epsilon$ against $Q \in H_1^\epsilon$, where $H_0^\epsilon$ consists of an $\epsilon$-TV ball around the original composite null $H_0$ (and analogously for $H_1^\epsilon$). 
Our tests are valid (control type-1 error at $\alpha$) at arbitrary stopping times;  these yield batch tests at fixed times as a special case. 
%Our tests in both settings are new: we are not aware of any (TV) robust level-$\alpha$ tests for general composite $H_0,H_1$. 

Our advances build on two threads of the literature: classic papers in robustness and recent advances in composite null testing using \emph{e-values}. First, for a simple point null and alternative, Huber \cite{huber1965robust} constructed a robust sequential likelihood ratio test. This will be a central starting point for the current paper. While Huber constructed least favourable distribution (LFD) pair to design optimal tests in both finite sample size and sequential settings, he did not provide a way to construct ``power-one" sequential tests that can control the type-1 error at a prespecified level $\alpha$. 

The second thread of literature that is relevant involves recent fundamental advances in (non-robust) composite hypothesis testing. For example, 
Grünwald et al. \cite{grunwald2020safe} introduces growth-rate optimal (GRO) e-variables
for testing problems with composite null and alternative.
The universal inference method~\cite{wasserman2020universal} used an \emph{e-value} to propose a level-$\alpha$ test for any composite null without requiring any regularity conditions, but requiring reference measures. Recently, Larsson et al.~\cite{larsson2024numeraire} constructed the log-optimal numeraire e-value based on the \emph{Reverse Information Projection} (RIPr), which also does not require any regularity conditions or even reference measures, and is always more powerful than universal inference. We will introduce the history and definition of the RIPr in detail later, but it has been of great interest in information theory, and more recently, statistics, for several decades~\cite{li1999mix,grunwald2020safe,lardy2023universal}. Since the numeraire always dominates the universal inference e-value in the non-robust setting, we only extend the numeraire to the robust setting. We use the robust numeraire to construct a nonnegative supermartingale under the (contaminated) null, that is easy to threshold to get a level-$\alpha$ test. 
%hile the finite-sample optimality of our methods is unresolved, our supermartingales are, at least in a certain weak limiting sense, proven to be \emph{log-optimal}, a notion of optimality that has been long employed in information theory~\cite{kelly1956new,cover1987log}.

To summarize, this work combines and significantly extends the techniques in~\cite{huber1965robust,grunwald2020safe} and~\cite{larsson2024numeraire} to yield new fixed-sample and sequential robust tests for general composite nulls and alternatives.  Beyond  testing, the proposed supermartingales can also directly serve as building blocks for  Huber-robust change detection \cite{shin2022detectors} and localization \cite{saha2025post-detection}.

The rest of this introduction recaps Huber's idea, introduces e-values and test supermartingales, discusses related work in more detail and delineates our contributions relative to these.

% ;  see~\cite{li1999mix,grunwald2020safe,lardy2023universal}for earlier developments of the RIPr concept.

\subsection{Huber's proposal for a simple null versus simple alternative and our modification}

Consider sequentially testing 
$Q=P_0$ versus $Q=P_1$, for some given $P_0,P_1 \in \mathcal M$. By this, we mean that we observe a stream of data from $P_0^\infty$ or from $P_1^\infty$. Let $p_0$ and $p_1$ be the respective densities with respect to some dominating measure $\mu$. The classical sequential probability ratio test (by Wald) for this testing problem is not robust: a single factor $p_1(X_j)/p_0(X_j)$ equal or close to $0$ or $\infty$ may ruin the entire nonnegative martingale $T_n=\prod_{i=1}^n\frac{p_1(X_i)}{p_0(X_i)}$.

Huber \cite{huber1965robust} formalized the problem of robustly testing a simple  $P_0$ against a simple  $P_1$ by assuming that the true underlying distribution $Q$ lies in some appropriate neighborhood of either of the idealized models $P_0$ or $P_1$. To account for the possibility of small deviations from the idealized models $P_j$, he expanded the hypotheses to:
\begin{equation}
\label{contamination-model}
    H^{\epsilon}_j=\{Q\in\mathcal{M}: Q=(1-\epsilon)P_j+\epsilon H, H\in\mathcal{M}\},
\end{equation}
or, more generally, to
\begin{equation}
\label{tv-model}
    H^{\epsilon}_j=\{ Q \in \mathcal{M} : \tv(P_j, Q) \leq \epsilon \},
\end{equation}
where $j=0,1$ and  $\tv$ denotes the total variation distance. Thus, a robust test of $P_0$ versus $P_1$ is effectively a test of the composite null $H^{\epsilon}_0$ against the composite alternative $H^{\epsilon}_1$. Note that \eqref{tv-model} is strictly larger than \eqref{contamination-model}. He proposed risk-optimal tests $\phi$ for the following three notions of optimality:
\begin{equation}
\begin{aligned}
\label{eq:risk-optimality}
    &\min_\phi \max_{j=0,1} \sup_{Q_j\in H^{\epsilon}_j} R(Q_i,\phi)\\
     &\min_\phi \sup_{Q_1\in H^{\epsilon}_1} R(Q_1,\phi) 
        \;\; \text{s.t. } \sup_{Q_0\in H^{\epsilon}_0} R(v_0,\phi)\leq \alpha\\
        &\min_\phi \sup_{(Q_0,Q_1)\in H^{\epsilon}_0\times H^{\epsilon}_1}
        \big[ \pi_0 R(Q_0,\phi) + \pi_1 R(Q_1,\phi)\big].
\end{aligned}
\end{equation}

He solved the preceding problems by looking for a ``least favorable distribution pair" in terms of risk, where the risk $R$ for a hypothesis test $\phi$ is defined as
\begin{equation}
\label{risk}
      R(P_0, \phi) = C_0 P_{0}(\phi \text{ accepts } H_1),        ~~
        R(P_1, \phi) = C_1 P_{1}(\phi \text{ accepts } H_0).
\end{equation}
The notion of a least favorable distribution pair in terms of risk is formally defined below.
\begin{definition}
    A pair $(P_0^*, P_1^*) \in \mathcal{P}_0 \times \mathcal{P}_1$ (having densities $p_0^*$ and $p_1^*$ with respect to some common dominating measure) is called \textit{least favorable distribution (LFD) pair} in terms of risk $R$ for testing $\mathcal{P}_0$ vs. $\mathcal{P}_1$, if for every likelihood ratio test $\phi$ between $P_0^*$ and $P_1^*$,
    \[
       R(P_j^*,\phi) \geq R(P_j,\phi), ~\text{ for all } P_j \in \mathcal{P}_j,~ j=0,1.
    \]
\end{definition}

Huber derived the LFD pair 
$(Q_{0,\epsilon},Q_{1,\epsilon})$ in $H^{\epsilon}_0\times H^{\epsilon}_1$, such that any probability ratio test of $Q_{0,\epsilon}$ vs.\ $Q_{1,\epsilon}$ is precisely the optimal test for $H^{\epsilon}_0$ vs.\ $H^{\epsilon}_1$.

He showed that their likelihood ratio $\frac{q_{1,\epsilon}(X)}{q_{0,\epsilon}(X)}=\max\{c^{\prime},\min\{c^{\prime\prime},\frac{p_1(X)}{p_0(X)}\}\}$ is a truncation of the original likelihood ratio $\frac{p_1(X)}{p_0(X)}$, at some specific values of the $c^{\prime}$ and $c^{\prime\prime}$, which depend on $\epsilon, P_0$ and $P_1$. And hence, Huber's truncated likelihood ratio test is the optimal solution for the finite-sample problems defined in \eqref{eq:risk-optimality}. He also considered (two-sided) sequential tests with the stopping time defined below:
\begin{equation}
    N=\inf\left\{n\geq 0:S_n\leq a \text{ or } S_n\geq b \right\},
\end{equation}
where $a<b$ are fixed numbers. The testing procedure is to stop at stage $N$ and reject $H^{\epsilon}_0$ if $S_n\geq b$ and accept $H^{\epsilon}_0$ if $S_n\leq a$. Huber \cite{huber1965robust} proved that  $Q_{j,\epsilon}\in H_{j}^\epsilon, j=0,1$ are ``least favorable distribution pair"in the sequential setting as well, i.e,
\begin{equation}
    Q_{0,\epsilon}[S_n\geq b]=\sup\{Q[S_n\geq b]: Q\in H_0^\epsilon\} ,  \text{ and }
\end{equation}
\begin{equation}
    Q_{1,\epsilon}[S_n\leq a]=\sup\{Q[S_n\leq a]: Q\in H_1^\epsilon\} .
\end{equation}
and hence, the test is optimal for type-I and type-II error probabilities. 

\noindent \textbf{Our level-$\alpha$, anytime-valid variant of Huber's test:} Huber does not describe how exactly to calculate the thresholds if a targeted type-I error of $\alpha$ is specified. 
% Thus, the stopping rule is not particularly practical, despite providing a tradeoff between the type-II and type-II errors. 
% Further, Huber's method has a pre-determined stopping rule and does not provide ``anytime-valid guarantees",  meaning its validity is only assured at the specific, pre-defined stopping point $N$, but one cannot make inferences at other stopping times (for example, if the test was stopped for any other reason).
We prove an interesting general result (that holds true even outside the context of this paper) that the likelihood ratio of the LFD pair is actually an e-value for the composite null $H_0^\epsilon$. Further, it can also be shown to be growth rate optimal (log-optimal). That allows us 
to construct
a level-$\alpha$ test, which has power one. As a side benefit, it will also be an ``anytime-valid" test \cite{ramdas2022game} for this problem, to be defined later.  We will actually construct an ``anytime-valid e-value" (or an \emph{e-process}, to be defined later) that is a continuous measure of evidence against the null, which can be thresholded at $1/\alpha$ to yield a level-$\alpha$ test. Though we do not expand on the following point further, we note in passing that this e-process can also be converted to an ``anytime-valid p-value'', which would allow us to continuously monitor the data throughout the experiment, report a p-value at any stopping time, and also reject the null at a level $\alpha$ when that p-value drops below $\alpha$. A fixed sample size test can be obtained by simply stopping monitoring at a fixed time $n$.

%We begin by constructing a robust sequential test for $H^{\epsilon}_0$ vs.\ $H^{\epsilon}_1$ as defined in \eqref{contamination-model} or \eqref{tv-model},  within the recently emerging framework of sequential anytime-valid inference \cite{ramdas2022game}. 
This framework is an offshoot of Robbins' power-one sequential tests (or one-sided sequential probability ratio tests). The framework is rooted in the construction of “test supermartingales”~\cite{shafer2019game} --- or, more generally, “e-processes” \cite{ramdas2022testing} --- which can be interpreted as the wealth of a gambler playing a stochastic game; these are introduced below. To achieve anytime-valid inference, we construct a test supermartingale under $H^{\epsilon}_0$, ensuring its validity at arbitrary data-dependent stopping times, accommodating continuous monitoring and analysis of accumulating data, and optional stopping or continuation. In later sections, we extend our method to composite nulls and composite alternatives, leveraging the techniques of ``predicable plug-in'' and ``Reverse Information Projection'', which are well-known yet sophisticated tools employed in the non-robust setting.

{Note that $\epsilon$ should be chosen apriori, according to how suspicious one is of the data. Huber's model does not allow $\epsilon$ to be chosen based on the data.  The larger the $\epsilon$, the more robust the procedure will be to adversarial corruptions (stronger type-1 error guarantee), but the more conservative the tests will be (weaker power guarantee).}

Before we delve into the details of our method, it is crucial to discuss what test (super)martingales are and how they play a key role in constructing sequential anytime-valid tests. 

\begin{comment}
\paragraph{E-process.}
Consider a nonnegative sequence of adapted random variables $E \equiv \{E_n\}_{t\geq0}$ and
let $H_0$, the null hypothesis, be a set of distributions. We call $E$ as an e-process for $H_0$, if 
\begin{equation}
    \mathbb{E}_{\mathbb{P}}[E_\tau] \leq 1, \text{ for all stopping times }\tau, \text{ for all } \mathbb{P}\in H_0
\end{equation}
Large values of the e-process encode evidence against the null. (Ideally, the evidence $E$ should increase to infinity under $H_1$, almost surely.) Further, suppose we stop and reject the null at the stopping time
\begin{equation}\label{eq:stopping}
    \tau_{\alpha}=\inf\left\{t\geq 1 : E_n\geq \frac{1}{\alpha}\right\}.
\end{equation}
This rule results in a level $\alpha$ sequential test, meaning that if the null is true, the probability that it ever stops falsely rejecting the null is at most $\alpha$. This is easily seen by applying Ville's inequality to the stopped e-process $E_{\tau_{\alpha}}$.
\end{comment}

\subsection{Background on e-values, test supermartingales, one-sided tests, growth rate and consistency}
% \subsubsection{Test (super)martingale and Ville's Inequality} 
An e-variable is a nonnegative random variable that has expectation at most one under the null. Mathematically, a nonnegative random variable $B$ is an e-variable for the null $P \in H_0^\epsilon$ if $\mathbb E_P(B)\leq 1$, for all $P\in H_0^\epsilon$. The value
realized by an e-variable (that is, its instantiation) is called an e-value.

An integrable process $M \equiv \{M_n\}$ 
that is adapted to a filtration $\mathscr{F} \equiv \{\mathscr{F}_n\}_{n\geq 0}$,  is called a \emph{martingale} for $H_0^\epsilon$ if 
\begin{equation}\label{eq:martingale}
   \mathbb{E}_{{P}}[M_n \mid \mathscr{F}_{n-1}] = M_{n-1} , \forall P\in H_0^\epsilon  
\end{equation}
 for all $n\geq 1$. $M$ is called a \emph{supermartingale} for $H_0^\epsilon$  if for all $n\geq 1$,
\begin{equation}\label{eq:supermartingale}
   \mathbb{E}_{{P}}[M_n \mid \mathscr{F}_{n-1}] \leq M_{n-1}, \forall P\in H_0^\epsilon .
\end{equation}
Crucially, $M$ is called a \emph{test (super)martingale} for $H_0^{\epsilon}$ if it is a (super)martingale \emph{for every $ P \in H_0^{\epsilon}$}, and if it is nonnegative with $M_0=1$. A stopping time $\tau$ is a nonnegative integer-valued random variable such that $\{\tau \leq n\} \in \mathscr{F}_{n}$ for each $n\in \mathbb N$.
 Denote by $\mathcal{T}$ the set of all stopping times, including ones that may never stop.

Ville's inequality \cite{ville1939etude} implies that the test (super)martingale satisfies 
\begin{equation}
    \sup_{P\in H_0^{\epsilon}}P(\exists n\in \mathbb N :M_n\geq 1/\alpha)\leq \alpha.
\end{equation}
See \cite[Lemma 1]{howard2020time} for a short proof. 
The above equation is equivalent to $P(M_\tau\geq 1/\alpha)\leq \alpha, \forall \tau \in \mathcal{T}, P \in H_0^{\epsilon}$; see \cite[Lemma 3]{howard2021time}.
This ensures that if 
 we stop and reject the null at the stopping time
\begin{equation}\label{eq:stopping}
    \tau_{\alpha}=\inf\left\{n\geq 1 : M_n\geq 1/{\alpha}\right\},
\end{equation}
it results in a level-$\alpha$ sequential test, meaning that if the null is true, the probability that it ever stops falsely rejecting the null is at most $\alpha$ (under the null, $ \tau_{\alpha} = \infty$ with probability $1-\alpha$). 
% Note that this is a one-sided or ``power-one'' test that does not stop to accept the null, a phrase coined by~\cite{darling1968some}.

The above definition of a sequential test  differs from Wald’s original ideas \cite{wald1945sequential}. In the latter, the null hypothesis might eventually be accepted or rejected, with a predetermined stopping rule based on the desired Type-I and Type-II error rates. In contrast, our framework aligns with Robbins' ``power-one tests'' \cite{darling1968some,robbins1970statistical}, or one-sided tests, where one only specifies a target Type-I error level and we only stop for rejecting the null, but never stop for acceptance. Such a test is called consistent if it is (asymptotically) power one under any alternative $P \in H^1_\epsilon$.

A test (super)martingale is called consistent (or power one) for $H_1^{\epsilon}$, if $\lim_{n\to \infty} M_n = \infty$ almost surely for any $ P \in H_1^{\epsilon}$, meaning that under the alternative, it accumulates infinite evidence against the null in the limit and eventually crosses any finite threshold for rejection. As Kelly noted in his seminal work~\cite{kelly1956new}, the evidence grows exponentially fast in i.i.d.\ settings, so one can define the ``growth rate'' of $M$ is defined as $\inf_{ P \in H_1^{\epsilon}} \liminf_{n \to \infty} \log M_n/n$~\cite{shafer2021testing,grunwald2020safe,waudby2023estimating}. A positive growth rate implies consistency. The test (super)martingale with the largest possible growth rate is called log-optimal, an optimality notion advocated by many influential researchers, like \cite{breiman1961optimal}, \cite{cover1987log}, \cite{shafer2021testing} and \cite{grunwald2020safe} amongst others.

% \subsubsection{Non-robust anytime-valid SPRT.} 
If $\epsilon=0$ and $H_0 =\{P_0\}, H_1=\{P_1\}$, then we are back to the standard non-robust testing of $P_0$ against $P_1$ (except that we desire a one-sided, or power-one, sequential test). In this setting, 
\begin{equation}\label{eq:sprt}
T_n=\prod_{i=1}^n\frac{p_1(X_i)}{p_0(X_i)}
\end{equation}
is a test martingale for $P_0$. Recalling \eqref{eq:stopping}, we may reject the null at the stopping time $\inf\left\{n\geq 1 : T_n\geq 1/{\alpha}\right\}$ to get a level $\alpha$ sequential test. The growth rate of the test reduces to $\mathbb E_{P_1}\log\frac{p_1(X)}{p_0(X)} =\kl(P_1,P_0)$. And it is the optimal growth rate for testing $P_0$ vs.\ $P_1$ \cite{shafer2021testing}, a result with roots in \cite{kelly1956new} and \cite{breiman1961optimal}.

\subsection{Related Work}

The formalization of robust statistics can be traced back to the mid-20th century. Early pioneers in the development of the concept include
\cite{tukey1960survey,huber1965robust,huber1968robust}. Prominent ideas in robust estimation include, among others, M-estimators, trimming, and influence functions. The
use of M-estimators in robust statistical procedures dates back to \cite{huber1964robust}, which achieves
robustness by curbing and bounding the influence that individual data points can make on the
statistics. On the other hand, trimming refers
to the practice of directly discarding outliers \cite{anscombe1960rejection}, and trimmed means have long
been known to be robust \cite{bickel1965some}.  More recent developments extend this framework by incorporating data-driven methods, using Wasserstein, MMD, and entropic-regularized Wasserstein distances. For example,  \cite{gao2018robust} constructs uncertainty
sets under the null and the alternative distributions, which are sets centered around
the empirical distribution defined via the Wasserstein metric. \cite{wang2022data} proposes a data-driven approach to robust hypothesis testing using
Sinkhorn uncertainty sets. \cite{magesh2023robust} uses moment-constrained uncertainty sets. \cite{sun2023kernel} introduces a kernel-based approach to robust hypothesis testing, where they construct tests that optimize worst-case performance within an MMD uncertainty set. The main difference of our work from all of these is that we consider (a generalization of) Huber's original $\epsilon$-contamination model as our notion of robustness, which is akin to using the TV distance.

Recently, \cite{park2023robust} proposed a robust version of the universal inference \cite{wasserman2020universal} approach for constructing valid confidence sets under weak regularity conditions, despite possible model misspecification. However, their notion of robustness differs from ours, and they test whether $P_0$ or $P_1$ is closer to the true data-generating distribution, { meaning that before contamination it is a simple versus simple hypothesis test while we consider composite classes (even before contamination).} \cite{wang2023huber} introduced Huber-robust confidence sequences leveraging supermartingales.
The work of \cite{agrawal2024crimed} is related to ours, as they investigated the multi-armed bandit problem under the Huber corruption model and briefly discussed its connection to the mean testing problem in their Appendix. They analyzed the minimizer of a function involving KL divergences, which is related to RIPr used in our methodology. However, their analysis is restricted to the Gaussian setting with known variance, whereas our framework and theory are general. Furthermore, their work addresses only the $\epsilon$ corruption model~\eqref{contamination-model}, while our approach accommodates the more general settings including $\epsilon$-TV neighborhoods~\eqref{tv-model}.
% and we show asymptotic optimality as $\epsilon\to 0$.

Sequential hypothesis testing has a long-standing history, beginning with the sequential probability ratio test (SPRT) of \cite{wald1945sequential}. There have been a few studies in the literature that address the feasibility of robustifying
sequential tests, mostly notably the works by \cite{huber1965robust} and by \cite{quang1985robust}, both of them
robustifying the likelihood ratio by censoring.
\cite{quang1985robust} considered the sequential testing of two distributions $P_{-\epsilon}$ and $P_{\epsilon}$, where these two distributions approach each other as $\epsilon\to 0$ and proved (under regularity assumptions) that the SPRT based on the least favorable pair of distributions 
 given by \cite{huber1965robust} is asymptotically least favorable for the expected sample size and is asymptotically minimax.
While these earlier tests yield valid inference at a particular prespecified
stopping rule, e-values and e-processes (as used in the current paper) have recently emerged as a promising tool to construct ``anytime-valid" tests \cite{ramdas2022game}. Further, these previous works only considered simple versus simple tests (when contamination is excluded), while we work on composite tests.

\subsection{Contributions}

We first show in general that the likelihood ratio of an LFD pair is an e-value, which is also
growth rate optimal (or, log optimal). As a consequence of the result, it follows that Huber's truncated likelihood ratio test statistic \cite{huber1965robust} is in fact already a test supermartingale for $H_0^\epsilon$. The supermartingale tools enable us to generalize the usual Huber contamination model, demonstrating that our sequential test retains type-I error validity even under sequentially adaptive contaminations under the null, as well as at arbitrary stopping times, thereby extending Huber's contributions in this case.

Next, we extend the framework to handle both composite nulls and composite alternative hypotheses (before contamination), a direction that, to the best of our knowledge, has not been pursued previously. We prove that when an LFD pair exists for a composite null and alternative, the LFDs of Huber’s $\epsilon$-contamination neighborhoods around that pair form the optimal LFD pair for the corresponding robustified hypotheses. Together with the earlier result, this enables the construction of optimal tests across all regimes: fixed-sample settings, two-sided (Wald-type) sequential settings, and anytime-valid testing scenarios.

Furthermore, where an LFD pair does not
exist for the pre-contamination composite null and composite alternative, we develop new robust composite tests for general settings, where we utilize very recent advances in composite null hypothesis testing, such as universal inference~\cite{wasserman2020universal} and the reverse information projection~\cite{grunwald2020safe,lardy2023universal,larsson2024numeraire}.  However, for this case, while we cannot prove finite sample log-optimality, we do analyze the growth rate under the alternative, which, as $\epsilon \to 0$, converges to the optimal non-robust growth rate (a certain KL divergence between null and alternative), under some assumptions. To justify these assumptions, we provide sufficient conditions under which they hold.
%In all cases, we analyze the growth rate under the alternative, which, as $\epsilon \to 0$, converges to the optimal non-robust growth rate (a certain KL divergence between null and alternative). Our experimental results support and validate the theoretical findings.
We emphasize that our robust tests control the type-I error at level $\alpha$ without requiring any assumptions on the pre-contamination composite null and composite alternative.

To summarize, this paper provides a relatively comprehensive set of methods for Huber-robust sequential hypothesis testing, with a theory that accurately predicts performance in experiments.

% We believe this is the first study to introduce an anytime-valid robust SPRT. 

\subsection{Paper outline}
The rest of the paper is organized as follows. In \cref{sec:lfd-optimality}, we show that the likelihood ratio of the LFD pair is an e-value, and in fact the log-optimal e-value. In \Cref{lfd-composite}, we prove that if an LFD
pair exists for some composite null and alternative, then the LFDs of Huber’s
$\epsilon$-contamination neighborhoods around that pair form the optimal LFD pair
for the corresponding robustified hypotheses, yielding log-optimal e-values for this case as well. Finally, in \cref{sec:comp-general}, we consider general composite nulls and alternatives for which an LFD pair may not exist, and fixed $\epsilon$ optimality for tractable methods appears hard to prove, and demonstrate that our methods are at least asymptotically optimal as $\epsilon\to 0$.  \Cref{sec:expt} presents a comprehensive set of simulation studies that validate our theoretical findings and demonstrate the practical performance of our approach. This
article is concluded in \cref{sec:conc}. Omitted proofs of the theoretical results
are provided in \cref{a:proofs}.

\section{LFD pairs yield log-optimal e-values}
\label{sec:lfd-optimality}
In this section, we prove a very general result that if a least favorable distribution pair for testing any two hypotheses $\mathcal{P}_0$ vs.\ $\mathcal{P}_1$ exists, then their likelihood ratio yields optimal e-values. The implications are that the least favorable distribution pair that \cite{huber1965robust} defined for solving \eqref{eq:risk-optimality} also solves the problem of finding ``growth rate optimal in worst-case (GROW)" or ``log optimal" e-values. \cite{grunwald2020safe} introduced the latter notion  and showed their existence under several assumptions, including full support of distributions and convexity of $\mathcal{P}_0, \mathcal{P}_1$. Our theorem does not require such assumptions (but needs an LFD pair).

\begin{theorem}
\label{thm:lfd-optimal}
  If  $(P_0^*, P_1^*)$ is a least favorable distribution (LFD) pair in terms of the risk defined in \eqref{risk} for testing $\mathcal{P}_0$ vs.\ $\mathcal{P}_1$, then  $\frac{dP_{1}^*(X)}{dP_{0}^*(X)}$ is a GROW e-variable for testing $\mathcal{P}_0$ vs.\ $\mathcal{P}_1$. That is, if $\mathcal{E}(\mathcal{P}_0)$ denotes the set of all e-variables for $\mathcal{P}_0$, we have that $\frac{dP_{1}^*(X)}{dP_{0}^*(X)}\in\mathcal{E}(\mathcal{P}_0)$ and $$\sup_{E\in \mathcal{E}(\mathcal{P}_0)} \inf_{P_1\in \mathcal{P}_1}\mathbb{E}_{P_1}\!\left(\log E\right)=\inf_{P_1\in \mathcal{P}_1}\mathbb{E}_{P_1}\!\left(\log\tfrac{d P_1^*(X)}{d P_0^*(X)}\right)= \kl(P_1^*,P_0^*)=\inf_{(P_0,P_1)\in\mathcal{P}_0 \times \mathcal{P}_1}\kl(P_1,P_0).$$ 
\end{theorem}
\begin{proof}
    Since $(P_0^*, P_1^*)$ is an LFD pair for testing $\mathcal{P}_0$ vs.\ $\mathcal{P}_1$, for any $P_0\in\mathcal{P}_0,$ we have
\begin{align*}
  &{R(P_0, \phi) \leq R(P_0^*, \phi)} \\
  &{\implies C_0 {P_0}(\phi \text{ accepts } {H_1}) 
      \leq C_0 {P_0^*}(\phi \text{ accepts } {H_1})} \text{ for every LRT } \phi \text{ between } P_0^* \text{ and } P_1^*\\
  &{\implies {P_0}\!\left(\tfrac{d P_1^*}{d P_0^*}(X) > \eta\right) 
      \leq {P_0^*}\!\left(\tfrac{d P_1^*}{d P_0^*}(X) > \eta\right), \text{ for all }\eta }\\
  &{\implies \int_{\eta \geq 0}  {P_0}\!\left(\tfrac{d P_1^*}{d P_0^*}(X) > \eta\right)\, d\eta
      \leq
    \int_{\eta \geq 0}  {P_0^*}\!\left(\tfrac{d P_1^*}{d P_0^*}(X) > \eta\right)\, d\eta} \\
    &{\implies \mathbb{E}_{P_0}\!\left(\tfrac{d P_1^*}{d P_0^*}(X)\right)
      \leq
    \mathbb{E}_{P_0^*}\!\left(\tfrac{d P_1^*}{d P_0^*}(X)\right)
      = 1.}
\end{align*}
Therefore,  $\frac{dP_{1}^*(X)}{dP_{0}^*(X)}$ is an e-variable. Now, for the log-optimality part, we first show that $ \mathbb{E}_{P_1}\!\left(\log\tfrac{d P_1^*(X)}{d P_0^*(X)}\right)
      \geq
    \mathbb{E}_{P_1^*}\!\left(\log\tfrac{d P_1^*(X)}{d P_0^*(X)}\right)$. Since $(P_0^*, P_1^*)$ is the least favorable pair, for testing $\mathcal{P}_0$ vs.\ $\mathcal{P}_1$, for any $P_1\in\mathcal{P}_1$,
    
\begin{align*}
  &{R(P_1, \phi) \leq R(P_1^*, \phi)} \\
  &{\implies C_1 {P_1}(\phi \text{ accepts } {H_0}) 
      \leq C_1 {P_1^*}(\phi \text{ accepts } {H_0})}, \text{ for every LRT } \phi \text{ between } P_0^* \text{ and } P_1^*\\
  &{\implies {P_1}\!\left(\tfrac{d P_1^*}{d P_0^*}(X) \leq \eta\right) 
      \leq {P_1^*}\!\left(\tfrac{d P_1^*}{d P_0^*}(X) \leq \eta\right), \text{ for all }\eta }\\
      &{\implies   {P_1}\!\left(\log\tfrac{d P_1^*(X)}{d P_0^*(X)} > \log\eta\right)\,
      \geq
     {P_1^*}\!\left(\log\tfrac{d P_1^*(X)}{d P_0^*(X)} > \log\eta\right)\ } \text{ for all }\eta >0\\
  &{\implies \int_{t \geq 0}  {P_1}\!\left(\log\tfrac{d P_1^*(X)}{d P_0^*(X)} > t\right)\, dt
      \geq
    \int_{t \geq 0}  {P_1^*}\!\left(\log\tfrac{d P_1^*(X)}{d P_0^*(X)} > t\right)\, dt}.
\end{align*}
Similarly, for any $P_1\in\mathcal{P}_1$,
\begin{align*}
  &{R(P_1, \phi) \leq R(P_1^*, \phi)} \\
  &{\implies C_1 {P_1}(\phi \text{ accepts } {H_0}) 
      \leq C_1 {P_1^*}(\phi \text{ accepts } {H_0})}, \text{ for every LRT } \phi \text{ between } P_0^* \text{ and } P_1^*\\
  &{\implies {P_1}\!\left(\tfrac{d P_1^*}{d P_0^*}(X) < \eta\right) 
      \leq {P_1^*}\!\left(\tfrac{d P_1^*}{d P_0^*}(X) < \eta\right), \text{ for all }\eta }\\
      &{\implies   {P_1}\!\left(\log\tfrac{d P_1^*(X)}{d P_0^*(X)} < \log\eta\right)\,
      \leq
     {P_1^*}\!\left(\log\tfrac{d P_1^*(X)}{d P_0^*(X)} < \log\eta\right)\ } \text{ for all }\eta >0\\
  &{\implies \int_{t \geq 0}  {P_1}\!\left(-\log\tfrac{d P_1^*(X)}{d P_0^*(X)} > t\right)\, dt
      \leq
    \int_{t \geq 0}  {P_1^*}\!\left(-\log\tfrac{d P_1^*(X)}{d P_0^*(X)} > t\right)\, dt}.
\end{align*}
Therefore, for any $P_1\in\mathcal{P}_1$,
\begin{align*}
    \mathbb{E}_{P_1}\!\left(\log\tfrac{d P_1^*(X)}{d P_0^*(X)}\right) &=\int_{t \geq 0}  {P_1}\!\left(\log\tfrac{d P_1^*(X)}{d P_0^*(X)} > t\right)\, dt-\int_{t \geq 0}  {P_1}\!\left(-\log\tfrac{d P_1^*(X)}{d P_0^*(X)} > t\right)\, dt\\
    &\geq\int_{t \geq 0}  {P_1^*}\!\left(\log\tfrac{d P_1^*(X)}{d P_0^*(X)} > t\right)\, dt-\int_{t \geq 0}  {P_1^*}\!\left(-\log\tfrac{d P_1^*(X)}{d P_0^*(X)} > t\right)\, dt\\
     & =
    \mathbb{E}_{P_1^*}\!\left(\log\tfrac{d P_1^*(X)}{d P_0^*(X)}\right) = \kl(P_1^*,P_0^*).
\end{align*}
So, we have shown above that $ \inf_{P_1\in \mathcal{P}_1}\mathbb{E}_{P_1}\!\left(\log\tfrac{d P_1^*(X)}{d P_0^*(X)}\right)= \kl(P_1^*,P_0^*).$ Now, suppose, if possible, $E$ is an e-variable for $\mathcal{P}_0$, which satisfies 
$ \inf_{P_1\in \mathcal{P}_1}\mathbb{E}_{P_1}\!\left(\log E\right)> \kl(P_1^*,P_0^*).$ Then, we would have $\mathbb{E}_{P_1^*}\!\left(\log E\right)> \kl(P_1^*,P_0^*)$, this is a contradiction because $E$ is an e-variable for the simple null ${P}_0^*$ (against simple alternative $P_1^*$) as well and hence its maximum possible growth rate is $\kl(P_1^*,P_0^*)$.  Thus, we have proved
 $\sup_{E\in \mathcal{E}(\mathcal{P}_0)} \inf_{P_1\in \mathcal{P}_1}\mathbb{E}_{P_1}\!\left(\log E\right)=\inf_{P_1\in \mathcal{P}_1}\mathbb{E}_{P_1}\!\left(\log\tfrac{d P_1^*(X)}{d P_0^*(X)}\right)= \kl(P_1^*,P_0^*).$

If $\kl(P_1^*,P_0^*)>\inf_{(P_0,P_1)\in\mathcal{P}_0 \times \mathcal{P}_1} \kl(P_1,P_0)$, then $\kl(P_1^*,P_0^*)> \kl(P_1,P_0)$, for some $(P_0,P_1)\in\mathcal{P}_0 \times \mathcal{P}_1$, but then  $dP_1^*/dP_0^*$ is an e-value for $P_0$ vs $P_1$, whose growth rate 
$\mathbb{E}_{P_1}\!\left(\log\tfrac{d P_1^*(X)}{d P_0^*(X)}\right)\geq\inf_{P_1\in \mathcal{P}_1}\mathbb{E}_{P_1}\!\left(\log\tfrac{d P_1^*(X)}{d P_0^*(X)}\right)= \kl(P_1^*,P_0^*)> \kl(P_1,P_0)$, but
 $\kl(P_1,P_0)$ is largest possible growth rate for an e-value for $P_0$ vs $P_1$, which is a contradiction. 

 Therefore, we finally have 
$$\sup_{E\in \mathcal{E}(\mathcal{P}_0)} \inf_{P_1\in \mathcal{P}_1}\mathbb{E}_{P_1}\!\left(\log E\right)=\inf_{P_1\in \mathcal{P}_1}\mathbb{E}_{P_1}\!\left(\log\tfrac{d P_1^*(X)}{d P_0^*(X)}\right)= \kl(P_1^*,P_0^*)=\inf_{(P_0,P_1)\in\mathcal{P}_0 \times \mathcal{P}_1}\kl(P_1,P_0).$$
\end{proof}

\begin{remark}
   The last part of the theorem $\kl(P_1^*,P_0^*)=\inf_{(P_0,P_1)\in\mathcal{P}_0 \times \mathcal{P}_1}\kl(P_1,P_0)$ also follows from Corollary 1 of \cite{poor1980robust}, but we have a direct proof here.
\end{remark}

\begin{remark}
    It is easy to verify that LFD is symmetric around the null and alternative hypotheses. That is, $(P_0^*, P_1^*)$ is an LFD pair in terms of the risk defined in \eqref{risk} for testing $\mathcal{P}_0$ vs.\ $\mathcal{P}_1$ is equivalent to saying $(P_1^*, P_0^*)$ is an LFD pair in terms of the risk defined in \eqref{risk} for testing $\mathcal{P}_1$ vs.\ $\mathcal{P}_0$. And hence, we also have that
  $\frac{dP_{0}^*(X)}{dP_{1}^*(X)}$ is a GROW e-variable for testing $\mathcal{P}_1$ vs.\ $\mathcal{P}_0$, that is, $\frac{dP_{0}^*(X)}{dP_{1}^*(X)}\in\mathcal{E}(\mathcal{P}_1)$ and $$\sup_{E\in \mathcal{E}(\mathcal{P}_1)} \inf_{P_0\in \mathcal{P}_0}\mathbb{E}_{P_0}\!\left(\log E\right)=\inf_{P_0\in \mathcal{P}_0}\mathbb{E}_{P_0}\!\left(\log\tfrac{d P_0^*(X)}{d P_1^*(X)}\right)= \kl(P_0^*,P_1^*)=\inf_{(P_0,P_1)\in\mathcal{P}_0 \times \mathcal{P}_1}\kl(P_0,P_1),$$ where $\mathcal{E}(\mathcal{P}_1)$ is the set of all e-variables for $\mathcal{P}_1.$
\end{remark}

The above theorem allows us to construct level-$\alpha$ tests using Huber's LFD pairs for testing \eqref{contamination-model} or \eqref{tv-model}, which we elaborate on next. However, it is worth noting that the applicability of this theorem extends beyond these specific settings, for instance, to classical testing problems such as those involving monotone likelihood ratio (MLR) families or the $2 \times 2$ contingency table. Interestingly, \cite{grunwald2020safe} identified GROW e-variables in these latter settings, although their connection to LFDs was not explicitly recognized.

The log-optimality result in \Cref{thm:lfd-optimal} can be extended to other utilities and
divergences. For example, the following theorem shows how it yields the Rényi-optimal e-value, but the proof technique can likely be extended to other concave utilities. 

\begin{theorem}
\label{thm:lfd-optimal-renyi}
Let $\gamma>1$ and consider the utility function defined as $U(x)=x^{1-\gamma}/(1-\gamma)$.
  If  $(P_0^*, P_1^*)$ is a least favorable distribution (LFD) pair in terms of the risk defined in \eqref{risk} for testing $\mathcal{P}_0$ vs.\ $\mathcal{P}_1$, then  $$E^*:=\left(\frac{dP_{0}^*(X)}{dP_{1}^*(X)}\right)^{-\frac{1}{\gamma}}\Bigg/~\mathbb E_{P_{1}^*}\left[\left(\frac{dP_{0}^*(X)}{dP_{1}^*(X)}\right)^{1-\frac{1}{\gamma}}\right]$$ is a Rényi-GROW e-variable for testing $\mathcal{P}_0$ vs.\ $\mathcal{P}_1$. That is, if $\mathcal{E}(\mathcal{P}_0)$ denotes the set of all e-variables for $\mathcal{P}_0$, we have that $E^*\in\mathcal{E}(\mathcal{P}_0)$ and 
  \begin{align*}
      \sup_{E\in \mathcal{E}(\mathcal{P}_0)} \inf_{P_1\in \mathcal{P}_1}\mathbb{E}_{P_1}\!\left(U(E)\right)=\inf_{P_1\in \mathcal{P}_1}\mathbb{E}_{P_1}\!\left(U(E^*)\right)&= \mathbb{E}_{P_1^*}\!\left(U(E^*)\right)\\
      &=\frac{1}{1-\gamma}\exp\left({(1-\gamma)\displaystyle\inf_{(P_0,P_1)\in\mathcal{P}_0 \times \mathcal{P}_1}R_{1/\gamma}(P_1,P_0))}\right),
  \end{align*}
  where $R_{1/\gamma}(P_1,P_0)=\frac{\gamma}{1-\gamma}\log\mathbb{E}_{P_1}\left[\left(\frac{dP_1}{dP_0}\right)^{\frac{1}{\gamma}-1}\right]$ denotes Rényi divergence of order $1/\gamma$.
\end{theorem}
\begin{proof}
% Since $(P_0^*, P_1^*)$ is an LFD pair for  testing $\mathcal{P}_0$ vs.\ $\mathcal{P}_1$, it is also an LFD pair for  testing $\mathcal{P}_0$ vs.\ ${P}_1^*$, and so it follows from Corollary 1 of \cite{poor1980robust} (setting the continuous convex function $C$ as $C(t)=t^{\frac{1}{\gamma}-1}$) that
% \[
% R_{1/\gamma}(P_1^*,P_0^*)=\inf_{P_0\in\mathcal{P}_0}R_{1/\gamma}(P_1^*,P_0). 
% \]
%The fact that $E^*\in\mathcal{E}(\mathcal{P}_0)$ follows from Theorem 6.1 of \cite{larsson2024numeraire} by setting their $Q$ as $P_1^*$. Another (direct) way of proving this fact is the following: 
Since $(P_0^*, P_1^*)$ is the LFD pair for  $\mathcal{P}_0$ vs.\ $\mathcal{P}_1$, for any $P_0\in\mathcal{P}_0$,
\begin{align*}
  &{R(P_0, \phi) \leq R(P_0^*, \phi)} \\
  &{\implies C_1 {P_0}(\phi \text{ rejects } {H_0}) 
      \leq C_1 {P_1^*}(\phi \text{ rejects } {H_0})}, \text{ for every LRT } \phi \text{ between } P_0^* \text{ and } P_1^*\\
  &{\implies {P_0}\!\left(\tfrac{d P_1^*}{d P_0^*}(X) > \eta^{{\gamma}}\right) 
      \leq {P_0^*}\!\left(\tfrac{d P_1^*}{d P_0^*}(X) >\eta^{{\gamma}}\right), \text{ for all }\eta>0 }\\
      &{\implies   {P_0}\!\left(\left(\frac{dP_{0}^*(X)}{dP_{1}^*(X)}\right)^{-\frac{1}{\gamma}} > \eta\right)\,
      \leq
     {P_0^*}\!\left(\left(\frac{dP_{0}^*(X)}{dP_{1}^*(X)}\right)^{-\frac{1}{\gamma}}> \eta\right)\ } \text{ for all }\eta >0\\
  &{\implies \int_{\eta \geq 0}  {P_0}\!\left(\left(\frac{dP_{0}^*(X)}{dP_{1}^*(X)}\right)^{-\frac{1}{\gamma}} > \eta\right)\, d\eta
      \leq
   \int_{\eta \geq 0}  {P_0^*}\!\left(\left(\frac{dP_{0}^*(X)}{dP_{1}^*(X)}\right)^{-\frac{1}{\gamma}} > \eta\right)\, d\eta}\\
   &\implies\mathbb E_{P_0}\left[\left(\frac{dP_{0}^*(X)}{dP_{1}^*(X)}\right)^{-\frac{1}{\gamma}}\right]\leq \mathbb E_{P_0^*}\left[\left(\frac{dP_{0}^*(X)}{dP_{1}^*(X)}\right)^{-\frac{1}{\gamma}}\right]=E_{P_1^*}\left[\left(\frac{dP_{0}^*(X)}{dP_{1}^*(X)}\right)^{1-\frac{1}{\gamma}}\right].
\end{align*}
Therefore, $\mathbb E_{P_0}[E^*]\leq 1,$ for all $P_0\in\mathcal{P}_0$, proving that $E^*\in\mathcal{E}(\mathcal{P}_0)$.

Now, for the Rényi-optimality part, we first show that $ \mathbb{E}_{P_1}\!\left(U(E^*)\right)
      \geq
    \mathbb{E}_{P_1^*}\!\left(U(E^*)\right)$. Since $(P_0^*, P_1^*)$ is the LFD pair for $\mathcal{P}_0$ vs.\ $\mathcal{P}_1$, we have that for any $P_1\in\mathcal{P}_1$,
\begin{align*}
  &{R(P_1, \phi) \leq R(P_1^*, \phi)} \\
  &{\implies C_1 {P_1}(\phi \text{ accepts } {H_0}) 
      \leq C_1 {P_1^*}(\phi \text{ accepts } {H_0})}, \text{ for every LRT } \phi \text{ between } P_0^* \text{ and } P_1^*\\
  &{\implies {P_1}\!\left(\tfrac{d P_1^*}{d P_0^*}(X) < \eta^{\frac{-\gamma}{\gamma-1}}\right) 
      \leq {P_1^*}\!\left(\tfrac{d P_1^*}{d P_0^*}(X) < \eta^{-\frac{\gamma}{\gamma-1}}\right), \text{ for all }\eta>0 }\\
      &{\implies   {P_1}\!\left(\left(\frac{dP_{0}^*(X)}{dP_{1}^*(X)}\right)^{1-\frac{1}{\gamma}} > \eta\right)\,
      \leq
     {P_1^*}\!\left(\left(\frac{dP_{0}^*(X)}{dP_{1}^*(X)}\right)^{1-\frac{1}{\gamma}}> \eta\right)\ } \text{ for all }\eta >0\\
  &{\implies \int_{\eta \geq 0}  {P_1}\!\left(\left(\frac{dP_{0}^*(X)}{dP_{1}^*(X)}\right)^{1-\frac{1}{\gamma}} > \eta\right)\, d\eta
      \leq
   \int_{\eta \geq 0}  {P_1^*}\!\left(\left(\frac{dP_{0}^*(X)}{dP_{1}^*(X)}\right)^{1-\frac{1}{\gamma}} > \eta\right)\, d\eta}\\
   &\implies\mathbb E_{P_1}\left[\left(\frac{dP_{0}^*(X)}{dP_{1}^*(X)}\right)^{1-\frac{1}{\gamma}}\right]\leq \mathbb E_{P_1^*}\left[\left(\frac{dP_{0}^*(X)}{dP_{1}^*(X)}\right)^{1-\frac{1}{\gamma}}\right].
\end{align*}
Now, since $1-\gamma<0$,
\begin{align*}
     \mathbb{E}_{P_1}\!\left(U(E^*)\right)&=  \mathbb{E}_{P_1}\left[\left(\frac{dP_{0}^*(X)}{dP_{1}^*(X)}\right)^{1-\frac{1}{\gamma}}\right]\Bigg/~\left((1-\gamma)\left(\mathbb E_{P_{1}^*}\left[\left(\frac{dP_{0}^*(X)}{dP_{1}^*(X)}\right)^{1-\frac{1}{\gamma}}\right]\right)^{1-\gamma}
      \right)\\
      &\geq \mathbb{E}_{P_1^*}\left[\left(\frac{dP_{0}^*(X)}{dP_{1}^*(X)}\right)^{1-\frac{1}{\gamma}}\right]\Bigg/~\left((1-\gamma)\left(\mathbb E_{P_{1}^*}\left[\left(\frac{dP_{0}^*(X)}{dP_{1}^*(X)}\right)^{1-\frac{1}{\gamma}}\right]\right)^{1-\gamma}
      \right)\\
      &
    =\mathbb{E}_{P_1^*}\!\left(U(E^*)\right)\\
   & = \frac{\exp\left({(1-\gamma)R_{1/\gamma}(P_1^*,P_0^*))}\right)}{1-\gamma}\\
   %&=\frac{\exp\left({(1-\gamma)\inf_{(P_0,P_1)\in\mathcal{P}_0 \times \mathcal{P}_1}R_{1/\gamma}(P_1,P_0))}\right)}{1-\gamma}.
\end{align*}
So, we have shown above that $\inf_{P_1\in \mathcal{P}_1}\mathbb{E}_{P_1}\!\left(U(E^*)\right)= \mathbb{E}_{P_1^*}\!\left(U(E^*)\right).$

 Now, suppose, if possible, $E$ is an e-variable for $\mathcal{P}_0$, which satisfies 
$ \inf_{P_1\in \mathcal{P}_1}\mathbb{E}_{P_1}\!\left(U (E)\right)>\mathbb{E}_{P_1^*}\!\left(U (E^*)\right).$ Then, we would have $$\mathbb{E}_{P_1^*}\!\left(U(E)\right)> \mathbb{E}_{P_1^*}\!\left(U (E^*)\right)=\frac{\exp\left({(1-\gamma)R_{1/\gamma}(P_1^*,P_0^*))}\right)}{1-\gamma}.$$ However, Theorem 6.1 of \cite{larsson2024numeraire} gives us
$$\sup_{E\in\mathcal{E}(\mathcal{P}_0)}\mathbb{E}_{P_1^*}\!\left(U(E)\right)\leq\frac{\exp\left({(1-\gamma)R_{1/\gamma}(P_1^*,P_0^*))}\right)}{1-\gamma},$$
yielding a contradiction.
Therefore, we must have $$\sup_{E\in \mathcal{E}(\mathcal{P}_0)} \inf_{P_1\in \mathcal{P}_1}\mathbb{E}_{P_1}\!\left(U(E)\right)=\inf_{P_1\in \mathcal{P}_1}\mathbb{E}_{P_1}\!\left(U(E^*)\right)= \mathbb{E}_{P_1^*}\!\left(U(E^*)\right).$$
To complete the proof, we need to show that $$R_{1/\gamma}(P_1^*,P_0^*)=\inf_{(P_0,P_1)\in\mathcal{P}_0 \times \mathcal{P}_1}R_{1/\gamma}(P_1,P_0),$$ which follows from Corollary~1 of \cite{poor1980robust}. 
However, a direct way to prove it is the following: Suppose, for the sake of contradiction, $R_{1/\gamma}(P_1^*,P_0^*)>\inf_{(P_0,P_1)\in\mathcal{P}_0 \times \mathcal{P}_1}R_{1/\gamma}(P_1,P_0).$ 
Then, there exists $(P_0,P_1)\in\mathcal{P}_0 \times \mathcal{P}_1$ such that  $R_{1/\gamma}(P_1^*,P_0^*)>R_{1/\gamma}(P_1,P_0)$, which would imply 
\begin{align*}
    \mathbb{E}_{P_1}\!\left(U(E^*)\right)\geq\inf_{P_1\in \mathcal{P}_1}\mathbb{E}_{P_1}\!\left(U(E^*)\right)=\mathbb{E}_{P_1^*}\!\left(U(E^*)\right)&= \frac{\exp\left({(1-\gamma)R_{1/\gamma}(P_1^*,P_0^*))}\right)}{1-\gamma}\\
    &> \frac{\exp\left({(1-\gamma)R_{1/\gamma}(P_1,P_0))}\right)}{1-\gamma},
\end{align*}
which is a contradiction because for simple null $P_0$ and simple alternative $P_1$, Theorem 6.1 of \cite{larsson2024numeraire} gives us
$$\sup_{E\in\mathcal{E}({P}_0)}\mathbb{E}_{P_1}\!\left(U(E)\right)=\frac{\exp\left({(1-\gamma)R_{1/\gamma}(P_1,P_0))}\right)}{1-\gamma}.$$

Thus, we have finally proved 
 \begin{align*}
      \sup_{E\in \mathcal{E}(\mathcal{P}_0)} \inf_{P_1\in \mathcal{P}_1}\mathbb{E}_{P_1}\!\left(U(E)\right)=\inf_{P_1\in \mathcal{P}_1}\mathbb{E}_{P_1}\!\left(U(E^*)\right)&= \mathbb{E}_{P_1^*}\!\left(U(E^*)\right)\\
      &=\frac{1}{1-\gamma}\exp\left({(1-\gamma)\displaystyle\inf_{(P_0,P_1)\in\mathcal{P}_0 \times \mathcal{P}_1}R_{1/\gamma}(P_1,P_0))}\right).
  \end{align*} 
\end{proof}

\subsection{An adaptive contamination model for the null \texorpdfstring{$H^{\epsilon,\infty}_0$}{Lg} and alternative \texorpdfstring{$H^{\epsilon,\infty}_1$}{Lg}}
\label{subsec:adapt-contamination}

Our test retains its validity properties under the null hypothesis and growth rate under the alternative, even in the presence of an adaptive contamination model, where the data sequence $X_1,X_2,\cdots$ is generated such that the conditional distribution of $X_n$ given $X^{n-1}$ lies within $H^{\epsilon}_i$, where $H^{\epsilon}_i$ is some $\epsilon$ neighborhood around the hypothesized distribution $P_i$, for example, $\epsilon$-contamination or total variation balls defined in \eqref{contamination-model} or \eqref{tv-model}. We denote the set of all possible joint distributions of the sequence $X_1,X_2,\cdots$ under the adaptive contamination model as $H^{\epsilon,\infty}_i$. In the special case where the distribution of each $X_n$ is independent of the previously observed data, then 
\begin{equation}\label{eq:robust-null}
H^{\epsilon,\infty}_i = H^{\epsilon}_i \times H^{\epsilon}_i \times \cdots,
\end{equation}
but we emphasize that $H^{\epsilon,\infty}_i$ allows for non-i.i.d.\ sequentially adaptive adversarial contamination.
In words, while $P_i$ is fixed and known, the unknown contaminations may vary over time, and indeed, the contamination at each time $n$ can be influenced by the previously observed data. Formally, at each $n$, the adversary defines a measurable mapping $Q_n: \Omega^{n-1}\to H^{\epsilon}_i$ and
$X_n$ is drawn from the distribution $Q_n(x^{n-1}),$ given that $X^{n-1}=x^{n-1}$; here we use the notation $a^{n-1} = (a_1,\dots,a_{n-1})$ and $x^{n-1}$ denotes the realized value of the random vector $X^{n-1}$.

\subsection{Robust test for simple null vs. simple alternative}
\label{sec:lfd-contamination-simp}

\cite{huber1965robust} defined the LFD pair $(Q_{0,\epsilon}, Q_{1,\epsilon})$ for \eqref{contamination-model} by their densities $(q_{0,\epsilon}, q_{1,\epsilon})$ as follows:
\begin{align}
\label{eq:q0}
q_{0,\epsilon}(x) & = \begin{cases}
    \left(1-\epsilon\right) p_0(x), &   \text { for } p_1(x) / p_0(x)<c^{\prime \prime} , \\
 \frac{1}{c^{\prime \prime}}\left(1-\epsilon\right) p_1(x),  & \text { for } p_1(x) / p_0(x) \geq c^{\prime \prime}; \\
\end{cases}\\
\label{eq:q1}
q_{1,\epsilon}(x) & = \begin{cases} \left(1-\epsilon\right) p_1(x), &    \text { for } p_1(x) / p_0(x)>c^{\prime} , \\
c^{\prime}\left(1-\epsilon\right) p_0(x),& \text { for } p_1(x) / p_0(x) \leq c^{\prime}.
\end{cases}
\end{align}
The numbers $0 \leq c^{\prime}<c^{\prime \prime} \leq\infty$ are constants that are determined to ensure $q_{0,\epsilon}, q_{1,\epsilon}$ are probability densities:
\begin{align}
\label{eq:c2}
\left(1-\epsilon\right)\left\{P_0\left[p_1 / p_0<c^{\prime \prime}\right]+\frac{1}{c^{\prime \prime}} P_1\left[p_1 / p_0 \geq c^{\prime \prime}\right]\right\} & =1, \\
\label{eq:c1}
\left(1-\epsilon\right)\left\{P_1\left[p_1 / p_0>c^{\prime}\right]+c^{\prime} P_0\left[p_1 / p_0 \leq c^{\prime}\right]\right\} & =1.
\end{align}
We assume that $\epsilon$ is small enough so that $H^{\epsilon}_0\cap H^{\epsilon}_1=\emptyset.$ Then, the above ensures $ c^{\prime}<c^{\prime \prime}$ and 
$\frac{q_{1,\epsilon}(x)}{q_{0,\epsilon}(x)}=\max\{c^\prime,\min\{c^{\prime\prime},q_{1,\epsilon}(x)/q_{0,\epsilon}(x)\}\}$. Huber used  $R_{n,\epsilon}=\prod_{i=1}^n\frac{q_{1,\epsilon}(X_i)}{q_{0,\epsilon}(X_i)}$ for deriving risk optimal tests. Then we have the following as a corollary of Theorem 1 of \cite{huber1965robust} and our  \Cref{thm:lfd-optimal}.

\begin{corollary}
  $\{R_{n,\epsilon}\}_n$ is a test supermartingale for the adaptive null contamination model $H^{\epsilon,\infty}_0$.
Hence, recalling~\eqref{eq:stopping}, the stopping time
\begin{equation}
   \tau^*=\inf\{n : R_{n,\epsilon}\geq 1/\alpha\} 
\end{equation}
  at which we reject the null yields a level-$\alpha$ sequential test.  Moreover, it achieves the optimal growth rate under the adaptive alternative contamination model $H^{\epsilon,\infty}_1$, i.e.,  
  \begin{equation}
      \inf_{\mathbf{P}\in H^{\epsilon,\infty}_1}\mathbb{E}_{\mathbf{P}}\!\left(\log(R_{n,\epsilon})/{n}\right)=\inf_{{P}\in H^{\epsilon}_1}\mathbb{E}_{P}\!\left(\log\frac{q_{1,\epsilon}(X)}{q_{0,\epsilon}(X)}\right)= \kl(Q_{1,\epsilon},Q_{1,\epsilon}).
  \end{equation}
\end{corollary}
As discussed in \cite{huber1973minimax} and \cite[Chapter~10]{huber2004robust}, LFDs can be constructed for a broader class of models, of which \eqref{contamination-model} and \eqref{tv-model} arise as special cases.
 Consequently, the above result also extends to other types of robust models—such as those defined via $\epsilon$-neighborhoods around two distributions $P_0$ and $P_1$ in Kolmogorov, Lévy, Prohorov metrics, etc., when $\frac{p_1(x)}{p_0(x)}$ is monotone in $x$.
%Throughout the paper, we assume that $\epsilon$ is sufficiently small such that we have $c^{\prime \prime}>c^{\prime}$, where $c^{\prime \prime}$ and $c^{\prime}$ are as defined in \eqref{eq:c2} and \eqref{eq:c1} respectively.
\begin{comment}

{ 
\begin{proposition}
\label{prop:eps-bound}
 If $0<\epsilon< \frac{1}{2}\tv(P_0,P_1)$, then  $c^{\prime \prime}>c^{\prime}$, where $c^{\prime \prime}$ and $c^{\prime}$ are as defined in \eqref{eq:c2} and \eqref{eq:c1} respectively. In other words, $H^{\epsilon}_0\cap H^{\epsilon}_1=\emptyset$  ensures that $c^{\prime \prime}>c^{\prime}$ (where $H^{\epsilon}_j, j=0,1$
are as defined in \eqref{tv-model}).
\end{proposition}
Recall from \eqref{eps} that $0<\epsilon< \frac{1}{2}\tv(P_0,P_1)$ can be assumed without loss of generality, since if $\epsilon\geq \frac{1}{2}\tv(P_0,P_1)$, the null $H^{\epsilon}_0$ and the alternative $H^{\epsilon}_1$ would overlap.
}

\end{comment}

\section{Robust tests for composite nulls and alternatives when an LFD pair exists}
\label{lfd-composite}

In this section, we consider both the null $\mathcal P_0$ and the alternative $\mathcal P_1$ to be composite and assume that an LFD pair exists for testing between these two hypotheses.
To handle small deviations from the idealized models,
we test $\mathcal{P}^{\epsilon}_0$ vs.\ $\mathcal{P}^{\epsilon}_1$, where for~$j=0,1$,
\begin{align}
\label{contamination-model-comp}
\mathcal{P}^{\epsilon}_j=\bigcup_{P\in\mathcal{P}_j}B({P},{\epsilon}), ~~&B({P},{\epsilon})=\{Q\in\mathcal{M}: Q=(1-\epsilon)P+\epsilon H, H\in\mathcal{M}\} \text{ or }\\
\label{tv-model-comp}
&B({P},{\epsilon})=\{ Q \in \mathcal{M} : \tv(P, Q) \leq \epsilon \}.
\end{align}

\begin{theorem}
\label{thm:lfd-comp}
  Let $(P_0^*, P_1^*)$ be the least favorable pair in terms of risk defined in \eqref{risk} for testing $\mathcal{P}_0$ vs.\ $\mathcal{P}_1$ and $(Q_{0,\epsilon},Q_{1,\epsilon})$ be the least favorable pair in terms of risk for testing $B({P}_0^*,{\epsilon})$ vs.\ $B({P}_1^*,{\epsilon})$, which are $\epsilon$ neighborhoods around $P_0^*$ and $P_1^*$ respectively, as defined in \eqref{contamination-model-comp} or \eqref{tv-model-comp}. Then, for any fixed sample size $n$, $(Q_{0,\epsilon},Q_{1,\epsilon})$ is the least favorable pair in terms of risk for testing $\mathcal{P}_0^\epsilon$ vs.\ $\mathcal{P}_1^\epsilon$.
\end{theorem}
\begin{proof}

  Since  $(P_0^*, P_1^*)$ is the least favorable pair for testing $\mathcal{P}_0$ vs.\ $\mathcal{P}_1$ we have,
  $R(P_0^*,\phi)\geq R(P,\phi),$ $\forall  P\in \mathcal{P}_0,$ and for all  likelihood ratio test $\phi$ between ${P}_0^*$ and ${P}_1^*$,  which implies

  \begin{equation}
      P_0^*\left[\frac{p_1^*(X)}{p_0^*(X)}\geq \eta\right] \geq P\left[\frac{p_1^*(X)}{p_0^*(X)}\geq \eta\right] ~ \forall \eta \in\mathbb R, \forall  P\in \mathcal{P}_0.
  \end{equation}
% Therefore,
% \begin{equation}
%     Q_{0,\epsilon}\left[\frac{q_{1,\epsilon}}{q_{0,\epsilon}}\geq \eta\right] \geq P\left[\frac{p_1}{p_0}\geq \eta\right] ~ \forall \eta \leq c_2, \forall  P\in \mathcal{P}_0,
% \end{equation}
From \cite{huber1965robust}, we know that $\frac{q_{1,\epsilon}}{q_{0,\epsilon}}=\max\{c_1,\min\{c_2,\frac{p_1^*}{p_0^*}\}\}$, for some $c_1\leq c_2$. So, for $c_1\leq \eta\leq c_2,$ 
    \begin{equation}
        P\left[\frac{q_{1,\epsilon}(X)}{q_{0,\epsilon}(X)}\geq \eta\right]= P\left[\frac{p_1^*(X)}{p_0^*(X)}\geq \eta\right] ~~\text{ and }~~  P_0^*\left[\frac{q_{1,\epsilon}(X)}{q_{0,\epsilon}(X)}\geq \eta\right]= P_0^*\left[\frac{p_1^*(X)}{p_0^*(X)}\geq \eta\right]
    \end{equation}
Hence, for any $c_1\leq \eta\leq c_2,$ 
  \begin{equation}
  \label{eq:p0-p}
    P_{0}^*\left[\frac{q_{1,\epsilon}(X)}{q_{0,\epsilon}(X)}\geq \eta\right] \geq P\left[\frac{q_{1,\epsilon}(X)}{q_{0,\epsilon}(X)}\geq \eta\right] ~  \forall  P\in \mathcal{P}_0.
\end{equation}
Also, note that $ P_{0}^*\left[\frac{q_{1,\epsilon}(X)}{q_{0,\epsilon}(X)}\geq \eta\right] = 0=P\left[\frac{q_{1,\epsilon}(X)}{q_{0,\epsilon}(X)}\geq \eta\right],$ for any  $ \eta> c_2$ and  $ P_{0}^*\left[\frac{q_{1,\epsilon}(X)}{q_{0,\epsilon}(X)}\geq \eta\right] = 1=P\left[\frac{q_{1,\epsilon}(X)}{q_{0,\epsilon}(X)}\geq \eta\right],$ for any $\eta<c_1$. Therefore, \eqref{eq:p0-p} holds for any $\eta\in\mathbb R$.

Now, for any $Q\in \mathcal{P}_{0}^\epsilon$, there exists $P^\prime\in\mathcal{P}_{0}$ such that $Q\in  B(P^\prime,\epsilon)$.

 For the contamination model \eqref{contamination-model-comp}, for any $\eta\in \mathbb R,$
\begin{align*}
    Q_{0,\epsilon}\left[\frac{q_{1,\epsilon}(X)}{q_{0,\epsilon}(X)}\geq \eta\right]&=\max_{P\in B(P_{0}^*,\epsilon)}~ P\left[\frac{q_{1,\epsilon}(X)}{q_{0,\epsilon}(X)}\geq \eta\right]\\
    &=\max_{H}\left[(1-\epsilon)P_{0}^*\left[\frac{q_{1,\epsilon}(X)}{q_{0,\epsilon}(X)}\geq \eta\right]+\epsilon H\left[\frac{q_{1,\epsilon}(X)}{q_{0,\epsilon}(X)}\geq \eta\right]\right]\\
    &=(1-\epsilon)P_{0}^*\left[\frac{q_{1,\epsilon}(X)}{q_{0,\epsilon}(X)}\geq \eta\right]+\epsilon\\
    &\geq (1-\epsilon)P^\prime\left[\frac{q_{1,\epsilon}(X)}{q_{0,\epsilon}(X)}\geq \eta\right]+\epsilon\\
    &=\max_{P\in B(P^\prime,\epsilon)}~ P\left[\frac{q_{1,\epsilon}(X)}{q_{0,\epsilon}(X)}\geq \eta\right]\\
    &\geq Q\left[\frac{q_{1,\epsilon}(X)}{q_{0,\epsilon}(X)}\geq \eta\right] .
\end{align*}
Similarly, for the TV model \eqref{tv-model-comp}, for any $\eta\in \mathbb R,$ 
\begin{align*}
    Q_{0,\epsilon}\left[\frac{q_{1,\epsilon}(X)}{q_{0,\epsilon}(X)}\geq \eta\right]&=\sup_{P\in B(P_{0}^*,\epsilon)}~ P\left[\frac{q_{1,\epsilon}(X)}{q_{0,\epsilon}(X)}\geq \eta\right]\\
    &=P^*_{0}\left[\frac{q_{1,\epsilon}(X)}{q_{0,\epsilon}(X)}\geq \eta\right]+\epsilon\quad (\text{By definition of TV distance})\\
    &\geq P^\prime\left[\frac{q_{1,\epsilon}(X)}{q_{0,\epsilon}(X)}\geq \eta\right]+\epsilon\\
    &=\sup_{P\in B(P^\prime,\epsilon)}~ P\left[\frac{q_{1,\epsilon}(X)}{q_{0,\epsilon}(X)}\geq \eta\right]\\
    &\geq Q\left[\frac{q_{1,\epsilon}(X)}{q_{0,\epsilon}(X)}\geq \eta\right].
\end{align*}

Therefore, for both cases, we have $ Q_{0,\epsilon}\left[\frac{q_{1,\epsilon}(X)}{q_{0,\epsilon}(X)}\geq \eta\right]\geq Q\left[\frac{q_{1,\epsilon}(X)}{q_{0,\epsilon}(X)}\geq \eta\right], \forall\eta\in\mathbb R, \forall Q\in\mathcal{P}_{0}^\epsilon$ and similarly, one can show that for both cases, $ Q_{0,\epsilon}\left[\frac{q_{1,\epsilon}(X)}{q_{0,\epsilon}(X)}>\eta\right]\geq Q\left[\frac{q_{1,\epsilon}(X)}{q_{0,\epsilon}(X)}> \eta\right], \forall\eta\in\mathbb R, \forall Q\in\mathcal{P}_{0}^\epsilon$.

 Suppose, we have i.i.d. samples $X_1,\cdots,X_n$, for any fixed $n$.

 Let $F$ and $G$ be the CDFs of $\frac{q_{1,\epsilon}(X_i)}{q_{0,\epsilon}(X_i)}$, under $Q_{0,\epsilon}$ and $Q$ respectively. From the stochastic monotonicity established above, we have $F\leq G$. Now, let $Z_1,\cdots,Z_n$ are i.i.d from Unif$(0,1)$. Then, $F^{-1}(Z_i)\geq G^{-1}(Z_i)$ and 
 $F^{-1}(Z_i)\overset{\text{d}}{=}\frac{q_{1,\epsilon}(X_i)}{q_{0,\epsilon}(X_i)}$ when $X_i\sim Q_{0,\epsilon}$
 and $G^{-1}(Z_i)\overset{\text{d}}{=}\frac{q_{1,\epsilon}(X_i)}{q_{0,\epsilon}(X_i)}$, when $X_i\sim Q$, $i=1,\cdots,n$. Therefore,
we obtain 
\begin{equation}
    Q_{0,\epsilon}\left[\prod_{i=1}^n\frac{q_{1,\epsilon}(X_i)}{q_{0,\epsilon}(X_i)}>\eta\right]\geq Q\left[\prod_{i=1}^n\frac{q_{1,\epsilon}(X_i)}{q_{0,\epsilon}(X_i)}> \eta\right], \forall\eta\in\mathbb R, \forall Q\in\mathcal{P}_{0}^\epsilon,
\end{equation}
and
\begin{equation}
Q_{0,\epsilon}\left[\prod_{i=1}^n\frac{q_{1,\epsilon}(X_i)}{q_{0,\epsilon}(X_i)}\geq\eta\right]\geq Q\left[\prod_{i=1}^n\frac{q_{1,\epsilon}(X_i)}{q_{0,\epsilon}(X_i)}\geq \eta\right], \forall\eta\in\mathbb R, \forall Q\in\mathcal{P}_{0}^\epsilon
\end{equation}

Note that for any likelhihood ratio $\phi^*$ between $Q_{0,\epsilon}$ and $Q_{1,\epsilon}$, there exists $\eta\in\mathbb R, \gamma\in[0,1]$ such that
\begin{equation}
\label{eq:lrt}
      \phi^*(x_1,\cdots,x_n) =
        \begin{cases}
          1 & \text{if } \prod_{i=1}^n\frac{q_{1,\epsilon}}{q_{0,\epsilon}}(x_i) > \eta, \\
          \gamma & \text{if } \prod_{i=1}^n \frac{q_{1,\epsilon}}{q_{0,\epsilon}}(x_i) = \eta, \\
          0 & \text{if } \prod_{i=1}^n\frac{q_{1,\epsilon}}{q_{0,\epsilon}}(x_i) < \eta.
        \end{cases}
\end{equation}

Then,
\begin{align*}
    R(Q_{0,\epsilon},\phi^*)&=C_0\times\left[Q_{0,\epsilon}\left(\prod_{i=1}^n\frac{q_{1,\epsilon}(X_i)}{q_{0,\epsilon}(X_i)} > \eta\right)+\gamma \times Q_{0,\epsilon}\left(\prod_{i=1}^n\frac{q_{1,\epsilon}(X_i)}{q_{0,\epsilon}(X_i)} =\eta\right)\right]\\
&=C_0\times\left[(1-\gamma)\times Q_{0,\epsilon}\left(\prod_{i=1}^n\frac{q_{1,\epsilon}(X_i)}{q_{0,\epsilon}(X_i)} > \eta\right)+\gamma \times Q_{0,\epsilon}\left(\prod_{i=1}^n\frac{q_{1,\epsilon}(X_i)}{q_{0,\epsilon}(X_i)} \geq\eta\right)\right]\\
&\geq C_0\times\left[(1-\gamma)\times Q\left(\prod_{i=1}^n\frac{q_{1,\epsilon}(X_i)}{q_{0,\epsilon}(X_i)} > \eta\right)+\gamma \times Q\left(\prod_{i=1}^n\frac{q_{1,\epsilon}(X_i)}{q_{0,\epsilon}(X_i)} \geq\eta\right)\right]\\
    &= R(Q,\phi^*), \forall Q\in\mathcal{P}_0^\epsilon
\end{align*}

 The proof for the part $R(Q_{1,\epsilon},\phi^*)\geq R(Q,\phi^*)$ is identical.
\end{proof}

\begin{remark}
   The above theorem holds for sequential probability ratio tests as well We first consider Wald's SPRT, $\phi^*_{\text{SPRT}}$ with stooping rule $\tau=\inf\{n\in\mathbb N: \prod_{i=1}^n\frac{q_{1,\epsilon}(X_i)}{q_{0,\epsilon}(X_i)}\notin(A,B)\}$. Define, $\tau_1=\inf\{n\in\mathbb N: \prod_{i=1}^n\frac{q_{1,\epsilon}(X_i)}{q_{0,\epsilon}(X_i)}\geq B\}$ and $\tau_0=\inf\{n\in\mathbb N: \prod_{i=1}^n\frac{q_{1,\epsilon}(X_i)}{q_{0,\epsilon}(X_i)}\leq A\}$.  Let $F$ and $G$ be the CDFs of $\frac{q_{1,\epsilon}(X_i)}{q_{0,\epsilon}(X_i)}$, under $Q_{0,\epsilon}$ and $Q$ respectively. For $Z_1,\cdots,Z_n\overset{i.i.d.}{\sim}U(0,1)$, define
$\tau^\prime_1=\inf\{n\in\mathbb N: \prod_{i=1}^n F^{-1}(Z_i)\geq B\}$ and $\tau^{\prime\prime}_1=\inf\{n\in\mathbb N: \prod_{i=1}^n G^{-1}(Z_i)\geq B\}$. Also, $\tau^\prime_0=\inf\{n\in\mathbb N: \prod_{i=1}^n F^{-1}(Z_i)\leq A\}$ and $\tau^{\prime\prime}_0=\inf\{n\in\mathbb N: \prod_{i=1}^n G^{-1}(Z_i)\leq A\}$. Then, it follows from the stochastic monotonicity established in the proof of the theorem that $F\leq G$ and hence, $\tau^\prime_0\geq\tau^{\prime\prime}_0$ and $\tau^\prime_1\leq\tau^{\prime\prime}_1$. So,
 \begin{align*}
       R(Q_{0,\epsilon}, \phi^*_{\text{SPRT}}) &= C_0 Q_{0,\epsilon}(\phi^*_{\text{SPRT}} \text{ accepts } H_1)\\
       &=C_0 Q_{0,\epsilon}\left(\tau_1<\tau_0\right)\\
       &=C_0 \mathbb P_{Z_i\overset{i.i.d.}{\sim}U(0,1)}\left(\tau_1^\prime<\tau_0^\prime\right)\\
       &\geq C_0 \mathbb P_{Z_i\overset{i.i.d.}{\sim}U(0,1)}\left(\tau_1^{\prime\prime}<\tau_0^{\prime}\right)\\
       &\geq C_0 \mathbb P_{Z_i\overset{i.i.d.}{\sim}U(0,1)}\left(\tau_1^{\prime\prime}<\tau_0^{\prime\prime}\right)\\
       &=C_0 Q(\tau_1<\tau_0)\\
       &= R(Q, \phi^*_{\text{SPRT}}),
   \end{align*}
for all  $Q\in\mathcal{P}_{0}^\epsilon$.
Simliarly, one can show $R(Q_{1,\epsilon}, \phi^*_{\text{SPRT}})\geq R(Q, \phi^*_{\text{SPRT}})$, for all $Q\in\mathcal{P}_{1}^\epsilon$.

Similarly, for the one-sided SPRT, $\phi^*_{\text{S}}$ with $\tau=\inf\{n\in\mathbb N: \prod_{i=1}^n\frac{q_{1,\epsilon}(X_i)}{q_{0,\epsilon}(X_i)}\geq B\}$, we show the same:
\begin{equation*}
      R(Q_{0,\epsilon}, \phi^*_{\text{S}})=C_0 Q_{0,\epsilon}\left(\tau<\infty\right)=C_0 \mathbb P_{Z_i\overset{iid.}{\sim}U(0,1)}\left(\tau^\prime_1<\infty\right)\geq C_0 \mathbb P_{Z_i\overset{iid.}{\sim}U(0,1)}\left(\tau^{\prime\prime}_1<\infty\right)=R(Q, \phi^*_{\text{S}}),
\end{equation*}
for all $Q\in\mathcal{P}_{0}^\epsilon$.
\end{remark}

\begin{remark}
\label{rem:lfd-comp}
    Note that for the above theorem to hold, it is not necessary for the neighborhoods $B(P,\epsilon)$ to be $\epsilon$-contamination or total variation (TV) balls. The only properties of $B(\cdot,\epsilon)$ required in the proof are the following: an LFD pair exists for testing $B(\cdot,\epsilon)$ vs. $B(\cdot,\epsilon)$ as long as they do not intersect each other and
   for any $Q, Q^\prime\in \mathcal{P}_0\cup\mathcal{P}_1,\eta\in\mathbb R$ satisfying  $$Q\left[\frac{q_{1,\epsilon}(X)}{q_{0,\epsilon}(X)}\geq \eta\right]\geq  Q^\prime\left[\frac{q_{1,\epsilon}(X)}{q_{0,\epsilon}(X)}\geq \eta\right]\implies\displaystyle\sup_{P\in B(Q,\epsilon)}P\left[\frac{q_{1,\epsilon}(X)}{q_{0,\epsilon}(X)}\geq \eta\right]\geq  \sup_{P\in B(Q^\prime,\epsilon)}P\left[\frac{q_{1,\epsilon}(X)}{q_{0,\epsilon}(X)}\geq \eta\right].$$
   Clearly, both $\epsilon$-contamination and TV neighborhoods satisfy this condition, as shown in the proof. However, it is worth noting that other robust models, eg.,  $\epsilon$-balls based on Kolmogorov or Lévy distances, also satisfy it, under the assumption that $\frac{p_1(x)}{p_0(x)}$ is monotone in $x$, for all $P_0\in\mathcal{P}_{0}, P_1\in\mathcal{P}_{1}$.
\end{remark}

\subsection{Adaptive contamination models}

Recalling the adaptive contamination model for the null from \cref{subsec:adapt-contamination}, we define the adaptive contamination model for the composite hypothesis, $\mathcal H^{\epsilon,\infty}_i$ to be the set of all possible joint distribution of the sequence $X_1,X_2,\cdots$ such that the conditional distribution of $X_n$ given $X^{n-1}$ lies within $\mathcal H^{\epsilon}_i$. So, at each $n$, the adversary defines a measurable mapping $Q_n: \Omega^{n-1}\to \mathcal{H}^{\epsilon}_i$ and
$X_n$ is drawn from the distribution $Q_n(x^{n-1}),$ given that $X^{n-1}=x^{n-1}$.

Our test retains its validity properties under the null hypothesis and growth rate under the alternative, even in the presence of an adaptive contamination model, where $H^{\epsilon}_i$ is some $\epsilon$ neighborhood around the hypothesized distribution $P_i$, for example, $\epsilon$-contamination or total variation balls defined in \eqref{contamination-model-comp} or \eqref{tv-model-comp}. 

\subsection{Robust test for composite null vs. composite alternative}
\label{sec:lfd-contamination-comp}
Let $(P_0, P_1)$ be the least favorable pair for testing $\mathcal{P}_0$ vs.\ $\mathcal{P}_1$ and $(Q_{0,\epsilon},Q_{1,\epsilon})$ be the least favorable pair in terms of risk for testing $B({P}_0,{\epsilon})$ vs.\ $B({P}_1,{\epsilon})$, defined in \eqref{contamination-model-comp}.
 Then, with $R_{n,\epsilon}=\prod_{i=1}^n\frac{q_{1,\epsilon}(X_i)}{q_{0,\epsilon}(X_i)}$, we have the following as a corollary of Theorem 1 of \cite{huber1965robust} and our  \Cref{thm:lfd-optimal,thm:lfd-comp}.

\begin{corollary}
  $\{R_{n,\epsilon}\}_n$ is a test supermartingale for the adaptive null contamination model $H^{\epsilon,\infty}_0$.
Hence, recalling~\eqref{eq:stopping}, the stopping time
\begin{equation}
   \tau^*=\inf\{n : R_{n,\epsilon}\geq 1/\alpha\} 
\end{equation}
  at which we reject the null yields a level-$\alpha$ sequential test.  Moreover it achieve the optimal growth rate, i.e.,  $$\inf_{\mathbf{P}\in H^{\epsilon,\infty}_1}\mathbb{E}_{\mathbf{P}}\!\left(\log(R_{n,\epsilon})/{n}\right)=\inf_{{P}\in H^{\epsilon}_1}\mathbb{E}_{P}\!\left(\log\frac{q_{1,\epsilon}(X)}{q_{0,\epsilon}(X)}\right)= \kl(Q_{1,\epsilon},Q_{1,\epsilon}).$$
\end{corollary}
As discussed in \Cref{rem:lfd-comp}, similar results can also be established for other robust models, e.g., $\epsilon$-balls based on total variation, Kolmogorov
or Levy distance.
%Throughout the paper, we assume that $\epsilon$ is sufficiently small such that we have $c^{\prime \prime}>c^{\prime}$, where $c^{\prime \prime}$ and $c^{\prime}$ are as defined in \eqref{eq:c2} and \eqref{eq:c1} respectively.

\section{Robust tests for composite nulls and alternatives when an LFD pair does not exist}
\label{sec:comp-general}
Now, we consider the problem of robust hypothesis testing, where an LFD pair does not
exist for the pre-contamination composite null and composite alternative.

\subsection{Robust predictable plug-in for composite alternatives}
\label{sec:comp-alt}

In this section, we address the challenge of robustly testing the composite alternatives. Let the idealized models be $P_0$ vs.\ $\mathcal{P}_1$, where $\mathcal{P}_1$ is a set of distribution functions that does not include $P_0$. We will assume that there exists a common reference measure $\mu$ for $\mathcal P_1$ and $P_0$ so that we can associate the distributions with their densities.

To handle a composite alternative hypothesis, one natural way is to attempt
to learn it from past observations, at each round $n$, and plug it in as an estimate of the alternative distribution based on past observations. This is known as the ``predictable plug-in'' or simply ``plug-in'' method. This method is originally due to \cite{wald1947sequential}, which has recently been used for handling various parametric and non-parametric composite alternatives \cite{waudby2020confidence,waudby2023estimating,saha2024testing}.
 
 However, we often encounter deviations from these idealized models in real-world data, potentially due to contamination or deviations from the assumed distribution.  To navigate this problem, we use a robust estimate of the alternative distribution, which lies in $\mathcal{P}_1$, where the actual data might come from $\epsilon$-neighbourhood of some distribution in the alternative $\mathcal{P}_1$, which is 
\begin{equation}
\label{contamination-model-comp-alt}
    \mathcal{P}^{\epsilon}_1=\bigcup_{P_1\in\mathcal{P}_1}\{Q\in\mathcal{M}: Q=(1-\epsilon)P_1+\epsilon H, H\in\mathcal{M}\} \text{ or }
\end{equation}
\begin{equation}
\label{tv-model-comp-alt}
    \mathcal{P}^{\epsilon}_1=\bigcup_{P_1\in\mathcal{P}_1}\{ Q \in \mathcal{M} : \tv(P_1, Q) \leq \epsilon \}.
\end{equation} 

\subsubsection{An adaptive contamination model for the null}

Under the null, we assume that the data $X_1,X_2,\dots$ is generated from some distribution in $H^{\epsilon,\infty}_0$ defined in \cref{subsec:adapt-contamination}, i.e., while $P_0$ is fixed and known, the unknown contaminations may vary across time. In fact, the contamination at time $n$ may depend on the previously observed data, as noted earlier.

Let $\hat{p}_{1,n}$ be some robust estimate of the density of the data based on past observations $X^{n-1}$, which belongs to the alternative. Let $\hat{P}_{1,n}$ be the distribution corresponding to the density $\hat{p}_{1,n}$, with $\hat{P}_{1,n}\in\mathcal{P}_1$. 

Now we define $\hat{q}_{n,j,\epsilon}$ similarly as ${q}_{j,\epsilon}$  was defined in \eqref{eq:q1} (for $j=0,1$) .
\begin{align}
\label{eq:q0-comp-alt}
\hat q_{n,0,\epsilon}(x) & = \begin{cases}
    \left(1-\epsilon\right) p_0(x), \quad \text { for } \hat{p}_{1,n}(x) / p_0(x)<c^{\prime \prime}_n , \\
\frac{1}{c^{\prime \prime}_n}\left(1-\epsilon\right) \hat{p}_{1,n}(x),  \text { for } \hat{p}_{1,n}(x) / p_0(x) \geq c^{\prime \prime}_n; \\
\end{cases}\\
\label{eq:q1-comp-alt}
\hat q_{n,1,\epsilon}(x) & = \begin{cases} \left(1-\epsilon\right) \hat{p}_{1,n}(x), \quad \text { for } \hat{p}_{1,n}(x) / p_0(x)>c^{\prime}_n , \\
c^{\prime}_n\left(1-\epsilon\right) p_0(x), \text { for } \hat{p}_{1,n}(x) / p_0(x) \leq c^{\prime}_n.
\end{cases}
\end{align}
We calculate the numbers $0 \leq c^{\prime}_n,c^{\prime \prime}_n \leq\infty$, that are determined such that $\hat q_{n,0,\epsilon},\hat q_{n,1,\epsilon}$ are probability densities:
\begin{align}
\label{eq:c2-comp-alt}
\left(1-\epsilon\right)\left\{P_0\left[\hat{p}_{1,n} / p_0<c^{\prime \prime}_n\right]+\frac{1}{c^{\prime \prime}_n} \hat{P}_{1,n}\left[\hat{p}_{1,n} / p_0 \geq c^{\prime \prime}_n\right]\right\} & =1, \\
\label{eq:c1-comp-alt}
\left(1-\epsilon\right)\left\{\hat{P}_{1,n}\left[\hat{p}_{1,n}/ p_0>c^{\prime}_n\right]+c^{\prime}_n P_0\left[\hat{p}_{1,n} / p_0 \leq c^{\prime}_n\right]\right\} & =1.
\end{align}
Note that if 
$c^{\prime}_n<c^{\prime \prime}_n$, we have ${\frac{\hat{q}_{n,1,\epsilon}(x)}{\hat q_{n,0,\epsilon}(x)}}=\max\{c^{\prime}_n,\min\{c^{\prime\prime}_n,\frac{\hat{p}_{1,n}(x)}{p_0(x)}\}\}$.  We now define $R_{0,\epsilon}^{\text{plug-in}}=1$ and $R_{n,\epsilon}^{\text{plug-in}}= R_{n-1,\epsilon}^{\text{plug-in}}\times \hat E_{n,\epsilon}(X_n), n=1,2,\dots$
where 

\begin{equation}
  \hat E_{n,\epsilon}(\cdot):=\begin{cases}
       \frac{\frac{\hat{q}_{n,1,\epsilon}(\cdot)}{\hat q_{n,0,\epsilon}(\cdot)}}{\mathbb E_{P_0}\left[\frac{\hat{q}_{n,1,\epsilon}(X)}{\hat q_{n,0,\epsilon}(X)}\mid X^{n-1}= x^{n-1}\right]+(c^{\prime \prime}_n-c^{\prime}_n)\epsilon} \text{, if } c^{\prime}_n<c^{\prime \prime}_n,\\
       1, \quad \quad \quad \quad \quad \quad \quad \quad \quad \quad \quad \quad \quad \text{ otherwise. }
   \end{cases} 
\end{equation}
In other words, $R_{n,\epsilon}^{\text{plug-in}}$ is the product of $ \hat E_{n,\epsilon}$. Thus, if each $\hat E_{n,\epsilon}$ is an e-variable (conditioned on the past), $R_{n,\epsilon}^{\text{plug-in}}$ will be a test supermartingale. The following short calculation shows that this is indeed the case:
\begin{align*}
   &\mathbb E_{Q_n}[\hat E_{n,\epsilon}(X_n) \mid X^{n-1}= x^{n-1}]\\
  % &=\mathbb E\left[ \mathbbm{I}_{c^{\prime}_n\geq c^{\prime \prime}_n}+\frac{\frac{\hat{q}_{n,1,\epsilon}(X_n)}{\hat q_{n,0,\epsilon}(X_n)}\mathbbm{I}_{c^{\prime}_n<c^{\prime \prime}_n}}{\mathbb E_{X\sim P_0}\left[\frac{\hat{q}_{n,1,\epsilon}(X)}{\hat q_{n,0,\epsilon}(X)}\mid X^{n-1}\right]+(c^{\prime \prime}_n-c^{\prime}_n)\epsilon} \mid \mathcal{F}_{n-1} \right]\\
  &=\mathbbm{I}_{c^{\prime}_n\geq c^{\prime \prime}_n}+\frac{\mathbb E_{Q_n} \left[\frac{\hat{q}_{n,1,\epsilon}(X_n)}{\hat q_{n,0,\epsilon}(X_n)}\mid X^{n-1}= x^{n-1}\right]}{\mathbb E_{P_0} \left[ \frac{\hat{q}_{n,1,\epsilon}(X)}{\hat q_{n,0,\epsilon}(X)} \mid X^{n-1}= x^{n-1} \right] + (c^{\prime \prime}_n-c^{\prime}_n)\epsilon}\mathbbm{I}_{c^{\prime}_n<c^{\prime \prime}_n}\\
  &\leq\mathbbm{I}_{c^{\prime}_n\geq c^{\prime \prime}_n}+\mathbbm{I}_{c^{\prime}_n<c^{\prime \prime}_n}=1.
\end{align*}

The last inequality follows from the fact that the total variation distance is an \emph{integral probability metric} in the sense that for any pair of real numbers $c_1 < c_2$,
\begin{equation}\label{eqn:tv-integral}
   \tv(P, Q) = \frac{1}{c_2 - c_1}\sup_{c_1 \le f \le c_2} \left| \Expw_{X\sim P} f(X) - \Expw_{X \sim Q} f(X) \right|,
\end{equation}
since the conditional distribution $Q_n=Q_n( x^{n-1})$ of $X_n$ conditioned on $X^{n-1}= x^{n-1}$ satisfies $Q_n\in H^{\epsilon}_0, \text{ i.e., } \tv(Q_n,P_0) \leq \epsilon$. 
We immediately conclude the following.
\begin{theorem}
$R_{n,\epsilon}^{\text{plug-in}}$ is a test supermartingale for the adaptive null contamination model $H^{\epsilon,\infty}_0$.
Hence, recalling~\eqref{eq:stopping}, the stopping time
\begin{equation}
   \tau^*=\inf\{n : R_{n,\epsilon}^{\text{plug-in}}\geq 1/\alpha\} 
\end{equation}
  at which we reject the null yields a level-$\alpha$ sequential test.
 \end{theorem}

\subsubsection{Growth rate analysis}
For analyzing the growth rate of this test supermartingale, we assume that the estimator $\hat{p}_{1,n}$ converges pointwise in the following sense:
{ 
\begin{assumption}
\label{assmp-1}
    If $X_1,X_2,\cdots\stackrel{iid}{\sim} H$, for some distribution $H\in\mathcal{P}^{\epsilon}_1$, then $\hat{p}_{1,n}\to p_1^H$ almost surely as $n\to \infty$, where $p_1^H$ is a density function (which may depend on $H$) corresponding to some distribution $P_1^H\in\mathcal{P}_1$.
\end{assumption}
}
{  Note that the above pointwise convergence assumption is much weaker than assuming the existence of a pointwise \emph{consistent} estimator since $P_1^H$ is any arbitrary distribution in $\mathcal P_1$. Pointwise convergence holds in a wide variety of estimation problems within the robust statistics literature. For instance, robust M-estimators for certain parametric testing problems are known to be strongly consistent under regularity conditions \cite[Chapter 6]{huber2004robust}. To illustrate with a concrete example, when estimating the mean of a Gaussian distribution under contamination, the sample median serves as a robust estimator that converges almost surely to the population median.}

 The next theorem shows that $\log R_{n,\epsilon}^{\text{plug-in}}/n$ converges almost surely, under some assumptions.
%as the Oracle (\Cref{thm:gr-lb}), i.e., when we perform the robust test with simple null $P_0$ vs.\ simple alternative $P_1^H$ is known as described in the previous section. 
%Thus, in this scenario, the robust plug-in approach performs asymptotically as well as the oracle method, where the alternative $P_1$ is known, demonstrating its effectiveness in practical settings.

\begin{theorem}
\label{thm:gr-lb-comp-alt}
    Suppose that $X_1,X_2,\cdots\stackrel{iid}{\sim} H$, for some distribution $H\in\mathcal{P}^{\epsilon}_1$ (where $\mathcal{P}^{\epsilon}_1$ is as defined in \eqref{contamination-model-comp-alt} or \eqref{tv-model-comp-alt}) and the estimator $\hat{p}_{1,n}$ satisfies \Cref{assmp-1}. We further assume that $P(p^H_1/p_0 = c)=0$, for all $c\in \mathbb R$ and for $P\in P_0\cup\mathcal{P}_1$. We consider $\epsilon$ being sufficiently small so that $c^{\prime \prime}_H>c^{\prime}_H$, where $c^{\prime \prime}_H,c^{\prime}_H$ are solutions of \eqref{eq:c2-comp-null} and \eqref{eq:c1-comp-null}, respectively, with $P_1$ being replaced by $P_1^{H}$. Then, as $n\to \infty$,
 \begin{equation*}
 \frac{1}{n}\log R_{n,\epsilon}^{\text{plug-in}}\to r_{H,\epsilon}^{\text{plug-in}} \text{ almost surely,} 
\end{equation*}
where 
%$r_Q^\epsilon$ is the same constant as in \Cref{thm:gr-lb}, and hence we have the same bound 
$r_{H,\epsilon}^{\text{plug-in}} \geq \kl(Q^H_{1,\epsilon},Q^H_{0,\epsilon})-2(\log c^{\prime \prime}_H-\log c^{\prime}_H)\epsilon-\log (1+2(c^{\prime \prime}_H-c^{\prime}_H)\epsilon)$  and
 $Q_{1,\epsilon}^H,Q_{0,\epsilon}^H$ are the distributions with densities $q_{1,\epsilon}^H,q_{0,\epsilon}^H$ as defined in \eqref{eq:q0} and \eqref{eq:q1} with $P_1$ being replaced by $P_1^{H}$.
\end{theorem}

\begin{remark}
    In \Cref{thm:gr-lb-comp-alt}, if we further assume that
    the maximum bias $b_{P_1}(\epsilon,x):=\sup_{H:\tv(H,P_1)\leq \epsilon}|p_1^H(x)-p_1(x)|$ is a real-valued function such that $\lim_{\epsilon\to 0}b_{P_1}(\epsilon,x)=0$, for all $x\in\mathbb R$, we still could not conclude in general that  
    $\inf_{H:\tv(H,P_{1})\leq \epsilon}r_{H,\epsilon}^{\text{plug-in}}$ converges to the oracle (when $P_1$ is known) growth rate $\kl(P_1,P_0)$, because of the difficulty in interchanging the expectation and the limit $\epsilon\to 0$. However, under the aforementioned vanishing bias assumption, in specific situations, such as  exponential family distributions, one can show that
 \begin{equation*}
    \lim_{\epsilon\to 0} \inf_{H:\tv(H,P_{1})\leq \epsilon}r_{H,\epsilon}^{\text{plug-in}}= \kl(P_1,P_0).
 \end{equation*}
\end{remark}

%Note that the assumption that $\hat{p}_{1,n}\to p_1$ almost surely is not a strong assumption. Such estimators are known to exist for a wide variety of testing problems. For example, one may use some robust M-estimators for certain parametric testing problems, which are strongly consistent \cite[Chapter 6]{huber2004robust} under some regularity conditions. In \cref{sec:expt}, for the Gaussian location problem, we have used the sample median as a robust and strongly consistent estimator for the location parameter. 
%Furthermore, it is natural to assume that $c^{\prime \prime}_n\to c^{\prime \prime}$ and $c^{\prime}_n\to c^{\prime}$, since $\hat{p}_{1,n}\to p_1$ implies that the equations \eqref{eq:c2-comp-alt} and \eqref{eq:c1-comp-alt} determining the values of $c^{\prime \prime}$ and $c^{\prime}$ converges to the equations \eqref{eq:c2} and \eqref{eq:c1} respectively, which determine the values of $c^{\prime \prime}$ and $ c^{\prime}$.

\subsection{{Robust numeraire: composite null vs.\ simple alternative}}
\label{sec:comp-null}

In this section, we consider the null $\mathcal P_0$ to be composite and the alternative to be simple $P_1$. To account for the small deviations from the idealized models,
consider the following:
\begin{equation}
\label{contamination-model-comp-null}
\mathcal{P}^{\epsilon}_0=\bigcup_{P\in\mathcal{P}_0}\{Q\in\mathcal{M}: Q=(1-\epsilon)P+\epsilon H, H\in\mathcal{M}\}, \quad 
\end{equation}
\begin{equation}
\label{tv-model-comp-null}
 \text{ or }   \mathcal{P}^{\epsilon}_0=\bigcup_{P\in\mathcal{P}_0}\{ Q \in \mathcal{M} : \tv(P, Q) \leq \epsilon \},
\end{equation}
and ${H}^{\epsilon}_1$ is either \eqref{contamination-model} or \eqref{tv-model}. Recalling the adaptive contamination model for the null from \cref{subsec:adapt-contamination}, we define the adaptive contamination model for the composite null, $\mathcal H^{\epsilon,\infty}_0$ to be the set of all possible joint distribution of the sequence $X_1,X_2,\cdots$ such that the conditional distribution of $X_n$ given $X^{n-1}$ lies within $\mathcal H^{\epsilon}_0$. So, at each $n$, the adversary defines a measurable mapping $Q_n: \Omega^{n-1}\to \mathcal{P}^{\epsilon}_0$ and
$X_n$ is drawn from the distribution $Q_n(x^{n-1}),$ given that $X^{n-1}=x^{n-1}$. 

\subsubsection{The reverse information projection (RIPr) and numeraire}

For testing composite nulls versus a simple alternative, the reverse information projection (RIPr) has recently emerged as an central tool \cite{grunwald2020safe,lardy2023universal,larsson2024numeraire}. In his PhD thesis, Li~\cite{li1999mix} defined the RIPr as follows: if we assume that $\mathcal{P}_{0}$ is convex and has a reference measure, and $\inf_{P\in\mathcal{P}_0}\kl(P_1,P)<\infty$,  then there exists a unique sub-probability measure $P_0$ that satisfies
\begin{equation}
\label{ripr-klinf}
\kl(P_1,P_0)=\inf_{P\in\mathcal{P}_0}\kl(P_1,P).
\end{equation}
This $P_0$ is defined as the RIPr of $P_1$ on $\mathcal{P}_0$.
%where $P_0$ exists uniquely under the above assumptions. 
%In particular, the infimum is often achieved by some $P_1\in\mathcal{P}_0$ and in that case $P_1=\arg\min_{P\in\mathcal{P}_0}\kl(P_1,P)$. 
%Define, $B^*(x):=p_0(x)/p_1(x)$. Then, $B^*(X)$ is an ``e-variable" for $\mathcal{P}_0$, which means that $B^*(X)$ is a nonnegative random variable satisfying $\mathbb E_P(B^*(X))\leq 1$, for all $P\in\mathcal{P}_0$. Additionally, it is even the growth rate optimal (GRO) e-variable \cite{grunwald2020safe}, meaning that it maximizes $\mathbb E_{P_1} \log B$ over all e-variables $B$ for $\mathcal{P}_0$. An equivalent statement is that for any e-variables $B$ for $\mathcal{P}_0$, $\mathbb E_{P_1}[B/B^*]\leq 1$ (``numeraire property'' \cite{larsson2024numeraire}).

Recently, \cite{larsson2024numeraire} significantly generalized this definition by exploiting a certain duality of the RIPr with e-variables. They first show that for any null $\mathcal{P}_0$ and simple alternative $P_1$,
there always exists a ($P_1$ almost surely) unique and strictly positive e-variable $B^*$ called the \emph{numeraire}, such that for any other e-variable $B$ for $\mathcal{P}_0$, $\mathbb E_{P_1}[B/B^*]\leq 1$ (``numeraire'' property, closely related to Cover's optimal portfolios). An equivalent statement is that for any e-variable $B$ for 
$\mathcal{P}_0$, $\mathbb E_{P_1} \log (B^*/B) \geq 0$ (``log-optimality''; see also \cite{grunwald2020safe,lardy2023universal}). This implies that the growth rate of $B^*$, $\mathbb E_{P_1} \log (B^*)$, is larger than that of any other e-variable $B$. It is somewhat remarkable that the existence and uniqueness of such an optimal e-variable can be established under absolutely no assumptions on $\mathcal{P}_0$ and $P_1$. \cite{larsson2024numeraire} then define a measure $P_0$ by defining its likelihood ratio
(Radon-Nikodym derivative) with respect to $P_1$, $\frac{dP_0}{dP_1}:=\frac{1}{B^*}.$ This $P_0$ is the more general definition of the Reverse Information Projection (RIPr) of $P_1$ onto $\mathcal{P}_0$, and it can be shown to match Li's  original definition under his additional assumptions (of convexity, reference measure, and finiteness of $\kl$). Further, a very general strong duality holds between $B^*$ and $P_0$ that we omit here for brevity. 

In what follows, the RIPr $P_0$ is the sole representative of the composite null $\mathcal{P}_0$, and even though it may be a sub-probability measure in general, we can still proceed as if we were dealing with a simple null. We elaborate below.

\subsubsection{Robustifying the numeraire}

Let ${P}_{0}$ be the RIPr of ${P}_{1}$ on the null $\mathcal{P}_0$. Suppose, for $j=1,2$, $p_j$ be the density of $P_j$ with respect to some common dominating measure $\mu$. Let, $k=\int p_0d\mu$. Note that since ${P}_{0}$ is a sub-probability measure, we have $k\leq 1$.
%We assume that $\hat{P}_{1,i}\in \mathcal{P}_1$. Then, $\hat{P}_{0,i}=\arg\min_{P\in\mathcal{P}_0}\kl(\hat{P}_{1,i},P)$.
We obtain ${q}_{j,\epsilon}$ as defined in \eqref{eq:q0} and \eqref{eq:q1}.
We calculate the numbers $0 \leq c^{\prime},c^{\prime \prime} \leq\infty$ are determined such that $q_{0,\epsilon}$ is a sub-probability density with $\int q_{0,\epsilon}d\mu=k$ and $ q_{1,\epsilon}$ is a probability density:
\begin{align}
\label{eq:c2-comp-null}
\left(1-\epsilon\right)\left\{P_0\left[p_1 / p_0<c^{\prime \prime}\right]+\frac{1}{c^{\prime \prime}} P_1\left[p_1 / p_0 \geq c^{\prime \prime}\right]\right\} & =k, \\
\label{eq:c1-comp-null}
\left(1-\epsilon\right)\left\{P_1\left[p_1 / p_0>c^{\prime}\right]+c^{\prime} P_0\left[p_1 / p_0 \leq c^{\prime}\right]\right\} & =1.
\end{align}
We choose $\epsilon$ sufficiently small, then
$c^{\prime}<c^{\prime \prime}$ and we have ${\frac{{q}_{1,\epsilon}(x)}{ q_{0,\epsilon}(x)}}=\max\{c^{\prime},\min\{c^{\prime\prime},\frac{{p}_{1}(x)}{ p_{0}(x)}\}\}$.  Define $R_{0,\epsilon}^{\text{RIPr}}=1$ and $R_{n,\epsilon}^{\text{RIPr}}= R_{n-1,\epsilon}^{\text{RIPr}}\times B_{\epsilon}(X_n), n=1,2,\dots$
where 
\begin{equation}
\label{bet:comp-null}
   B_{\epsilon}(x):=\frac{\frac{{q}_{1,\epsilon}(x)}{q_{0,\epsilon}(x)}}{\sup_{P\in \mathcal{P}_0}\mathbb E_{X\sim P}\left[\frac{{q}_{1,\epsilon}(X)}{q_{0,\epsilon}(X)}\right]+(c^{\prime \prime}-c^{\prime})\epsilon}.
\end{equation}
The term $\sup_{P\in \mathcal{P}_0}\mathbb E_{X\sim P}\left[\frac{{q}_{1,\epsilon}(X)}{q_{0,\epsilon}(X)}\right]$ might not have a closed form expression, in that case, we need to rely on numerical approximations in practice. In parametric setting, one can discretize the parameter space in $\mathcal{P}_{0}$ by selecting a grid of finitely many representative points, and at each grid point, we approximate the expectation using Monte Carlo sampling, where we generate independent samples from the corresponding distribution and compute the empirical mean of the likelihood ratio. The final approximation is then obtained by taking the maximum over all the values at the grid points. However, in many simple settings, we have that the RIPr satisfies the maximum, i.e.,
\begin{equation}
    \label{eq:sup-ripr}
    \sup_{P\in \mathcal{P}_0}\mathbb E_{X\sim P}\left[\frac{{q}_{1,\epsilon}(X)}{q_{0,\epsilon}(X)}\right]=\mathbb E_{X\sim P_0}\left[\frac{{q}_{1,\epsilon}(X)}{q_{0,\epsilon}(X)}\right].
\end{equation}
When this condition holds, computing the test supermartingale becomes significantly simpler. Next, we provide a sufficient condition under which this equality is guaranteed.

\begin{proposition}
\label{prop:sup-ripr-exp}
Suppose that $\{P_\theta:\theta\in\Theta\}$ is some exponential family distribution with one-dimensional sufficient statistic for the parameter $\theta$. Let ${P}_1=P_{\theta_1}$ and $\mathcal{P}_0=\{P_\theta:\theta\in\Theta_0\}$ are such that $\sup\Theta_0\leq \theta_1$ (or $\inf\Theta_0\geq \theta_1$). Then, the RIPr, $P_0$, is $P_{\theta_0^*}$, where $\theta_0^*=\sup\Theta_0$ (or $\theta_0^*=\inf\Theta_0$), and \eqref{eq:sup-ripr} holds.
\end{proposition}
The above result also holds for a general monotone likelihood ratio family, eg., see \cite{ramdas2024hypothesis}
[Chapter 2.6.1].

Now, we observe that, under our adaptive null contamination model, for any possible conditional distribution of $X_n$ conditioned on $X^{n-1}=x^{n-1}$ (denoted as $Q_n=Q_n(x^{n-1})$), there exists $P_{0,n}\in \mathcal{P}_0$, such that $\tv(Q_n,P_{0,n})\leq\epsilon$. Then, 
\begin{align*}
   &\mathbb E_{ Q_n} [ B_{\epsilon}(X_n) \mid X^{n-1}=x^{n-1}]\\
   &=\mathbb E_{ Q_n} \left[ \frac{\frac{{q}_{1,\epsilon}(X_n)}{ q_{0,\epsilon}(X_n)}}{\sup_{P\in \mathcal{P}_0}\mathbb E_{X\sim P}\left[\frac{{q}_{1,\epsilon}(X)}{ q_{0,\epsilon}(X)}\right]+(c^{\prime \prime}-c^{\prime})\epsilon} \mid X^{n-1}=x^{n-1} \right]\\
  &\leq \frac{\mathbb E_{Q_n} \left[\frac{{q}_{1,\epsilon}(X_n)}{ q_{0,\epsilon}(X_n)}\mid X^{n-1}=x^{n-1}\right]}{\mathbb E_{ P_{0,n}} \left[ \frac{{q}_{1,\epsilon}(X_n)}{ q_{0,\epsilon}(X_n)} \right] + (c^{\prime \prime}-c^{\prime})\epsilon}\\
  &\leq 1.
\end{align*}
The last inequality follows from \eqref{eqn:tv-integral}, since the conditional distribution $Q_n=Q_n(x^{n-1})$ satisfies $\tv(Q_n,P_{0,n})\leq\epsilon$. 
We immediately conclude the following.
\begin{theorem}
Suppose that $\epsilon>0$ is sufficiently small such that we have $c^{\prime \prime}>c^{\prime}$. Then, $ R_{n,\epsilon}^{\text{RIPr}}$ is a test supermartingale for the adaptive null contamination model $\mathcal{H}^{\epsilon,\infty}_0$.
Hence, recalling~\eqref{eq:stopping}, the stopping time
\begin{equation}
   \tau^*=\inf\{n :  R_{n,\epsilon}^{\text{RIPr}}\geq 1/\alpha\} 
\end{equation}
  at which we reject $\mathcal{P}^{\epsilon}_0$ yields a level-$\alpha$ sequential test.
\end{theorem}

It turns out that similar results can be established for the composite null case, akin to those observed in the simple null versus simple alternative scenario.

\begin{comment}

The next result shows that we recover the ``growth rate optimal'' or ``numeraire'' e-variable $B^*$ when $\epsilon$ approaches zero, under standard assumptions. 

\begin{proposition}
\label{lem:limit-comp-null}
Assume that $\inf_{P\in\mathcal{P}_0}\kl(P_1,P)<\infty$. Then, $B_{\epsilon}(x)\to B^*(x)$ $\mu$-almost surely as $\epsilon\to0$. In other words, for any fixed $n\in \mathbb N$, $R_{n,\epsilon}^{\text{RIPr}}\to\prod_{i=1}^n p_1(X_i)/p_0(X_i)$ $\mu$-almost surely as $\epsilon\to0.$
\end{proposition}

The above proposition can be seen as an extension of \Cref{prop:test-conv} for composite null. The assumption $\inf_{P\in\mathcal{P}_0}\kl(P_1,P)<\infty$ is analogous to the assumption $\kl(P_1,P_0)<\infty$ in \Cref{prop:test-conv}.
%However, this proposition introduces an additional assumption that $\mathcal{P}_{0}$ is closed and convex. We are uncertain if this assumption can be relaxed, and further study is needed to determine its necessity.
\end{comment}
%For Growth rate analysis, we assume that $X_1,X_2,\cdots \stackrel{iid}{\sim} Q\in H^{\epsilon}_1$. We now analyze the growth rate of our test supermartingale. The next result provides a lower bound on the growth rate.
\begin{theorem}
\label{thm:gr-lb-comp-null}
Suppose that $\epsilon>0$ be sufficiently small such that we have $c^{\prime \prime}>c^{\prime}$ and $X_1,X_2,\cdots \stackrel{iid}{\sim} Q\in H^{\epsilon}_1$. Then, as $n\to \infty$,
 \begin{equation*}
 \frac{\log R_{n,\epsilon}^{\text{RIPr}}}{n}\to r^{Q,\epsilon}_{\text{RIPr}} \text{ almost surely for some constant } r^{Q,\epsilon}_{\text{RIPr}}
\end{equation*}
and the growth rate, 
\begin{align*}
    r^\epsilon_{\text{RIPr}}=\inf_{ Q \in H^{\epsilon}_1}r^{Q,\epsilon}_{\text{RIPr}}\geq \kl(Q_{1,\epsilon},Q_{0,\epsilon})-2(\log c^{\prime \prime}-\log c^{\prime})\epsilon
    -\log\left(\sup_{P\in \mathcal{P}_0}\mathbb E_{P}\frac{{q}_{1,\epsilon}(X)}{q_{0,\epsilon}(X)}+(c^{\prime \prime}-c^{\prime})\epsilon\right).
\end{align*}
\end{theorem}

The next theorem shows that the growth rate of our test converges to $\inf_{P\in\mathcal{P}_0}\kl(P_1,P)$, as $\epsilon\to 0$. And it is the optimal e-power or growth rate for testing $\mathcal{P}_0$ vs.\ $P_1$ \cite{grunwald2020safe,lardy2023universal}.

\begin{theorem}
\label{thm:asymp-gr-comp-null}
If \eqref{eq:sup-ripr} holds, then the growth rate of our test, $r^{\epsilon}_{\text{RIPr}}\to \inf_{P\in\mathcal{P}_0}\kl(P_1,P)$, as $\epsilon\to 0$.
\end{theorem}

This result shows that the growth rate of our robust test supermartingale for composite null is asymptotically optimal as $\epsilon$ approaches zero.

\subsection{Combining robust predictable plug-in and robust numeraire:  composite null vs.\ composite alternative}
\label{sec:comp-null-alt}

In this section, we address the most general scenario, when both the null $\mathcal P_0$ and the alternative $\mathcal P_1$ are composite. We will assume that there exists a common reference measure for $\mathcal P_1$ and $\mathcal P_0$ so that we can associate the distributions with their densities.
To handle small deviations from the idealized models,
we test $\mathcal{P}^{\epsilon}_0$ vs.\ $\mathcal{P}^{\epsilon}_1$, defined in \eqref{contamination-model-comp} or \eqref{tv-model-comp}.

Our approach combines plug-in and RIPr methods from the last two sections. At each time $n$, let $\hat{p}_{1,n}$ be some robust estimate of the density of the data based on past observations $X^{n-1}$ which belongs to the alternative. Let $\hat{P}_{1,n}$ be the distribution corresponding to the density $\hat{p}_{1,n}$ and $\hat{P}_{1,n}\in\mathcal{P}_1$. Let $\hat{P}_{0,n}$ be the reverse information projection (RIPr) of $\hat{P}_{1,n}$ on the null $\mathcal{P}_0$. 
%We assume that $\hat{P}_{1,i}\in \mathcal{P}_1$. Then, $\hat{P}_{0,i}=\arg\min_{P\in\mathcal{P}_0}\kl(\hat{P}_{1,i},P)$.

%We obtain $\hat{q}_{i,j,\epsilon}$ similarly as ${q}_{j,\epsilon}$  was defined in \eqref{eq:q1} and with ${p}_{j}$ at time $i$ replaced by $\hat{p}_{j,i}$, for $i=1,2,\cdots$ and for $j=0,1$. We calculate $c^{\prime \prime}_i$ and $c^{\prime}_i$ from \eqref{eq:c1} and \eqref{eq:c2} with ${p}_{j}$ and ${P}_{j}$ replaced by $\hat{p}_{j,i}$ and $\hat{P}_{j,i}$ respectively.

Now we define $\hat{q}_{n,j,\epsilon}$ similarly as ${q}_{j,\epsilon}$  was defined in \eqref{eq:q1} (for $j=0,1$) .
\begin{align}
\label{eq:q0-comp-alt-null}
\hat q_{n,0,\epsilon}(x) & = \begin{cases}
    \left(1-\epsilon\right) \hat p_{0,n}(x), \quad \quad \text { for } \hat{p}_{1,n}(x) / \hat p_{0,n}(x)<c^{\prime \prime}_n , \\
\frac{1}{c^{\prime \prime}_n}\left(1-\epsilon\right) \hat{p}_{1,n}(x), ~~\text { for } \hat{p}_{1,n}(x) /\hat p_{0,n}(x) \geq c^{\prime \prime}_n; \\
\end{cases}\\
\label{eq:q1-comp-alt-null}
\hat q_{n,1,\epsilon}(x) & = \begin{cases} \left(1-\epsilon\right) \hat{p}_{1,n}(x), \quad \quad \text { for } \hat{p}_{1,n}(x) / \hat p_{0,n}(x)>c^{\prime}_n , \\
c^{\prime}_n\left(1-\epsilon\right) \hat p_{0,n}(x), \quad \text { for } \hat{p}_{1,n}(x) / \hat p_{0,n}(x) \leq c^{\prime}_n.
\end{cases}
\end{align}
$k_n=\int \hat p_{0,n}d\mu$. Since ${P}_{0}$ is a sub-probability measure, we have $k_n\leq 1$.
%We assume that $\hat{P}_{1,i}\in \mathcal{P}_1$. Then, $\hat{P}_{0,i}=\arg\min_{P\in\mathcal{P}_0}\kl(\hat{P}_{1,i},P)$.
We calculate the numbers $0 \leq c^{\prime}_n,c^{\prime \prime}_n \leq\infty$ are determined such that $\hat q_{n,0,\epsilon}$ is a sub-probability density with $\int \hat q_{n,0,\epsilon}d\mu=k_n$ and $\hat q_{n,1,\epsilon}$ is a probability density:
\begin{align}
\label{eq:c2-comp-alt-null}
\hat{P}_{0,n}\left[\hat{p}_{1,n} / \hat p_{0,n}<c^{\prime \prime}_n\right]+ \frac{1}{c^{\prime \prime}_n}\hat{P}_{1,n}\left[\hat{p}_{1,n} / \hat p_{0,n}\geq c^{\prime \prime}_n\right] & =\frac{k_n}{\left(1-\epsilon\right)}, \\
\label{eq:c1-comp-alt-null}
\hat{P}_{1,n}\left[\hat{p}_{1,n}/ \hat p_{0,n}>c^{\prime}_n\right]+c^{\prime}_n \hat P_{0,n}\left[\hat{p}_{1,n} / \hat p_{0,n} \leq c^{\prime}_n\right] & =\frac{1}{\left(1-\epsilon\right)}.
\end{align}

Note that if 
$c^{\prime}_n<c^{\prime \prime}_n$, we have ${\frac{\hat{q}_{n,1,\epsilon}(x)}{\hat q_{n,0,\epsilon}(x)}}=\max\{c^{\prime}_n,\min\{c^{\prime\prime}_n,\frac{\hat{p}_{1,n}(x)}{\hat p_{0,n}(x)}\}\}$.  Define $R_{0,\epsilon}^{\text{RIPr,plug-in}}=1$ and $R_{n,\epsilon}^{\text{RIPr,plug-in}}= R_{n-1,\epsilon}^{\text{RIPr,plug-in}}\times \hat B_{n,\epsilon}(X_n), n=1,2,\dots$
where 
\begin{equation}
\label{bet-comp-null-comp-alt}
  \hat B_{n,\epsilon}(x):=\begin{cases}
       \frac{\frac{\hat{q}_{n,1,\epsilon}(x)}{\hat q_{n,0,\epsilon}(x)}}{\sup_{P\in \mathcal{P}_0}\mathbb E_{ P}\left[\frac{\hat{q}_{n,1,\epsilon}(X)}{\hat q_{n,0,\epsilon}(X)}\mid X^{n-1}=x^{n-1}\right]+(c^{\prime \prime}_n-c^{\prime}_n)\epsilon}, \text{ if } c^{\prime}_n<c^{\prime \prime}_n,\\
       1, \text{ otherwise. }
   \end{cases} 
\end{equation}
As noted in the previous section, $\sup_{P\in \mathcal{P}_0}\mathbb E_{X\sim P}\left[\frac{\hat{q}_{n,1,\epsilon}(X)}{\hat q_{n,0,\epsilon}(X)}\mid X^{n-1}\right]$ might not have a closed form expression; in that case, we need to rely on numerical approximations.
We now show that this construction yields a valid test for the adaptive contamination model $\mathcal H^{\epsilon,\infty}_0$ defined in the previous section. Under $\mathcal H^{\epsilon,\infty}_0$,
 for any possible conditional distribution $Q_n=Q_n(x^{n-1})$ of $X_n$ conditioned on $X^{n-1}=x^{n-1}$, there exists $P_{0,n}\in \mathcal{P}_0$, such that $\tv(Q_n,P_{0,n})\leq\epsilon$. Then, 
\begin{align*}
   &\mathbb E_{Q_n} [\hat B_{n,\epsilon}(X_n) \mid X^{n-1}=x^{n-1}]\\
  % &=\mathbb E\left[ \mathbbm{I}_{c^{\prime}_n\geq c^{\prime \prime}_n}+\frac{\frac{\hat{q}_{n,1,\epsilon}(X_n)}{\hat q_{n,0,\epsilon}(X_n)}\mathbbm{I}_{c^{\prime}_n<c^{\prime \prime}_n}}{\displaystyle\sup_{P\in \mathcal{P}_0}\mathbb E_{X\sim P}\left[\frac{\hat{q}_{n,1,\epsilon}(X)}{\hat q_{n,0,\epsilon}(X)}\mid X^{n-1}\right]+(c^{\prime \prime}_n-c^{\prime}_n)\epsilon} \mid \mathcal{F}_{n-1} \right]\\
  &\leq\mathbbm{I}_{c^{\prime}_n\geq c^{\prime \prime}_n}+\frac{\mathbb E_{ Q_n} \left[\frac{\hat{q}_{n,1,\epsilon}(X_n)}{\hat q_{n,0,\epsilon}(X_n)}\mid X^{n-1}=x^{n-1}\right]\mathbbm{I}_{c^{\prime}_n<c^{\prime \prime}_n}}{\mathbb E_{P_{0,n}} \left[ \frac{\hat{q}_{n,1,\epsilon}(X_n)}{\hat q_{n,0,\epsilon}(X_n)} \mid X^{n-1} =x^{n-1}\right] + (c^{\prime \prime}_n-c^{\prime}_n)\epsilon}\\
  &\leq\mathbbm{I}_{c^{\prime}_n\geq c^{\prime \prime}_n}+\mathbbm{I}_{c^{\prime}_n<c^{\prime \prime}_n}=1.
\end{align*}

The last inequality follows from \eqref{eqn:tv-integral}, since the conditional distribution $Q_n$ satisfies $\tv(Q_n,P_{0,n})\leq\epsilon$.  
We immediately conclude the following.
\begin{theorem}
$R_{n,\epsilon}^{\text{RIPr,plug-in}}$ is a test supermartingale for the adaptive null contamination model $\mathcal H^{\epsilon,\infty}_0$.
Hence, recalling~\eqref{eq:stopping}, if
\begin{equation}
   \tau^*=\inf\{n :  R_{n,\epsilon}^{\text{RIPr,plug-in}}\geq 1/\alpha\} 
\end{equation}
 is the stopping time at which we reject $\mathcal{P}^{\epsilon}_0$, then this yields a level-$\alpha$ sequential test.
\end{theorem}
{ We summarize our method in \Cref{algo:1}. The computational complexity of this procedure primarily depends on lines 4 and 5, which in turn are determined by the structures of the pre-contaminated null and alternative classes, $\mathcal{P}_0$ and $\mathcal{P}_1$. In simple settings, these steps are computationally straightforward. For instance, consider the robust Gaussian mean testing problem where $\mathcal{P}_0={N(\mu,1):\mu\geq \mu_0}$ and $\mathcal{P}_1={N(\mu,1):\mu\leq \mu_1}$, for some $\mu_1<\mu_0$. In this case, Step 4 involves computing the median of $n$ data points, which takes $O(n)$ time, while Step 5 incurs only $O(1)$ cost because at each $n$, we know that $\hat{P}_{0,n}$ is $N(\mu_0,1)$ in this setting. However, for more general model classes, these steps can become substantially more involved.}
{
\begin{algorithm}[h!]
\caption{Level-$\alpha$ robust test for general composite null vs. composite alternative}
\label{algo:1}
\begin{algorithmic}[1]
\State $R_{0,\epsilon}^{\text{RIPr,plug-in}}= 1$
\State $n= 1$
\While{$R_{n,\epsilon}^{\text{RIPr,plug-in}}<\frac{1}{\alpha}$}
\State Compute $\hat{p}_{1,n}$: a robust density estimate based on past observations $x^{n-1}$\
\State Compute $\hat{P}_{0,n}$: reverse information projection of $\hat{P}_{1,n}$ on the null $\mathcal{P}_0$\
    \State Compute $c^{\prime}_n$ and $c^{\prime \prime}_n $\ by solving \eqref{eq:c2-comp-alt-null} and \eqref{eq:c1-comp-alt-null}
    \State $R_{n,\epsilon}^{\text{RIPr,plug-in}}=R_{n-1,\epsilon}^{\text{RIPr,plug-in}}$
    \If{$c^{\prime}_n<c^{\prime \prime}_n$}
    \State $\pi_{n,\epsilon}(x):=\max\{c^{\prime}_n,\min\{c^{\prime\prime}_n,\frac{\hat{p}_{1,n}(x)}{\hat p_{0,n}(x)}\}\}$\
    \State Compute $ B_{n,\epsilon}=\pi_{n,\epsilon}(x_n)/[{\sup_{P\in \mathcal{P}_0}\mathbb E_{ X\sim P}\left[\pi_{n,\epsilon}(X)\mid X^{n-1}=x^{n-1}\right]+(c^{\prime \prime}_n-c^{\prime}_n)\epsilon}]$\
    \State $R_{n,\epsilon}^{\text{RIPr,plug-in}}=R_{n,\epsilon}^{\text{RIPr,plug-in}}\times B_{n,\epsilon}$
    \EndIf
    \State $n=n+1$
    \EndWhile
\end{algorithmic}
\end{algorithm}
}

{ 
While the type-I error control holds without any additional assumptions, establishing the growth rate of the test requires additional assumptions.
 The main challenge lies in the fact that even if we assume $\hat{p}_{1,n}$ to be a strongly consistent estimator, it does not guarantee that the corresponding sequence of reverse information projections, $\hat{p}_{0,n}$, would converge in the almost sure sense. So, we need to explicitly assume the convergence of $\hat{p}_{0,n}$, as formalized below.

 \begin{assumption}
 \label{assmp-ripr-conv}
   If $\hat{p}_{1,n}$ satisfies \Cref{assmp-1}, then $\hat{p}_{0,n}\to{p}_{0}^H$ almost surely as $n\to\infty$, where ${p}_{0}^H$ is the density of the measure ${P}_{0}^H$, which is the RIPr of ${P}_{1}^H$ on $\mathcal{P}_0$.
 \end{assumption}

To provide a concrete example where the assumption does hold, consider a one-parameter exponential family of densities written in canonical form 
\begin{equation}
\label{eq:exp-fam}
    p_{\theta}(x)=h(x)\exp(\theta T(x)-A(\theta)),
\end{equation}
where $A: \mathbb R\to \mathbb R$ is a convex, differentiable function and
\begin{equation}
\label{exp-test}
    \mathcal{P}_0=\{P_{\theta}: \theta\in[a,b] \} \text{ and }\mathcal{P}_1=\{P_{\theta}: \theta\in \Theta_1\}, 
\end{equation}
$\text{for some } -\infty \leq a\leq b \leq \infty, \Theta_1\subseteq\Theta \text{ is such that } \Theta_1\cap [a,b]=\emptyset.$

\begin{proposition}
\label{prop:one-param-exp-fam}
  For a one-parameter exponential family of the form \eqref{eq:exp-fam}, the testing problem in \eqref{exp-test} satisfies \Cref{assmp-ripr-conv}. 
\end{proposition}

Note that \Cref{assmp-ripr-conv} is not restricted to univariate problems. Consider the multivariate Gaussian testing setup: $\mathcal{P}_0=\{N_d(\bm{\mu},I_d):\mu_1\leq a,\mu_2=\cdots=\mu_d=0\}$ and $\mathcal{P}_1=\{N_d(\bm{\mu},I_d):\mu_1 \geq b\}$, for some $a < b$.
The KL divergence between two such distributions is given by $\kl(N_d(\bm{\mu},I_d),N_d(\bm{\mu}^\prime,I_d))=\frac{1}{2}\|\bm{\mu}^\prime-\bm{\mu}\|_2^2$. It is straightforward to verify that for  $\bm{\mu}^*=(a,0,\cdots,0)$, $$\kl(P,N_d(\bm{\mu}^*,I_d))=\inf_{P_0\in\mathcal{P}_0}\kl(P,P_0), \text{ for any } P\in\mathcal{P}_1.$$ 
Hence, $\{\hat{p}_{0,n}\}$ is a constant sequence (i.e., for every $n$, the RIPr density $\hat{p}_{0,n}$ is the density of $N_d(\bm{\mu}^*,I_d)$) and thus  \Cref{assmp-ripr-conv} holds true.

%It is left as future work to characterize where \Cref{assmp-ripr-conv} holds.
}

\begin{theorem}
\label{thm:gr-exp-family}
Suppose that $X_1,X_2,\cdots\stackrel{iid}{\sim} H$, for some distribution $H\in\mathcal{P}^{\epsilon}_1$ (where $\mathcal{P}^{\epsilon}_1$ is as defined in \eqref{contamination-model-comp} or \eqref{tv-model-comp}) and \Cref{assmp-1,assmp-ripr-conv} hold. 
We consider $\epsilon$ being sufficiently small so that $c^{\prime \prime}_H>c^{\prime}_H$, where $c^{\prime \prime}_H,c^{\prime}_H$ are solutions of \eqref{eq:c2-comp-null} and \eqref{eq:c1-comp-null}, respectively with $P_1$ and $P_0$ being replaced by $P_{\theta_1(H)}$ and the RIPr of $P_{\theta_1(H)}$ on $\mathcal{P}_0$, respectively. Then,
 \begin{equation*}
 \frac{1}{n}\log R_{n,\epsilon}^{\text{RIPr,plug-in}}\to r^{H,\epsilon}_{\text{RIPr,plug-in}} \text{ almost surely as }n\to \infty,
\end{equation*}
where 
 $r^{H,\epsilon}_{\text{RIPr,plug-in}}\geq \kl(Q_{1,\epsilon}^H,Q_{0,\epsilon}^H)-2(\log c^{\prime \prime}_H-\log c^{\prime}_H)\epsilon-\log\left(\sup_{P\in \mathcal{P}_0}\mathbb E_{P}\frac{{q}^H_{1,\epsilon}(X)}{q^H_{0,\epsilon}(X)}+(c^{\prime \prime}_H-c^{\prime}_H)\epsilon\right),$ and
 $Q_{1,\epsilon}^H,Q_{0,\epsilon}^H$ are distributions with densities $q_{1,\epsilon}^H,q_{0,\epsilon}^H$ as defined in \eqref{eq:q0} and \eqref{eq:q1} with $P_1$ and $P_0$ being replaced by $P_{\theta_1(H)}$ and the RIPr of $P_{\theta_1(H)}$ on $\mathcal{P}_0$, respectively.
 %If we further assume that for $\theta_1\in\Theta_1$, the bias $b_{\theta_1}(\epsilon):=\sup_{H:\tv(H,P_{\theta_1})\leq \epsilon}|\theta(H)-\theta_1|$ is a real valued function such that $\lim_{\epsilon\to 0}b_{\theta_1}(\epsilon)=0$, then we have
 %\begin{equation*}
 %   \lim_{\epsilon\to 0} \inf_{H:\tv(H,P_{\theta_1})\leq \epsilon}r^{H,\epsilon}_{\text{RIPr,plug-in}}= r^*,
 %\end{equation*}
 %where $r^*=\inf_{P\in\mathcal{P}_0}\kl(P_{\theta_1},P)$ is the optimal growth rate for testing $\mathcal{P}_0$ against $P_{\theta_1}$.
 %Moreover, $r^{Q,\epsilon}_{\text{RIPr}}\to \kl(P_1,P_0)$, as $\epsilon\to 0.$
\end{theorem}
The above theorem provides a bound on the growth rate and proves asymptotic optimality as $\epsilon\to 0$, under suitable assumptions. 
%For example, \cite[Chapter 4]{huber2004robust} shows that for estimating the location parameter in a contaminated Gaussian model, the sample median $\hat\theta_n$ is minimum bias estimator and $b_{\theta_1}(\epsilon)=b_{0}(\epsilon)=\Phi^{-1}(\frac{1}{2(1-\epsilon)})$, where $\Phi$ is the Gaussian CDF. Hence $b_{\theta_1}(\epsilon)\to\Phi^{-1}(1/2)=0$, as $\epsilon\to 0$, and so it meets all the conditions of the theorem.

\section{Simulations}
\label{sec:expt}
In this section, we present a series of simulations designed to evaluate the performance of our robust tests for both simple and composite hypotheses. We use two key parameters in our analysis: $\epsilon^A$, which represents the value of $\epsilon$
specified to the test supermartingale and $\epsilon^R$, which denotes the true fraction of data contaminated (A = Algorithm, R = Reality). 

\subsection{Experiments with simple null}

In all the simulation experiments in this subsection, we consider the null $P_0$ to be $N(0,1)$, the simple and the composite alternative to be $P_1=N(\mu_1,1)$ for some fixed $\mu_1$ and $\mathcal{P}_1=\{N(\mu,1): \mu\neq 0\}$ respectively. For simple null vs simple alternative, we employ the test described in \Cref{sec:lfd-contamination-simp}. But for $P_0$ vs. the composite alternative $\mathcal{P}_1$, there is no LFD; hence, we use the test supermartingale introduced in \Cref{sec:comp-alt}. Both the non-robust predictable plug-in method and our robustified predictable plug-in method for composite alternative, we use the sample median as an estimate of $\mu$. All the results in Fig. \cref{fig:null,fig:different-eps} are the average of $10$ independent simulations.

\paragraph{Sanity check under the null.}
In this experiment, samples are simulated independently from the following $\epsilon_R$-contaminated null distribution: $Q=(1-\epsilon^R)\times N(0,1) + \epsilon^R \times\text{Cauchy}(-1,10)$. Here $\epsilon^A=\epsilon^R=0.01$. This mixture model ensures that the $\epsilon^R$ fraction of the sample is drawn from the heavy-tailed Cauchy distribution with location and scale parameters $-1$ and $10$ respectively. For the tests with simple alternative, we consider $\mu_1=1$.  Fig. \ref{fig:null} illustrates that our robust tests are ``safe'', i.e. they do not exhibit growth under the null hypothesis, whereas the non-robust methods show unreliable behaviour with significant fluctuations.

\begin{figure}[!htb]
    \centering
    \includegraphics[width=0.5\linewidth]{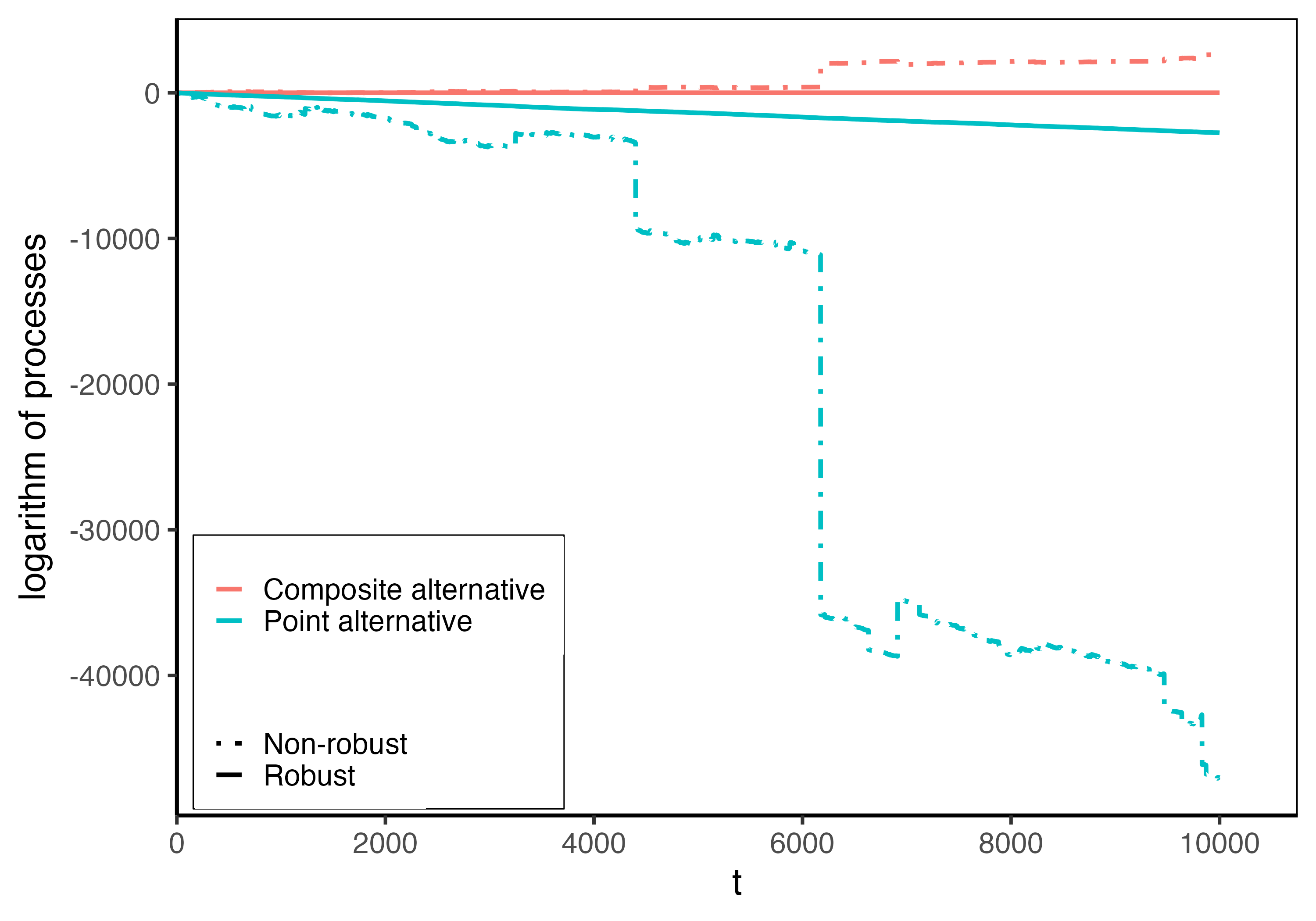}
    \caption{Data is drawn from $(1-\epsilon^R)\times N(0,1) + \epsilon^R \times\text{Cauchy}(-1,10)$ and $P_0=N(0,1), P_1=N(1,1)$, $\epsilon^A=\epsilon^R=0.01$. Robust tests are safe, but the non-robust tests exhibit unstable and unreliable behavior.}
    \label{fig:null}
\end{figure}

\begin{figure*}[!htb]

\centering
\centering
\subfloat[$P_0=N(0,1)$, $P_1=N(1,1)$]{\includegraphics[width=0.49\linewidth,height=0.3\linewidth]{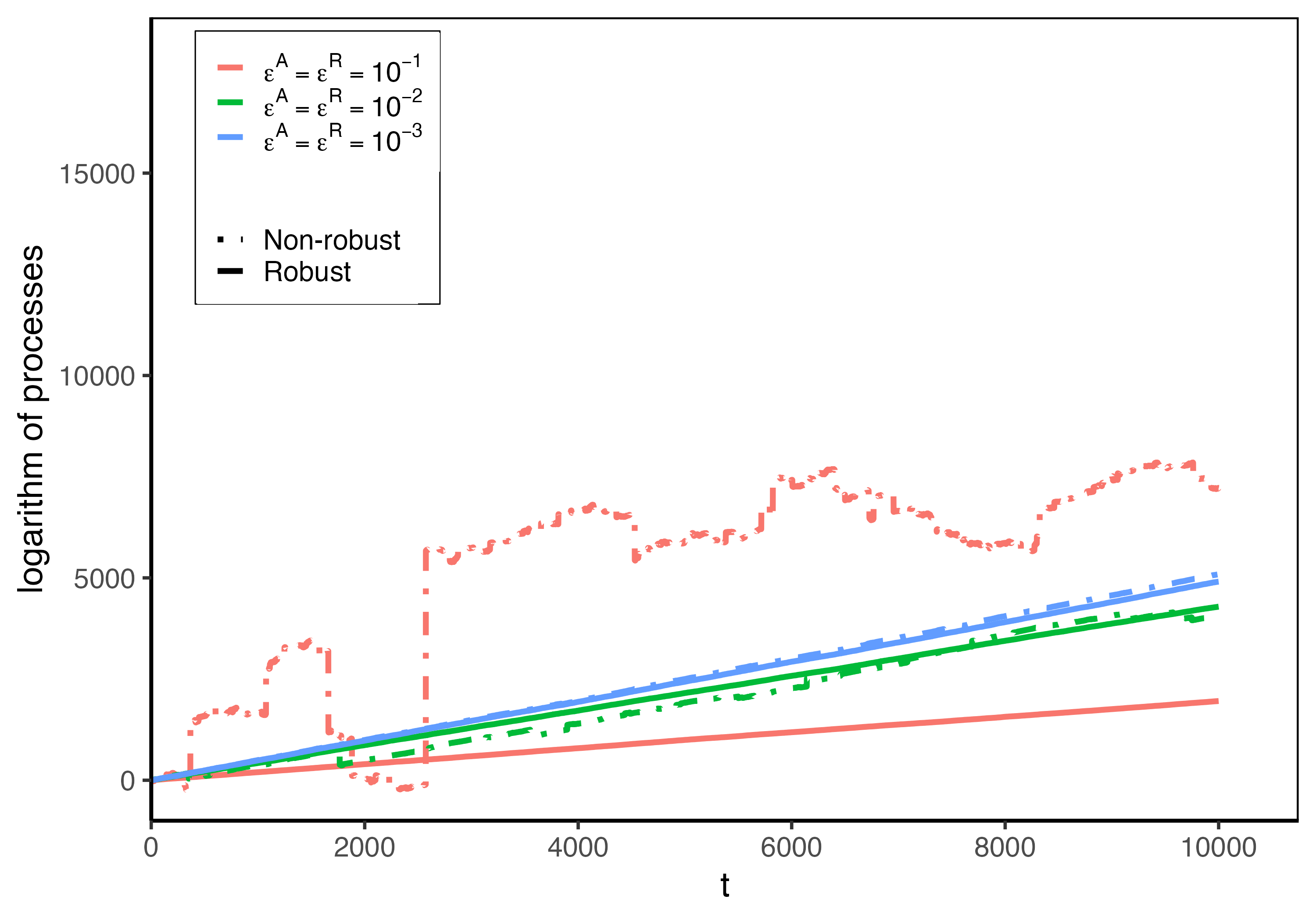}} 
\subfloat[$P_0=N(0,1), \mathcal{P}_1=\{N(\mu,1): \mu\neq 0\}$]{\includegraphics[width=0.49\linewidth,height=0.3\linewidth]{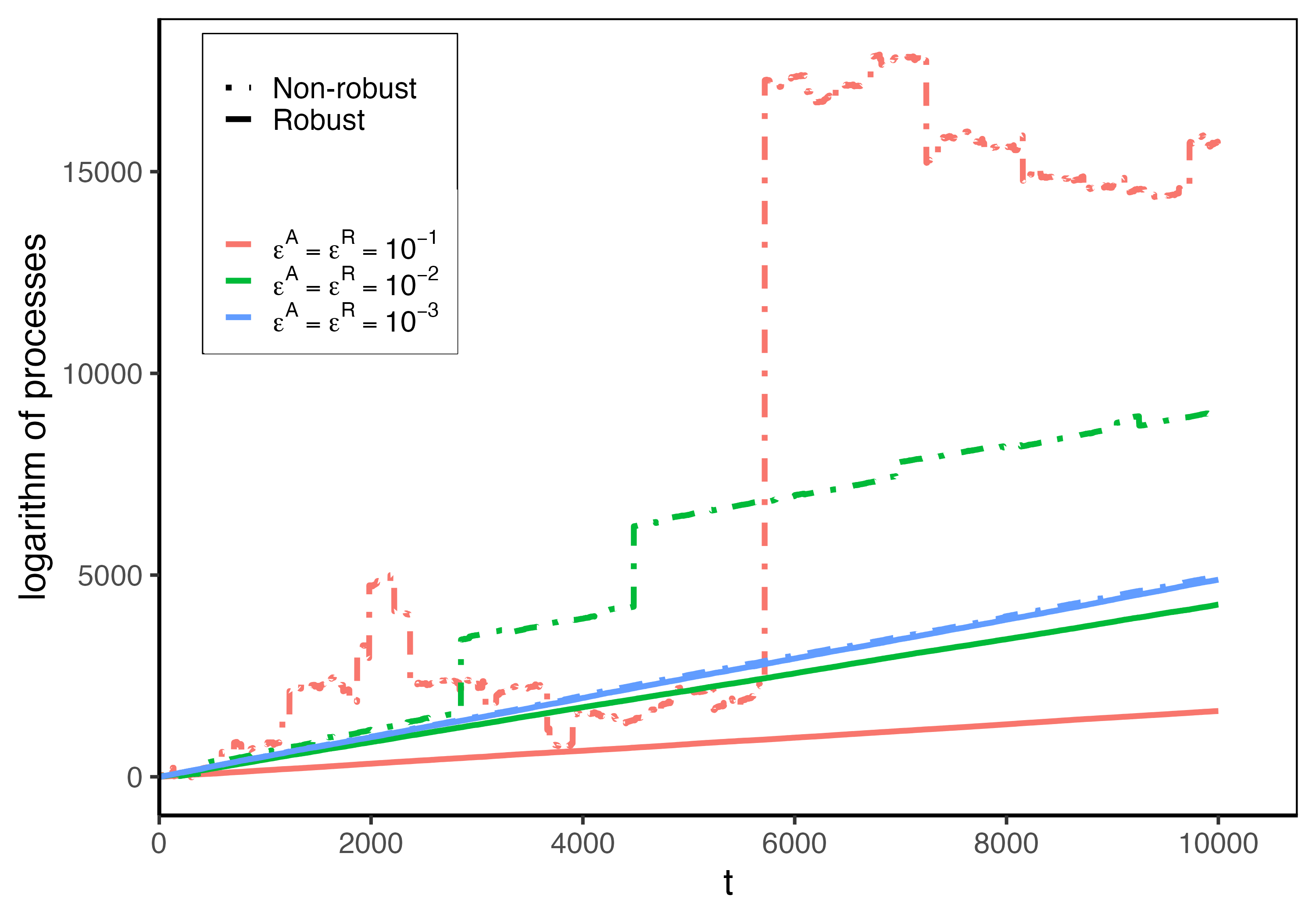}}
\caption[]{Data is drawn from $(1-\epsilon^R) \times N(1,1) + \epsilon^R \times\text{Cauchy}(-1,10)$ and $P_0=N(0,1), \mu_1=1$, $\epsilon^A=\epsilon^R=0.1,0.01,0.001$. The growth rate of our robust tests increases as $\epsilon$ decreases. As anticipated,  The growth rates for our robust tests based on simple and composite alternatives almost overlap. The growth rates for our robust tests based on simple and composite alternatives in the left and right subfigures look similar.} 
\label{fig:different-eps}  
\end{figure*}

\paragraph{Growth rate with different contamination.}

In this experiment, samples are simulated independently from the mixture distribution $(1-\epsilon^R)\times N(1,1) + \epsilon^R \times\text{Cauchy}(-1,10)$ for $\epsilon^R=10^{-3},10^{-2},10^{-1}$. This mixture model ensures that the $\epsilon^R$ fraction of the sample is drawn from the heavy-tailed Cauchy distribution with location and scale parameters $-1$ and $10$ respectively. For the simple alternative, we consider $\mu_1=1$.  \Cref{fig:different-eps} shows the growths of processes in logarithmic scale for both simple and composite alternative models: $P_1=N(1,1)$ (left) and $\mathcal{P}_1=\{N(\mu,1): \mu\neq 0\}$ (right). As expected, The growth rate of our robust tests increases as $\epsilon$ decreases. Notably, both simple and composite tests grow at similar rates (\cref{fig:different-eps} ). It is also evident that non-robust tests exhibit highly erratic behavior, even when plotting the averages of $10$ independent runs.

\subsection{Experiments with composite null}

In all the simulation experiments in this subsection, we consider the null to be $\mathcal P_0=\{N(\mu,1):-0.5\leq\mu\leq0.5\}$, the simple and the composite alternative to be $P_1=N(1,1)$ and $\mathcal{P}_1=\{N(\mu,1): \mu\leq 0.5 \text{ or } \mu \geq 0.5\}$ respectively.  For this composite null $\mathcal P_0$ and simple alternative $P_1$, ($N(0.5,1),N(1,1)$) is the LFD pair and we employ the test described in \Cref{sec:lfd-contamination-comp}. For this composite null $\mathcal P_0$ and composite alternative $\mathcal P_1$, there is no LFD. So, we use our robustified predictable plug-in method for composite alternative, as described in \Cref{sec:comp-null-alt} with the sample median as an estimate of $\mu$. All the results in Fig. \cref{fig:null-comp,fig:different-eps-comp} are the average of $10$ independent simulations. 
%We have used  numerical approximations for computing the terms $\sup_{P\in \mathcal{P}_0}\mathbb E_{X\sim P}\left[\frac{{q}_{1,\epsilon}(X)}{q_{0,\epsilon}(X)}\right]$ and $\sup_{P\in \mathcal{P}_0}\mathbb E_{X\sim P}\left[\frac{\hat{q}_{n,1,\epsilon}(X)}{\hat q_{n,0,\epsilon}(X)}\mid X^{n-1}\right]$ in the expressions \eqref{bet:comp-null} and \eqref{bet-comp-null-comp-alt}.
\paragraph{Sanity check under the null.}
In this experiment, samples are simulated independently from the following $\epsilon_R$-contaminated null distribution: $Q=(1-\epsilon^R)\times N(0,1) + \epsilon^R \times\text{Cauchy}(-1,10)$. Here $\epsilon^A=\epsilon^R=0.01$. For the tests with simple alternative, we consider $\mu_1=1$.  Fig. \ref{fig:null-comp} illustrates that our robust tests are ``safe'', i.e. they do not exhibit growth under the null hypothesis, whereas the non-robust methods show unreliable behavior with significant fluctuations.

\begin{figure}[!htb]
    \centering
    \includegraphics[width=0.5\linewidth]{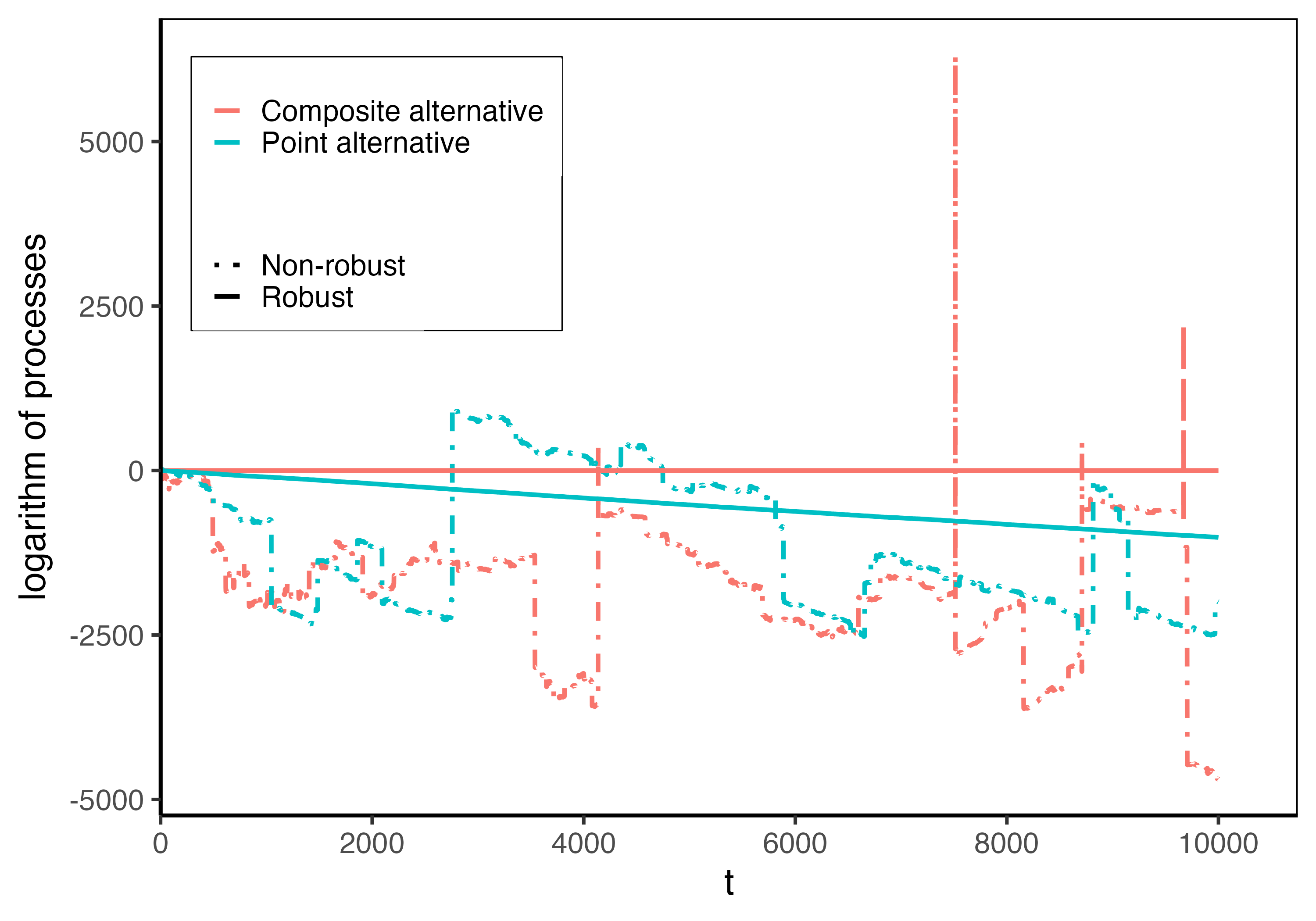}
    \caption{Data is drawn from $(1-\epsilon^R)\times N(0,1) + \epsilon^R \times\text{Cauchy}(-1,10)$. The null is $\mathcal P_0=\{N(\mu,1):-0.5\leq\mu\leq0.5\}$. Our robust tests are safe, but the non-robust tests exhibit unstable and unreliable behavior.}
    \label{fig:null-comp}
\end{figure}

\begin{figure*}[!htb]

\centering
\centering
\subfloat[ $P_1=N(1,1)$]{\includegraphics[width=0.49\linewidth,height=0.3\linewidth]{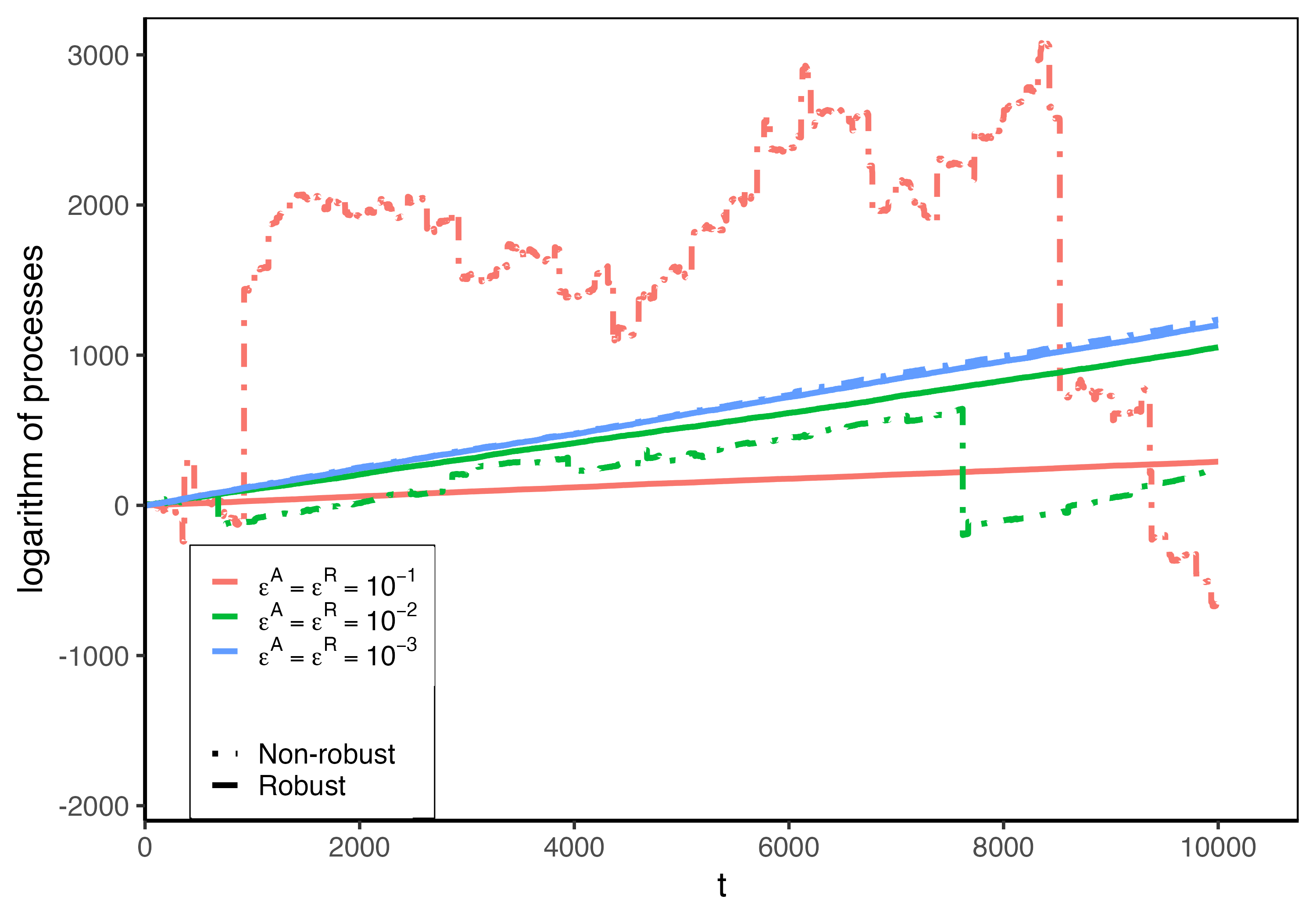}} 
\subfloat[$ \mathcal{P}_1=\{N(\mu,1): \mu\leq -0.5 \text{ or }\mu\geq 0.5\}$]{\includegraphics[width=0.49\linewidth,height=0.3\linewidth]{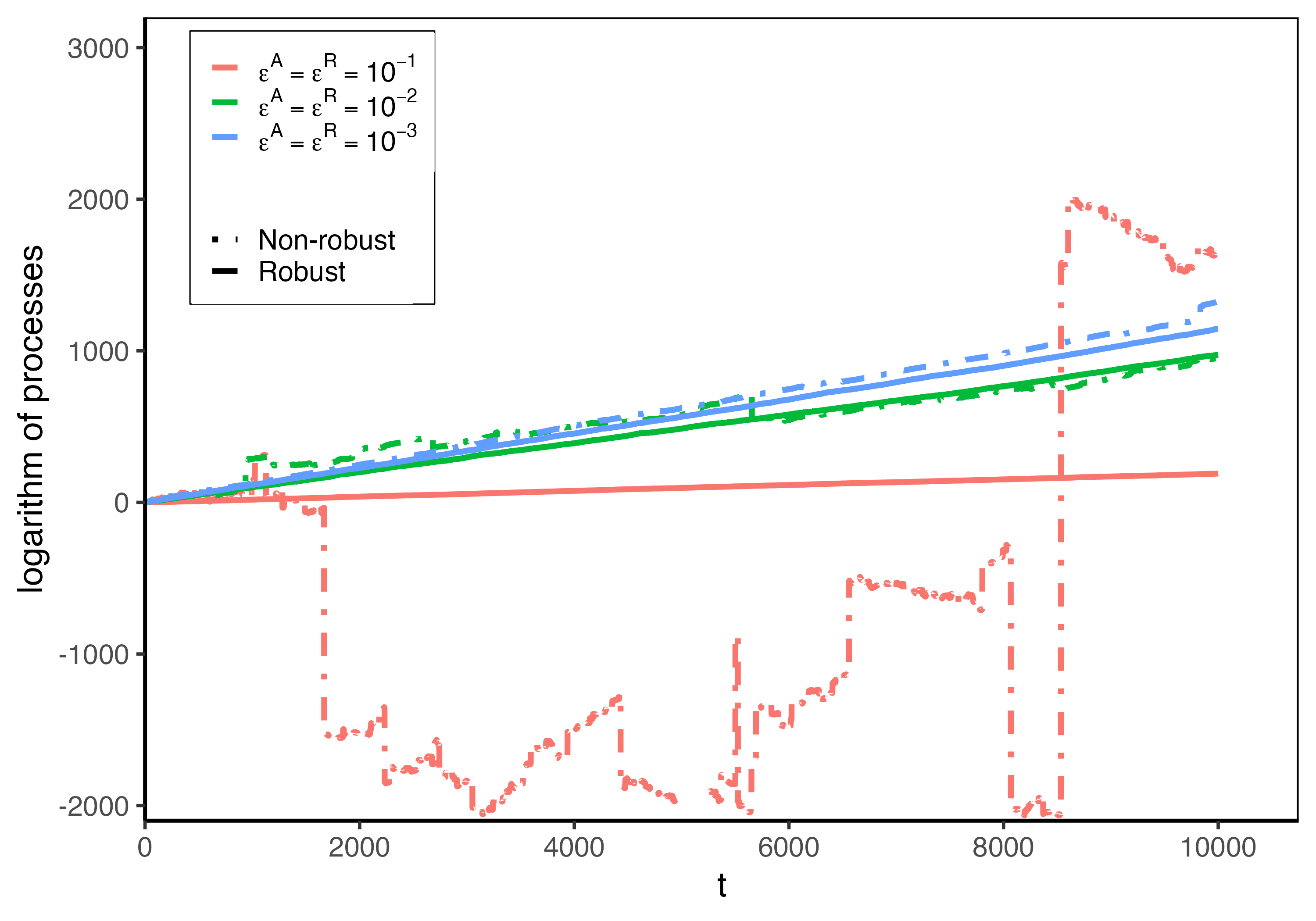}}
\caption[]{Data is drawn from $(1-\epsilon^R) \times N(1,1) + \epsilon^R \times\text{Cauchy}(-1,10)$ and $\epsilon^A=\epsilon^R=0.1,0.01,0.001$. The growth rate of our robust tests increases as $\epsilon$ decreases. As anticipated,  The growth rates for our robust tests based on simple and composite alternatives almost overlap. The growth rates for our robust tests based on simple and composite alternatives in the left and right subfigures look similar.} 
\label{fig:different-eps-comp}  
\end{figure*}

\paragraph{Growth rate with different contamination.}

In this experiment, samples are simulated independently from the mixture distribution $(1-\epsilon^R)\times N(1,1) + \epsilon^R \times\text{Cauchy}(-1,10)$ for $\epsilon^R=10^{-3},10^{-2},10^{-1}$.  \Cref{fig:different-eps-comp} shows the growths of processes in logarithmic scale for both simple and composite alternative models: $P_1=N(1,1)$ (left) and $\mathcal{P}_1=\{N(\mu,1): \mu\leq -0.5 \text{ or }\mu\geq 0.5\}$ (right). As expected, The growth rate of our robust tests increases as $\epsilon$ decreases. Notably, both simple and composite tests grow at similar rates (\Cref{fig:different-eps-comp}). It is also evident that non-robust tests exhibit highly erratic behavior, even when plotting the averages of $10$ independent runs.
\subsection{Experiment with multivariate data}

To evaluate the performance of our proposed robust likelihood ratio test in higher dimensions, we consider $P_0=N_d(\mathbf{0}_d,I_d)$ vs.\  $P_1=N_d(\mathbf{1}_d,I_d)$, where  $\mathbf{0}_d$ and $\mathbf{1}_d$ are $d$-dimentional vectors whose all elements are $0$ and $1$ respectively, $I_d$ is the $d\times d$ identity matrix. Data is drawn from $\epsilon^R$ contaminated $N_d(\mathbf{1}_d,I_d)$, where $\epsilon^R=0.1$ fraction of the data is corrupted with $d$ dimensional i.i.d. $\text{Cauchy}(-1,10)$ noise. \Cref{fig:multidim} shows the growths of our processes with $\epsilon^A=\epsilon^R$ (robust) and SPRT (non-robust) in
logarithmic scale for $d=5,10$ and $15$. Our test supermartingale exhibits stable exponential growth, whereas the non-robust SPRT demonstrates highly unstable behavior under contamination. We observe a higher growth rate of our test supermartingale for higher dimensions, which aligns with the fact that the KL divergence (optimal non-robust growth rate) between $N_d(\mathbf{0}_d,I_d)$ and $N_d(\mathbf{1}_d,I_d)$ increases with $d$.

\begin{figure}[!htb]
    \centering
    \includegraphics[width=0.5\linewidth]{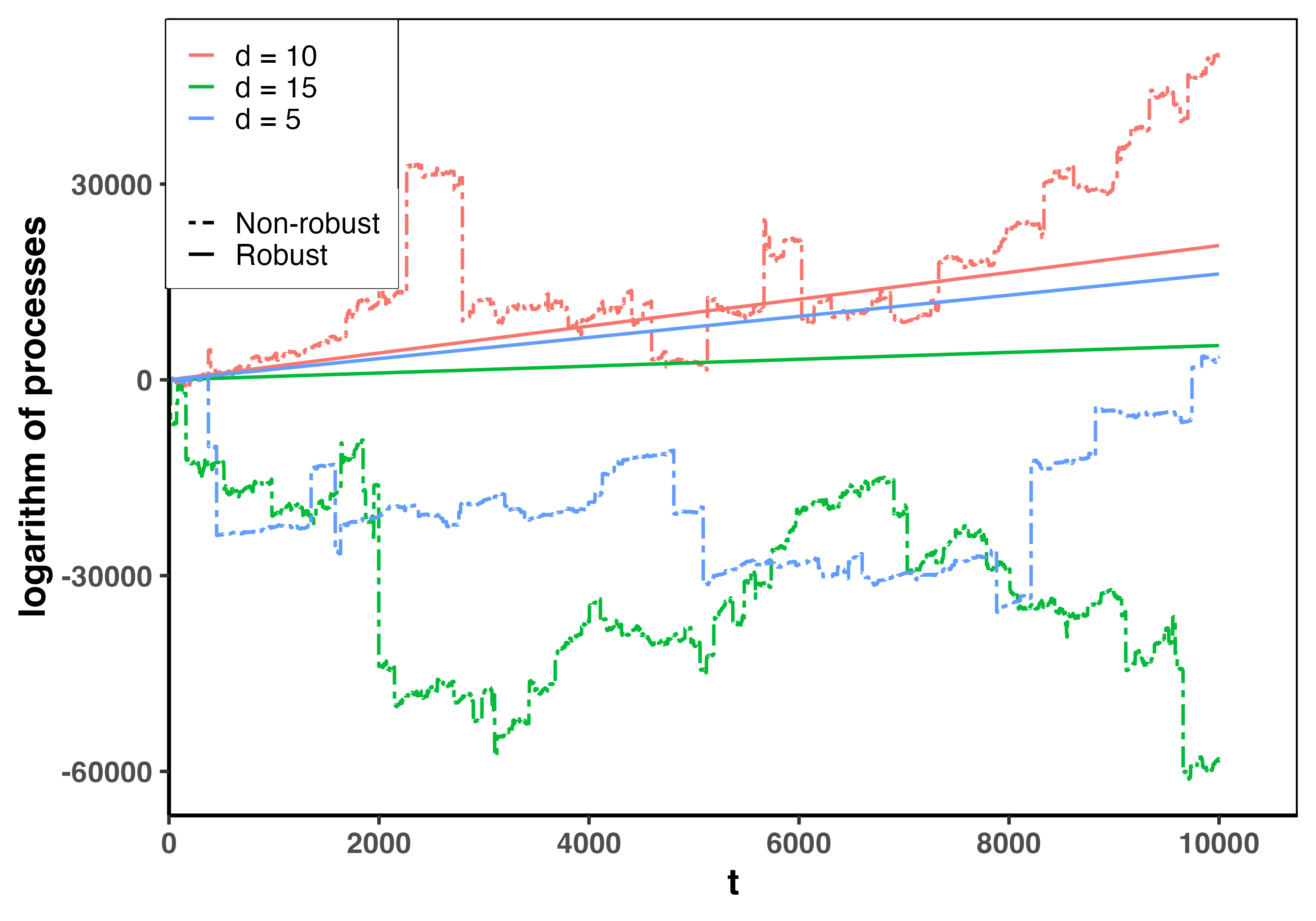}
    \caption{Data is drawn from $N_d(\mathbf{1}_d,I_d)$ contaminated with $d$ dimensional i.i.d. $\text{Cauchy}(-1,10)$. $\epsilon^R=\epsilon^A=0.1$.}
    \label{fig:multidim}
\end{figure}

\section{Conclusion}
\label{sec:conc}
In this paper, we established that the likelihood ratio of the least favorable distribution (LFD) pair forms the log-optimal e-value. Extending Huber's framework, we proved that if an LFD pair exists for a given composite null and alternative, then the LFDs of Huber’s $\epsilon$-contamination neighborhoods around that pair form the optimal LFD pair for the corresponding robustified hypotheses—thereby yielding log-optimal e-values in this broader setting as well.
For general composite nulls and alternatives where an exact LFD pair may not exist, we propose new methods and show that our proposed methods are at least asymptotically optimal as $\epsilon \to 0$. Establishing finite $\epsilon$ optimality for tractable testing procedures remains challenging in that setting. We leave this as an avenue for future research.

\subsection*{Acknowledgments}
A part of this work was done while AS was a student at the Indian Statistical Institute, Kolkata. The authors thank Hongjian Wang, Peter Grünwald, Shubhada Agrawal, and Sivaraman Balakrishnan for helpful conversations. The authors also thank the anonymous reviewers and AE for their helpful suggestions.

% \newpage
\appendix

\subsection{Mathematical details}
\label{a:proofs}
\begin{lemma}
\label{lem:c2-c1}
 As $\epsilon\downarrow 0$, $c^{\prime\prime}\uparrow\text{ess} \sup_{[\mu]} \frac{p_1}{p_0}$ and $c^{\prime}\downarrow\text{ess} \inf_{[\mu]} \frac{p_1}{p_0}$, and therefore, 
 \begin{equation*}
     q_{1,\epsilon}(X)/q_{0,\epsilon}(X)\to p_{1}(X)/p_{0}(X) \text{ almost surely.}
 \end{equation*}
\end{lemma}

\begin{proof}[Proof of \cref{lem:c2-c1}]
Define, 
\begin{equation}
    f(c)=P_0\left[p_1 / p_0<c\right]+\frac{1}{c} P_1\left[p_1 / p_0 \geq c\right]=1+\int_{p_1/p_0\geq c}(1/c-p_0/p_1)p_1d\mu.
\end{equation}
Note that $c=c^{\prime \prime}$ is a solution of the equation $f(c)=\frac{1}{1-\epsilon}.$
\begin{equation}
    f(c+\delta)-f(c)=-\int_{c \leq p_1/p_0\leq c+\delta}\left(\frac{1}{c}-\frac{p_0}{p_1}\right)p_1d\mu-\frac{\delta}{c(c+\delta)}\int_{p_1/p_0\geq c+\delta}p_1d\mu
\end{equation}
Therefore, $-\frac{\delta}{c(c+\delta)}\leq f(c+\delta)-f(c) \leq -\frac{\delta}{c(c+\delta)}P_1[p_1/p_0\geq c+\delta]$ for any $\delta>0$, which implies that $f$ is a continuous and decreasing function.

Let, $c_0=\text{ess} \sup_{[\mu]} \frac{p_1}{p_0}$. If $c_0<\infty$, we have $f(c)=1,$ for $c\ge c_0$ and for $c<c_0$, $f(c)$ is strictly decreasing because $f(c+\delta)-f(c) \leq -\frac{\delta}{c(c+\delta)}P_1[p_1/p_0\geq c+\delta]<0$, for small $\delta>0$. Therefore, the inverse of $f$ restricted on $[0,c_0]$ (denoted as $f^{-1}_{[0,c_0]}$) is a strictly decreasing and continuous function on the interval $[0,1]$ and hence, $c^{\prime \prime}=f^{-1}_{[0,c_0]}(\frac{1}{1-\epsilon})\to f^{-1}_{[0,c_0]}(1)=c_0,$ as $\epsilon\to 0$.

Now, if $c_0=\infty$, $f(c+\delta)-f(c) \leq -\frac{\delta}{c(c+\delta)}P_1[p_1/p_0\geq c+\delta]<0$, for all $c$ and hence $f$ is strictly decreasing with $\lim_{c\to\infty}f(c)=1$. 
Therefore, the inverse of $f$ is a strictly decreasing and continuous function on the interval $[0,1)$ and hence, $c^{\prime \prime}=f^{-1}_{[0,c_0]}(\frac{1}{1-\epsilon})\to \infty=c_0,$ as $\epsilon\to 0$.

%Note that $\frac{1}{1-\epsilon}\downarrow 1$, as $\epsilon\downarrow 0.$ Since, $f(c)$ is a strictly decreasing function for $c<c_0$, the solution of the equation $f(c)=\frac{1}{1-\epsilon}$ increases to $c_0$ in both cases. 
Thus, we have shown $c^{\prime\prime}\uparrow\text{ess} \sup_{[\mu]} \frac{p_1}{p_0}$, as $\epsilon\to 0.$ Similarly, one can show that $c^{\prime}\downarrow\text{ess} \inf_{[\mu]} \frac{p_1}{p_0}$, as $\epsilon\to 0.$

Therefore, $\frac{q_{1,\epsilon}(X)}{q_{0,\epsilon}(X)}=\frac{p_1(X)}{p_0(X)}\mathds{1}(c^{\prime}\leq \frac{p_1(X)}{p_0(X)}\leq c^{\prime\prime})+c^{\prime}\mathds{1}( \frac{p_1(X)}{p_0(X)}< c^{\prime})+c^{\prime\prime}\mathds{1}( \frac{p_1(X)}{p_0(X)}> c^{\prime\prime})\to \frac{p_1(X)}{p_0(X)},$ $\mu$-almost surely as $\epsilon\to 0.$
\end{proof}

\begin{lemma}
\label{lem:c2eps}
Suppose that $\kl(P_1,P_0)<\infty$. Then, $(c^{\prime\prime}-c^{\prime})\epsilon \to 0,$ as $\epsilon\to 0.$
\end{lemma}
\begin{proof}[Proof of \cref{lem:c2eps}]
Let, $c_0=\text{ess} \sup_{[\mu]} \frac{p_1}{p_0}$. If $c_0<\infty$, it follows from \Cref{lem:c2-c1} that $c^{\prime\prime}\leq c_0$ and so $c^{\prime\prime}\epsilon\to 0$, as $\epsilon\to 0$.

Now, if $c_0=\infty$, $c^{\prime\prime}\to \infty$ as $\epsilon\to 0.$

From \eqref{eq:c2}, $1+\frac{1}{c^{\prime \prime}}P_1\left[p_1 / p_0 \geq c^{\prime \prime}\right]\geq \frac{1}{1-\epsilon}$, which implies 
\begin{equation}
\label{eq:c2eps}
    c^{\prime \prime}\epsilon\leq (1-\epsilon) P_1\left[p_1 / p_0 \geq c^{\prime \prime}\right]
\end{equation}

If $\kl(P_1,P_0)<\infty$, we have $\mathbb E_{P_1}\left|\log(p_1/p_0)\right|<\infty.$ Then,
\begin{align*}
    P_1\left[p_1 / p_0 \geq c^{\prime \prime}\right]= P_1\left[\log (p_1 / p_0) \geq \log c^{\prime \prime}\right]&\le P_1\left[ |\log (p_1 / p_0)| \geq \log c^{\prime \prime}\right]\\
    &\leq \frac{\mathbb E_{P_1}  |\log (p_1 / p_0)|}{\log c^{\prime \prime}}\to 0,
\end{align*}

 as $c^{\prime \prime}\to \infty$. Hence, $c^{\prime\prime}\epsilon \to 0,$ since $ c^{\prime \prime}\to \infty$, as $\epsilon\to 0$ for the case when $c_0=\infty$.
\end{proof}

\begin{lemma}
\label{lem:c-conv}
Assume that $P_i({p}_1/p_0 = c)=0$, for all $c\in \mathbb R$ and for $i=0,1$. If $\hat{p}_{1,n}\to p_1$  almost surely as $n\to \infty$ and $c^{\prime}_n,c^{\prime \prime}_n$ are solutions of \eqref{eq:c2-comp-alt} and \eqref{eq:c1-comp-alt} respectively, then $c^{\prime \prime}_n\to c^{\prime \prime}$ and $c^{\prime}_n\to c^{\prime}$ almost surely as $n\to \infty$, where $c^{\prime}$ and $c^{\prime \prime}$ are solutions of \eqref{eq:c2} and \eqref{eq:c1}.
\end{lemma}
\begin{proof}
Define, $A_n=\{x:\hat{p}_{1,n}(x)/ p_0(x)>c\}$ and $A=\{x:{p}_{1}(x)/ p_0(x)>c\}$.

For \( x \in A \) (where \( {p}_1/p_0 > c \)): There exists an \( N \) such that for all \( n \geq N \), \( \hat{p}_{1,n}/p_0 > c \). Hence, \( x \in A_n \) for  \( n\geq N \). This implies that:
    \[
    A \subseteq \liminf_{n \to \infty} A_n.
    \]
 For \( x \notin A \) (where \( {p}_1/p_0 \leq c \)):
    \begin{itemize}
        \item If \( {p}_1/p_0 < c \), then for sufficiently large \( n \), \( \hat{p}_{1,n} < c \), and hence \( x \notin A_n \).
        \item If \( {p}_1/p_0 = c \), then the set of such points forms the boundary. By assumption, this set has zero probability.
    \end{itemize}
    Thus:
    \[
    P_1(\limsup_{n \to \infty} A_n) \leq P_1(A) \leq P_1(\liminf_{n \to \infty} A_n).
    \]
    %Since \( \liminf_{n \to \infty} A_n \subseteq A \subseteq \limsup_{n \to \infty} A_n \), we conclude that \( A_n \) converges to \( A \) in measure, meaning:

Consider:
\[
\hat{p}_{1,n}(A_n) = \int_{A_n} \hat{p}_{1,n} \, d\mu.
\]

Since \( \hat{p}_{1,n} \to {p}_1 \) pointwise, Scheffe's theorem gives $\int|\hat{p}_{1,n}-p_1|d\mu\to 0.$
\[
 \int_{A_n} \hat{p}_{1,n} \, d\mu \leq \int_{A_n} {p}_1 \, d\mu + \int_{A_n} |\hat{p}_{1,n}-p_1|\,d\mu \leq \int_{A_n} {p}_1 \, d\mu + \int|\hat{p}_{1,n}-p_1|\,d\mu.
\]
Therefore:
\[
\limsup_{n \to \infty} \hat{p}_{1,n}(A_n) \leq \limsup_{n \to \infty} \int_{A_n} {p}_1 \, d\mu = \limsup_{n \to \infty} P_1(A_n) \leq P_1(\limsup_{n \to \infty} A_n)\leq P_1(A).
\]
Similarly.
\[
 \int_{A_n} \hat{p}_{1,n} \, d\mu \geq \int_{A_n} {p}_1 \, d\mu - \int_{A_n} |\hat{p}_{1,n}-p_1|\,d\mu \geq \int_{A_n} {p}_1 \, d\mu - \int|\hat{p}_{1,n}-p_1|\,d\mu.
\]
Therefore:
\[
\liminf_{n \to \infty} \hat{p}_{1,n}(A_n) \geq \liminf_{n \to \infty} \int_{A_n} {p}_1 \, d\mu = \liminf_{n \to \infty} P_1(A_n) \geq P_1(\liminf_{n \to \infty} A_n)\geq P_1(A).
\]
Combining the upper and lower bounds, we conclude:
\[
\lim_{n \to \infty} \hat{p}_{1,n}[\hat{p}_{1,n}/ p_0>c] = {P}_1[{p}_{1}/ p_0>c].
\]
Similarly, one can show that
\[
\lim_{n \to \infty} \hat{p}_{1,n}[\hat{p}_{1,n}/ p_0<c] = {P}_1[{p}_{1}/ p_0<c].
\]
Define, $$f_n(c)=P_0\left[\hat{p}_{1,n} / p_0<c\right]+\frac{1}{c} \hat{P}_{1,n}\left[\hat{p}_{1,n} / p_0 \geq c\right] ,\quad  f(c)=P_0\left[p_1 / p_0<c\right]+\frac{1}{c} P_1\left[p_1 / p_0 \geq c\right]$$
$$g_n(c)=\hat{P}_{1,n}\left[\hat{p}_{1,n}/ p_0>c\right]+c P_0\left[\hat{p}_{1,n} / p_0 \leq c\right],\quad
g(c)={P}_1\left[{p}_1/ p_0>c\right]+cP_0\left[{p}_1 / p_0 \leq c\right].$$ 
Then, it follows from what we have shown above that $f_n\to f$ and $g_n\to g$ pointwise.
In \cref{lem:c2-c1}, we have shown that $f_n,g_n,f,g$ are all strictly monotone and continuous. Therefore, pointwise convergence implies uniform convergence, and hence,
$c^{\prime \prime}_n=f_n^{-1}(\frac{1}{1-\epsilon})\to f^{-1}(\frac{1}{1-\epsilon})=c^{\prime \prime}$ and $c^{\prime}_n=g_n^{-1}(\frac{1}{1-\epsilon})\to g^{-1}(\frac{1}{1-\epsilon})=c^{\prime}$.
\end{proof}

\begin{proof}[Proof of \Cref{thm:gr-lb-comp-alt}]
We have $\hat{p}_{1,n}\to p_1^H$. It follows from \cref{lem:c-conv} that $c^{\prime \prime}_n\to c^{\prime \prime}_H, c^{\prime}_n\to c^{\prime}_H.$ Therefore, for any fixed $\epsilon$, $\hat{q}_{n,1,\epsilon}/\hat q_{n,0,\epsilon}\to{q}_{1,\epsilon}^H/q_{0,\epsilon}^H$ almost surely as $n\to\infty$ and by bounded convergence theorem, $\mathbb E_{X\mid X^{n-1}\sim P_0}\left[\frac{\hat{q}_{n,1,\epsilon}(X)}{\hat{q}_{n,0,\epsilon}(X)}\mid X^{n-1}\right]\to\mathbb E_{X\mid X^{n-1}\sim P_0}\left[\frac{{q}^H_{1,\epsilon}(X)}{q^H_{0,\epsilon}(X)}\mid X^{n-1}\right]$ as $n\to\infty$ and
that would immediately imply $\log \hat E_{\epsilon,n}(X_n)-\log E_\epsilon^H(X_n)\to 0$ almost surely as $n\to\infty$, where
\begin{equation}
    E_\epsilon^H(x)=\frac{\frac{{q}^H_{1,\epsilon}(x)}{q^H_{0,\epsilon}(x)}}{\mathbb E_{X\mid X^{n-1}\sim P_0}\left[\frac{{q}^H_{1,\epsilon}(X)}{q^H_{0,\epsilon}(X)}\mid X^{n-1}\right]+(c^{\prime \prime}_H-c^{\prime}_H)\epsilon}.
\end{equation}
Therefore,
\begin{align*}
 \frac{1}{n}\log R_{n,\epsilon}^{\text{plug-in}}=\frac{1}{n}\sum_{i=1}^n\left(\log \hat E_{\epsilon,i}(X_i)-\log E^H_\epsilon(X_i) \right)+\frac{1}{n}\sum_{i=1}^n \log E_\epsilon^H(X_i)\to \mathbb E_H \log E^H_\epsilon(X),
\end{align*}
almost surely, since the first term converges to $0$ (since $\log \hat E_{\epsilon,n}(X_n)-\log E_\epsilon^H(X_n)\to 0$ and hence their average would converge to $0$) and the second term converges to $r_{H,\epsilon}^{\text{plug-in}}=\mathbb E_H \log E^H_\epsilon(X)$, by SLLN.

Now, one can easily show that  $$r_{H,\epsilon}^{\text{plug-in}}=\mathbb E_H \log E^H_\epsilon(X) \geq \kl(Q^H_{1,\epsilon},Q^H_{0,\epsilon})-2(\log c^{\prime \prime}_H-\log c^{\prime}_H)\epsilon-\log (1+2(c^{\prime \prime}_H-c^{\prime}_H)\epsilon).$$
\end{proof}

\begin{proposition}
\label{lem:limit-comp-null}
Assume that \eqref{eq:sup-ripr} holds and $\inf_{P\in\mathcal{P}_0}\kl(P_1,P)<\infty$. Then, $B_{\epsilon}(x)\to B^*(x)$ $\mu$-almost surely as $\epsilon\to0$. In other words, for any fixed $n\in \mathbb N$, $R_{n,\epsilon}^{\text{RIPr}}\to\prod_{i=1}^n p_1(X_i)/p_0(X_i)$ $\mu$-almost surely as $\epsilon\to0.$
\end{proposition}

\begin{proof}

It is easy to verify that \cref{lem:c2-c1} and \cref{lem:c2eps} holds when $P_0$ is sub-probability distribution as well. Now,
\cref{lem:c2-c1} implies $\frac{q_{1,\epsilon}(X)}{q_{0,\epsilon}(X)}\to \frac{p_{1}(X)}{p_{0}(X)}$ almost surely, as $\epsilon\to 0$ and therefore, $\mathbb E_P\frac{q_{1,\epsilon}(X)}{q_{0,\epsilon}(X)}\to \mathbb E_P\frac{p_{1}(X)}{p_{0}(X)}$, as $\epsilon\to 0$. Since $0\leq c^{\prime}<c^{\prime\prime}$, \cref{lem:c2eps} implies $(c^{\prime\prime}-c^{\prime})\epsilon \to 0,$ since we assumed $\kl(P_1, P_0) < \infty$.

Since $\frac{p_{1}(X)}{p_{0}(X)}$ is an e-variable for $\mathcal{P}_0$ \cite{larsson2024numeraire}, we have  $\mathbb E_P\frac{p_{1}(X)}{p_{0}(X)}\leq 1$, for all $P\in\mathcal{P}_0$, i.e., $\sup_{P\in\mathcal{P}_0}\mathbb E_P\frac{p_{1}(X)}{p_{0}(X)}\leq1.$

To show the reverse inequality, define, $B^\prime (x):=\frac{\frac{p_{1}(X)}{p_{0}(X)}}{\sup_{P\in\mathcal{P}_0}\mathbb E_P\frac{p_{1}(X)}{p_{0}(X)}}$. Then it is clear that $\mathbb E_{P}B^\prime (X)\leq 1$ for all $P\in\mathcal{P}_0$. Since $\frac{p_{1}(X)}{p_{0}(X)}$ is log-optimal, we have $\mathbb E_{P}\log B^\prime (X)\leq \mathbb E_{P}\log \frac{p_{1}(X)}{p_{0}(X)}$, which implies $\sup_{P\in\mathcal{P}_0}\mathbb E_P\frac{p_{1}(X)}{p_{0}(X)}\geq 1$.

Combining the above two arguments, we obtain $\sup_{P\in\mathcal{P}_0}\mathbb E_P\frac{p_{1}(X)}{p_{0}(X)}=1.$

Note that 
\begin{align*}
    & \sup_{P\in \mathcal{P}_0}\mathbb E_{ P}\frac{{q}_{1,\epsilon}(X)}{q_{0,\epsilon}(X)}=\mathbb E_{ P_0}\frac{{q}_{1,\epsilon}(X)}{q_{0,\epsilon}(X)}\to \mathbb E_{ P_0}\frac{p_{1}(X)}{p_{0}(X)}=1,
\end{align*}
as $\epsilon\to 0$. Therefore,
\begin{equation}
\label{eq:b}
    \sup_{P\in \mathcal{P}_0}\mathbb E_{ P}\left[\frac{{q}_{1,\epsilon}(X)}{q_{0,\epsilon}(X)}\right]\to 1, \text{ as } \epsilon\to 0
\end{equation} 

Thus, we obtain $B_{\epsilon}(X)\to B^*(X)$ almost surely as $\epsilon\to0$ and for any $n\in \mathbb N$,
\begin{equation}
    R_{n,\epsilon}^{\text{RIPr}}=\prod_{i=1}^n B_{\epsilon}(X_i)\stackrel{a.s}{\longrightarrow}\prod_{i=1}^n B^*(X_i)=\prod_{i=1}^n p_1(X_i)/p_0(X_i)  \text{ as } \epsilon\to 0.
\end{equation}
\end{proof}

\begin{proof}[Proof of \Cref{prop:sup-ripr-exp}]
 We prove it for $\theta_0^*=\sup\Theta_0\leq\theta_1$ only, and the other case, when $\inf\Theta_0\geq\theta_1$ can be proved in the same way. 
 The first part that the RIPr is $P_{\theta_0^*}$ follows from the proof of \Cref{prop:one-param-exp-fam}.
The exponential family $\{P_\theta\}$ with one-dimensional sufficient statistic $T(x)$ satisfies the monotone likelihood ratio (MLR) property in  $T(x)$, and since $\frac{p_{\theta_1}(x)}{p_{\theta_0^*}(x)}$ is non-decreasing in $T(x)$, it follows that  $\frac{q_{1,\epsilon}(X)}{q_{0,\epsilon}(X)}=\max\{c^\prime,\min\{p_{\theta_0^*}(X)/p_{\theta_1}(X),c^{\prime\prime}\}\}$ is also non-decreasing in $T(x)$. It follows from properties of MLR families that  $T(X)$ is stochastically increasing in $\theta$. That is, for $X \sim P_\theta$ and $X' \sim P_{\theta^*}$ with $\theta\leq \theta^*$, we have $P(T(X)>x)\leq P(T(X')>x)$ for all $x$, which is denoted as $T(X) \preceq_{\text{st}} T(X')$ (see e.g., \cite[Chapter 3.4]{lehmann1986testing}) and hence,
\[
\frac{q_{1,\epsilon}(X)}{q_{0,\epsilon}(X)} \preceq_{\text{st}} \frac{q_{1,\epsilon}(X')}{q_{0,\epsilon}(X')}.
\]
Therefore, it immediately follows that
\[
\mathbb E_{P_\theta}\left(\frac{q_{1,\epsilon}(X)}{q_{0,\epsilon}(X)}\right) \leq\mathbb E_{P_{\theta^*}}\left(\frac{q_{1,\epsilon}(X)}{q_{0,\epsilon}(X)}\right), \forall \theta\in\Theta_0.
\]
Thus, \eqref{eq:sup-ripr} holds.
\end{proof}

\begin{proof}[Proof of \Cref{thm:gr-lb-comp-null}]
By SLLN,
\begin{equation}
 \frac{\log R_{n,\epsilon}^{\text{RIPr}}}{n}\to r^{Q,\epsilon}_{\text{RIPr}} \text{ almost surely, }
\end{equation}
where $r^{Q,\epsilon}_{\text{RIPr}}=\mathbb E_{Q}\log\frac{q_{1,\epsilon}(X)}{q_{0,\epsilon}(X)}-\log\left(\sup_{P\in \mathcal{P}_0}\mathbb E_{X\sim P}\left[\frac{{q}_{1,\epsilon}(X)}{q_{0,\epsilon}(X)}\right]+(c^{\prime \prime}-c^{\prime})\epsilon\right).$

%Since $\tv(Q_{0,\epsilon},P_0)<\epsilon$,
%\begin{equation*}
%    1\geq\mathbb E_{Q_{0,\epsilon}}\frac{q_{1,\epsilon}(X)}{q_{0,\epsilon}(X)}\geq \mathbb E_{P_0}\frac{q_{1,\epsilon}(X)}{q_{0,\epsilon}(X)} - (c^{\prime \prime}-c^{\prime})\epsilon.\end{equation*}
%Hence, $r_Q^\epsilon\geq\mathbb E_{Q}\log\frac{q_{1,\epsilon}(X)}{q_{0,\epsilon}(X)}-\log(1+2(c^{\prime \prime}-c^{\prime})\epsilon)$.
Note that $\tv(Q_{1,\epsilon},Q)<2\epsilon$, so
\begin{equation}
%\label{eq:a}
    \left|\mathbb E_{Q}\log\frac{q_{1,\epsilon}(X)}{q_{0,\epsilon}(X)}- \mathbb E_{Q_{1,\epsilon}}\log\frac{q_{1,\epsilon}(X)}{q_{0,\epsilon}(X)}\right|\leq 2(\log c^{\prime \prime}-\log c^{\prime})\epsilon.
\end{equation}
Therefore,
\begin{equation*}
    r^{\epsilon}_{\text{RIPr}}=\inf_{Q\in H_1^\epsilon}r^{Q,\epsilon}_{\text{RIPr}}\geq \kl(Q_{1,\epsilon},Q_{0,\epsilon})-2(\log c^{\prime \prime}-\log c^{\prime})\epsilon-\log\left(\sup_{P\in \mathcal{P}_0}\mathbb E_{P}\frac{{q}_{1,\epsilon}(X)}{q_{0,\epsilon}(X)}+(c^{\prime \prime}-c^{\prime})\epsilon\right).
\end{equation*}
\end{proof}

\begin{proof}[Proof of \Cref{thm:asymp-gr-comp-null}]

From \eqref{eq:c2-comp-null}, we get $(1-\epsilon)(k+\frac{1}{c^{\prime \prime}})\geq k$, which implies $kc^{\prime \prime}\leq \frac{1}{\epsilon}-1$. Similarly, from  \eqref{eq:c1}, we get $kc^{\prime}\geq \frac{\epsilon}{1-\epsilon}$. Therefore, $(\log c^{\prime \prime}-\log c^{\prime})\epsilon\to 0$, as $\epsilon\to 0$. From \cref{lem:c2-c1}, we have $E_{P_1}\log\frac{q_{1,\epsilon}(X)}{q_{0,\epsilon}(X)}\to \kl(P_1,P_0)$, as $\epsilon\to 0$.
Note that $\tv(P_1,Q)\leq\epsilon$, so
\begin{equation}
\label{eq:a}
    \left|\mathbb E_{Q}\log\frac{q_{1,\epsilon}(X)}{q_{0,\epsilon}(X)}- \mathbb E_{P_1}\log\frac{q_{1,\epsilon}(X)}{q_{0,\epsilon}(X)}\right|\leq (\log c^{\prime \prime}-\log c^{\prime})\epsilon.
\end{equation} 
Therefore, $\mathbb E_{Q}\log\frac{q_{1,\epsilon}(X)}{q_{0,\epsilon}(X)}\to  \kl(P_1,P_0)$, as $\epsilon\to 0$.
From \cref{lem:c2eps} and \eqref{eq:b}, we have 

$\log\left(\sup_{P\in \mathcal{P}_0}\mathbb E_{X\sim P}\left[\frac{{q}_{1,\epsilon}(X)}{q_{0,\epsilon}(X)}\right]+(c^{\prime \prime}-c^{\prime})\epsilon\right)\to 0.$ Hence,
\begin{equation}
    r^{Q,\epsilon}_{\text{RIPr}}=\mathbb E_{Q}\log\frac{q_{1,\epsilon}(X)}{q_{0,\epsilon}(X)}-\log\left(\sup_{P\in \mathcal{P}_0}\mathbb E_{X\sim P}\left[\frac{{q}_{1,\epsilon}(X)}{q_{0,\epsilon}(X)}\right]+(c^{\prime \prime}-c^{\prime})\epsilon\right) \to  \kl(P_1,P_0).
\end{equation}
Thus, $r^{\epsilon}_{\text{RIPr}}=\inf_{Q\in H^\epsilon_1}r^{Q,\epsilon}_{\text{RIPr}}\to  \kl(P_1,P_0)$, as $\epsilon\to 0.$
\end{proof}

\begin{proof}[Proof of \Cref{prop:one-param-exp-fam}]
 Let, $P_{\theta^*}$ be the RIPr of $P_{\theta_1}$ on $\mathcal{P}_0=\{P_{\theta}: \theta\in[a,b] \text{ for some } -\infty\leq a\leq b\leq \infty\}$ and $\theta_1\notin[a,b]$. Then we shall show that $P_{\theta^*}$ has density $p_{\theta^*}$, where $\theta^*$ is the closest element in $[a,b]$ from $\theta_1$.

It is enough to show that $\mathbb E_{\theta_1}(p_{\theta}/p_{\theta^*})\leq 1$, for all $\theta\in \theta_0$ \cite[Theorem 4.7]{larsson2024numeraire}.
Now, 
\begin{align*}
    \mathbb E_{\theta_1}(p_{\theta}(X)/p_{\theta^*}(X))=&\mathbb E_{\theta_1}(\exp\{(\theta_1-\theta^*)T_k(X)-A(\theta)+A(\theta^*)\})\\
    =&\int \exp\{(\theta_1-\theta^*)T_k(x)-A(\theta)+A(\theta^*)+\theta_1T(x)-A(\theta_1)\}h(x)dx\\
    =& \exp\{A(\theta_1-\theta^*+\theta_1)-A(\theta)+A(\theta^*)-A(\theta_1)\}
\end{align*}
Since A is convex, its derivative $A^\prime$ is increasing and either $\theta_1<\theta^*\leq\theta$ or $\theta_1>\theta^*\geq\theta$. So,
\begin{align*}
    A(\theta_1-\theta^*+\theta_1)-A(\theta_1)&=\int_0^1 (\theta-\theta^*)A^\prime(\theta_1+t(\theta-\theta^*))dt\\
    &\leq \int_0^1 (\theta-\theta^*)A^\prime(\theta_1+t(\theta^*-\theta^*))dt\\
    &=A(\theta)-A(\theta^*).
\end{align*}
Therefore, $\mathbb E_{\theta_1}(p_{\theta}/p_{\theta^*})\leq 1$, for all $\theta\in \theta_0$.

Since $\hat{\theta}_n\to{\theta}_1(H)$, $\hat\theta^*_n$ converges to either $a$ or $b$, we call the limit $\theta_0(H)$. We have $\hat{\theta}^*_n\to{\theta}_0(H)$. We also have shown above that the RIPr $P_0^H$ of $P_1^H$ has density $p_0^H=p_{{\theta}_0(H)}$. 
So, we now have that $\hat p_{0,n}\to p_0^H$. Therefore, \Cref{assmp-ripr-conv} holds.
\end{proof}

\begin{proof}[Proof of \Cref{thm:gr-exp-family}]
By the above lemma, its RIPr would have density $\hat p_{0,n}=p_{\hat{\theta}^*_n}$, 
%with $$\hat{\theta}^*_n=(\hat\theta_{n,1},\cdots,\hat\theta_{n,k-1},\hat\theta^*_n,\hat\theta_{n,k+1},\cdots,\hat\theta_{n,d}),$$
where $\hat\theta^*_n$ is the nearest element in $[a,b]$ from $\hat\theta_{n}\notin[a,b]$ (so $\hat\theta^*_n$ is either $a$ or $b$). Since $\hat{\theta}_n\to{\theta}_1(H)$, $\hat\theta^*_n$ converges to either $a$ or $b$, we call the limit $\theta_0(H)$. We have $\hat{\theta}^*_n\to{\theta}_0(H)$. By lemma, we also have that the RIPr $P_0^H$ of $P_1^H$ has density $p_0^H=p_{{\theta}_0(H)}$. 
Therefore, we now have that $\hat p_{0,n}\to p_0^H$, and $\hat p_{1,n}\to p_1^H$. 
Define,
$$B_\epsilon^H(x)=\frac{\frac{{q}^H_{1,\epsilon}(x)}{q^H_{0,\epsilon}(x)}}{\sup_{P\in \mathcal{P}_0}\mathbb E_{X\sim P}\left[\frac{{q}^H_{1,\epsilon}(X)}{q^H_{0,\epsilon}(X)}\right]+(c^{\prime \prime}_H-c^{\prime}_H)\epsilon}.$$
Since $c^{\prime \prime}_n\to c^{\prime \prime}_H, c^{\prime}_n\to c^{\prime}_H$, we have $\hat{q}_{n,1,\epsilon}/\hat q_{n,0,\epsilon}\to{q}^H_{1,\epsilon}/q^H_{0,\epsilon}$ almost surely as $n\to\infty$ and that would immediately imply $\log \hat B_{\epsilon,i}(X_i)-\log B^H_\epsilon(X_i)\to 0$ almost surely as $n\to\infty$.
\begin{align*}
 \frac{1}{n}\log R_{n,\epsilon}^{\text{plug-in,RIPr}}=\frac{1}{n}\sum_{i=1}^n\left(\log \hat B_{n,\epsilon}(X_i)-\log B^H_\epsilon(X_i) \right)+\frac{1}{n}\sum_{i=1}^n \log B^H_\epsilon(X_i) \to r^{H,\epsilon}_{\text{RIPr,plug-in}},
\end{align*}
since the first term converges to $0$ and the second term converges to $r^{Q,\epsilon}_{\text{RIPr,plug-in}}=\mathbb E_H \log B^H_\epsilon(X)$, by SLLN. Imitating the proof in \Cref{thm:gr-lb-comp-null},
\begin{align*}
    r^{H,\epsilon}_{\text{RIPr,plug-in}}&=\mathbb E_H \log B^H_\epsilon(X)\\
    &\geq \mathbb E_{P_{\theta_1}}\log(q_{1,\epsilon}^H/q_{0,\epsilon}^H)-(\log c^{\prime \prime}_H-\log c^{\prime}_H)\epsilon-\log\left(\sup_{P\in \mathcal{P}_0}\mathbb E_{P}\frac{{q}^H_{1,\epsilon}(X)}{q^H_{0,\epsilon}(X)}+(c^{\prime \prime}_H-c^{\prime}_H)\epsilon\right).
\end{align*}

\end{proof}

{

\subsection{Figures with error bars}

We consider the (pre-contaminated) null $P_0$ to be $N(0,1)$, the alternative to be $P_1=N(1,1)$. For simple null vs simple alternative, we employ the robust test described in \Cref{sec:lfd-contamination-simp}, and the non-robust test is the classical SPRT.  \Cref{fig:error-bar} is the average of $10$ independent simulations, along with error bars.

\begin{figure*}[!htb]

\centering
\centering
\subfloat[Robust test]{\includegraphics[width=0.49\linewidth,height=0.3\linewidth]{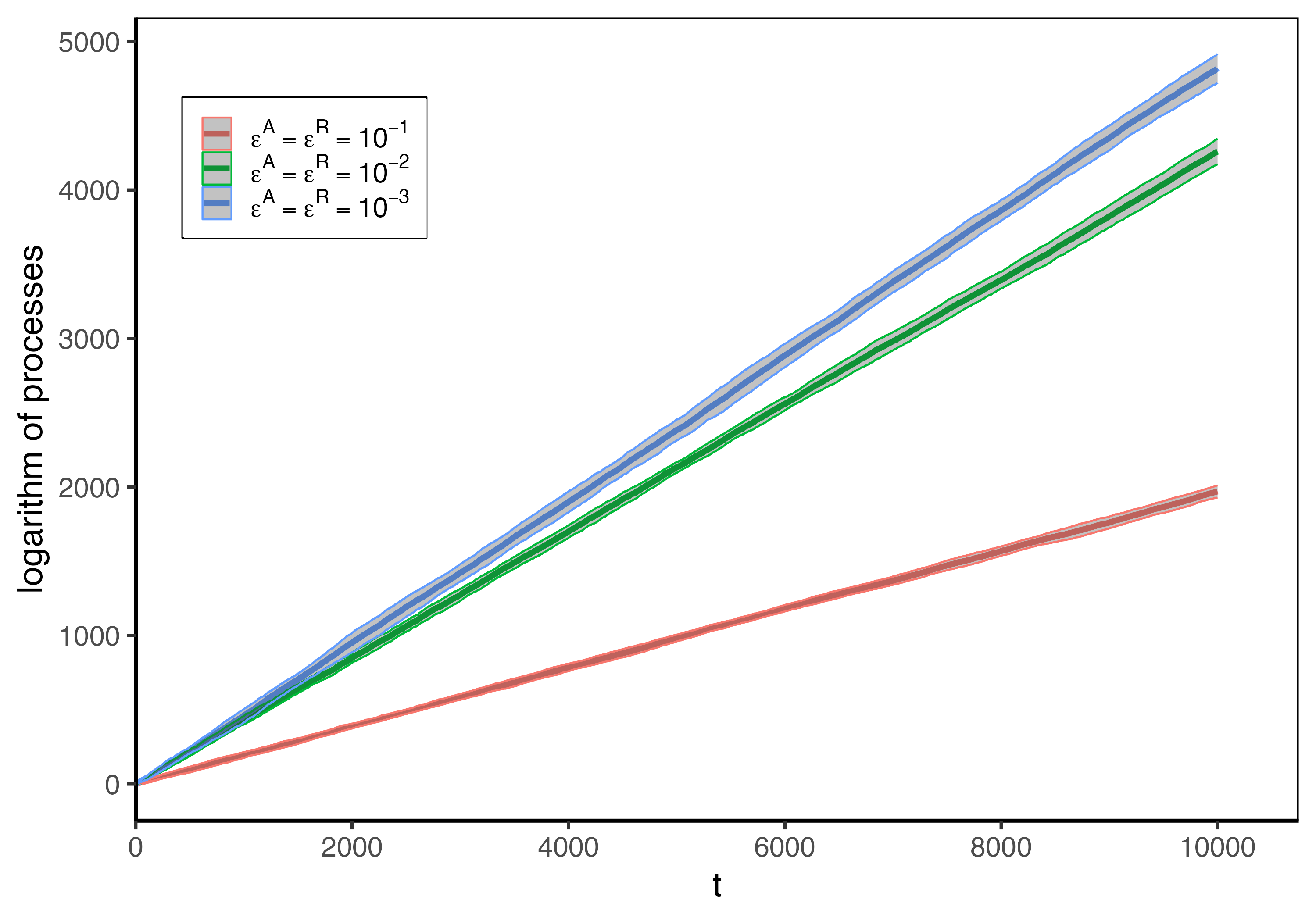}} 
\subfloat[Non-robust test]{\includegraphics[width=0.49\linewidth,height=0.3\linewidth]{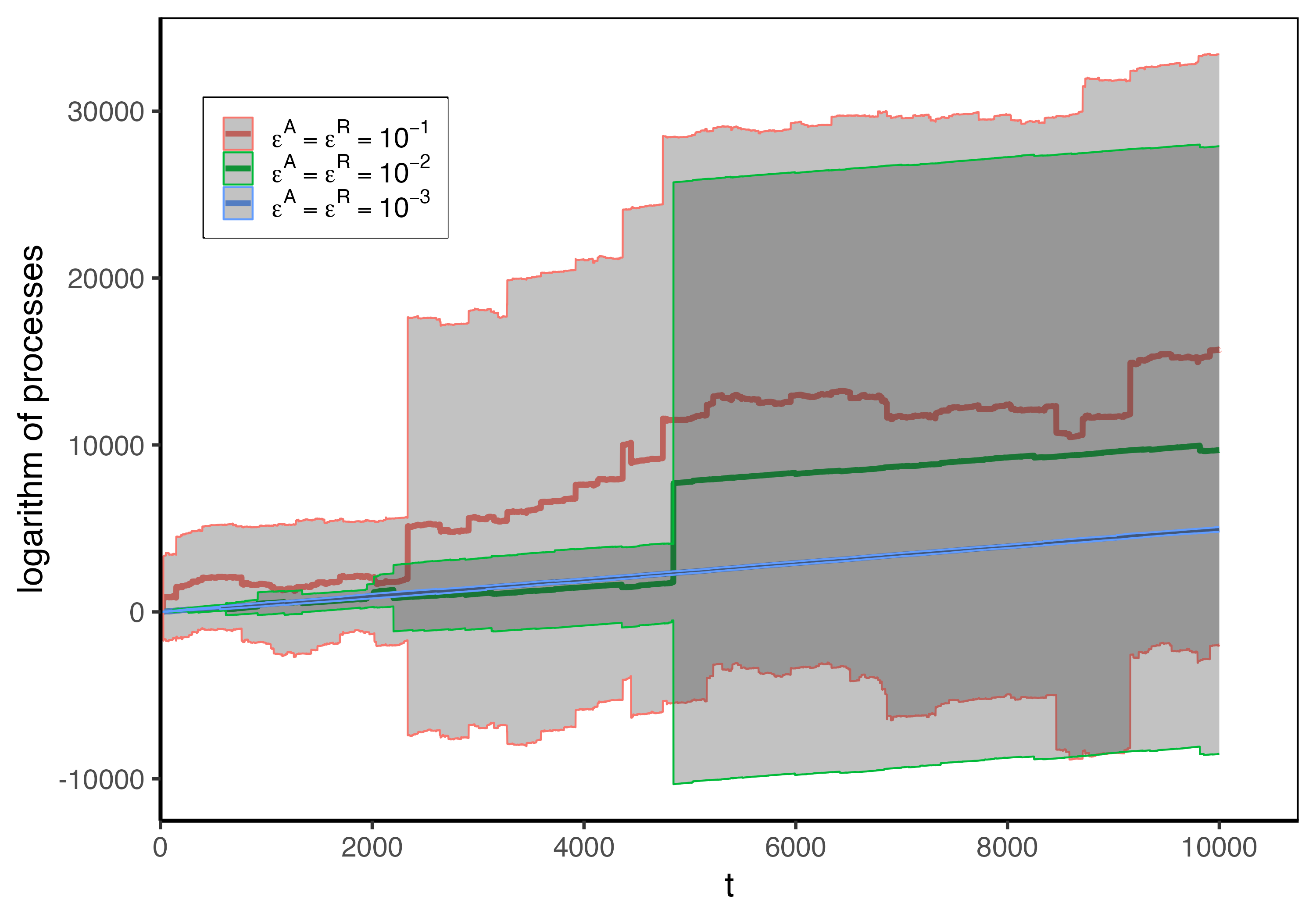}}
\caption[]{Data is drawn from $(1-\epsilon^R) \times N(1,1) + \epsilon^R \times\text{Cauchy}(-1,10)$ and $P_0=N(0,1), \mu_1=1$, $\epsilon^A=\epsilon^R=0.1,0.01,0.001$. The growth rate of our robust tests increases as $\epsilon$ decreases.} 
\label{fig:error-bar}  
\end{figure*}

Now, we consider the (pre-contaminated) null $P_0$ to be $N(0,1)$, the alternative to be $P_1=\{N(\mu,1):\mu\neq 0\}$. We employ the robust predictable plug-in method test described in \Cref{sec:comp-general} since there is no LFD. For both the non-robust predictable plug-in method and our robustified predictable plug-in method for the composite alternative, we use the sample median as an estimate of $\mu$. \Cref{fig:error-bar-comp} is the average of $10$ independent simulations, along with error bars.

\begin{figure*}[!htb]

\centering
\centering
\subfloat[Robust test]{\includegraphics[width=0.49\linewidth,height=0.3\linewidth]{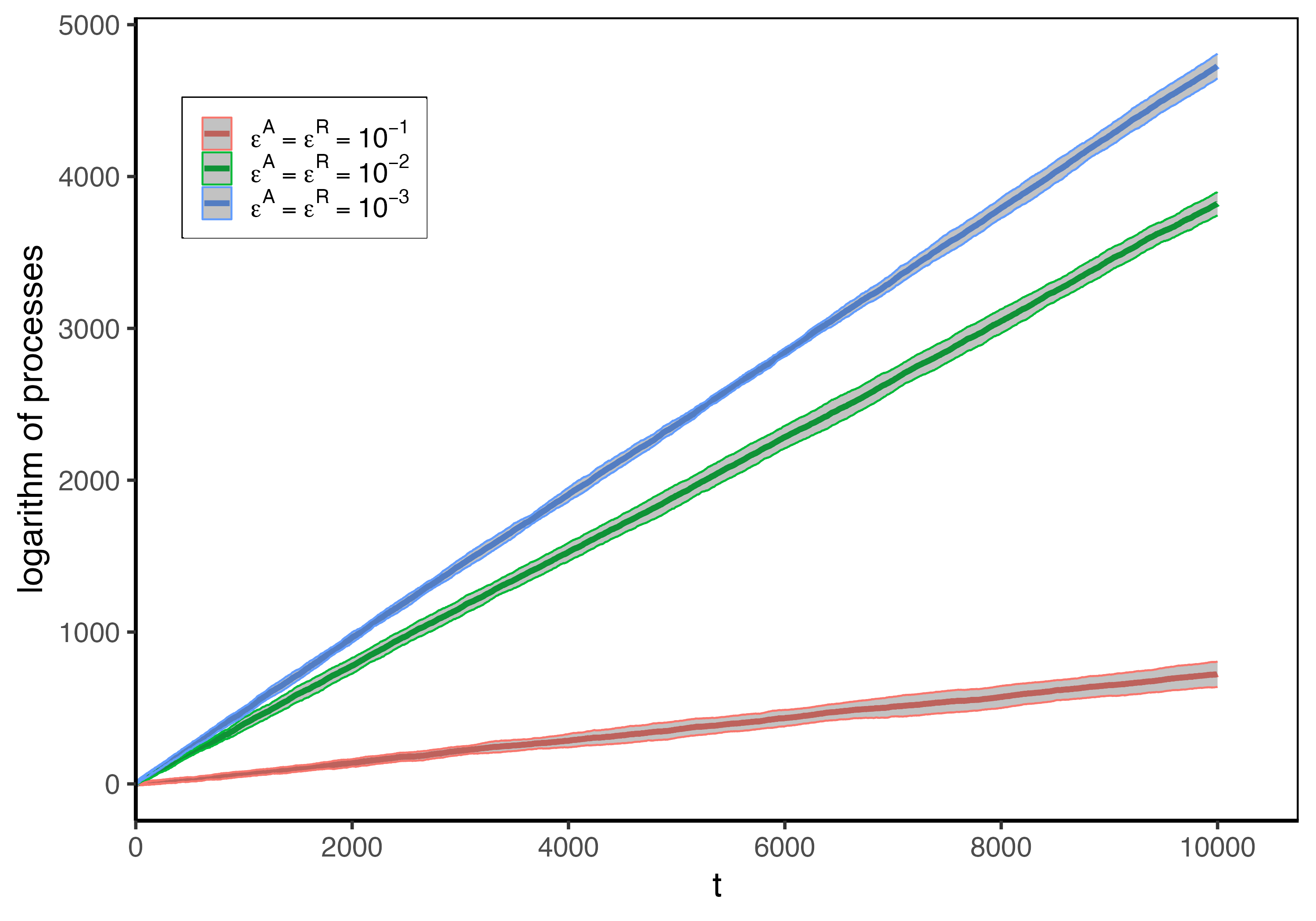}} 
\subfloat[Non-robust test]{\includegraphics[width=0.49\linewidth,height=0.3\linewidth]{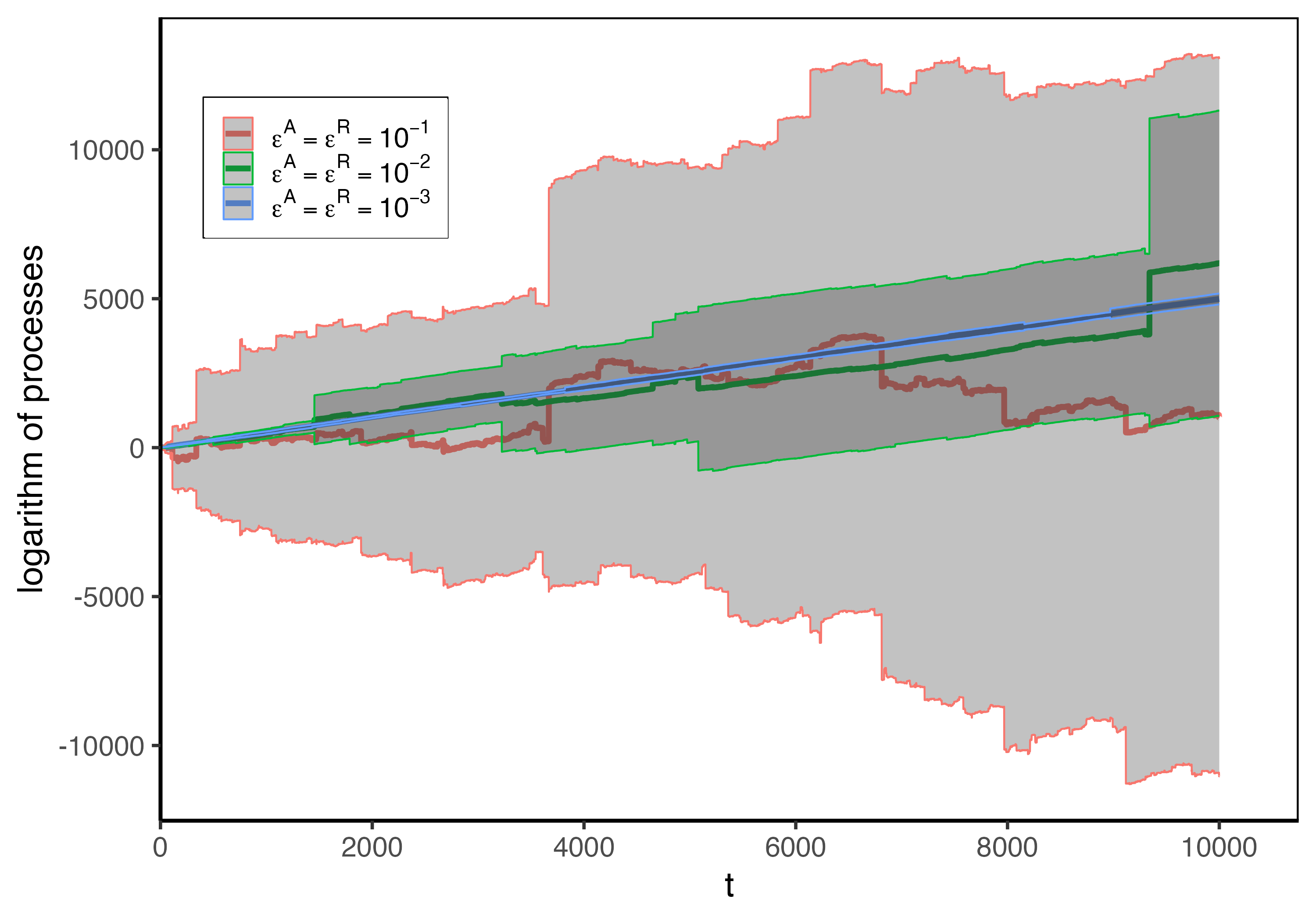}}
\caption[]{Data is drawn from $(1-\epsilon^R) \times N(1,1) + \epsilon^R \times\text{Cauchy}(-1,10)$ and $P_0=N(0,1), \mu_1=1$, $\epsilon^A=\epsilon^R=0.1,0.01,0.001$. The growth rate of our robust tests increases as $\epsilon$ decreases.} 
\label{fig:error-bar-comp}  
\end{figure*}

 We observe that the robust test exhibits very low variability across runs, resulting in narrow error bars even with only $10$ repetitions, whereas the non-robust test displays substantially higher variance. 
These figures correspond to \Cref{fig:different-eps}, but the robust and non-robust tests are displayed in separate subfigures to highlight the difference in variability better. The other figures in \Cref{sec:expt} can similarly be split; however, we omit them here because their behavior is nearly identical. The key takeaway is that the robust tests consistently exhibit very small variance, whereas the non-robust tests show huge variability.

\subsection{Empirical type-I error and power}
We conduct three experiments to compare the type I error and power of our robust tests and classical non-robust tests. 
\begin{itemize}
    \item Experiment 1: Pre-contaminated null and alternative are $\mathcal{P}_0=N(0,1)$ and $\mathcal{P}_1=N(1,1)$ respectively. 
     \item Experiment 2: Pre-contaminated null and alternative are $\mathcal{P}_0=\{N(\mu,1):-0.5\leq\mu\leq 0.5\}$ and $\mathcal{P}_1=N(1,1)$ respectively.
    \item Experiment 3: Pre-contaminated null and alternative are $\mathcal{P}_0=\{N(\mu,1):\mu\leq 0\}$ and $\mathcal{P}_1=\{N(\mu,1):\mu\geq 1\}$ respectively. 
     
\end{itemize}

In all experiments, we operate under the $\epsilon$-contamination model, and null samples are generated from $(1-\epsilon) \times N(0,1) + \epsilon \times\text{Cauchy}(-1,10)$ and alternative samples are drawn from $(1-\epsilon) \times N(1,1) + \epsilon \times\text{Cauchy}(-1,10)$  with $\epsilon=0.01$. For the robust tests, we apply likelihood ratio tests constructed using the LFDs of $\mathcal{P}_0^\epsilon$ versus $\mathcal{P}_1^\epsilon$ (with the same $\epsilon = 0.01$). For comparison, the non-robust tests are the classical likelihood ratio tests based on the LFDs of the uncontaminated model classes $\mathcal{P}_0$ versus $\mathcal{P}_1$.

Note that under the null, the robust tests are unlikely to stop, making a direct estimation of the type-I error $\mathbb{P}_{P_0}[\tau < \infty]$ impractical. However, note that it can be written as $\mathbb P_{P_0}[\tau<\infty]=\mathbb E_{P_1}[L_\tau\mathds{1}(\tau<\infty)]$, where $L_t:=\prod_{i=1}^t\frac{dP_0}{dP_1}(X_i)$ is the likelihood ratio process. So, we draw alternative samples and compute the empirical mean of $L_\tau\mathds{1}(\tau<\infty)$ (because under the alternative, we have $\tau<\infty$ almost surely). However, the non-robust tests tend to stop under both the null and alternative, so for these procedures we directly estimate type-I error and power by imposing a maximum sample size of $50{,}000$.

We report the results of $500$ simulation runs in \Cref{tab:simple}. Our robust tests have empirical type-I errors remaining below the nominal level $\alpha=0.05$, and achieve power $1$ in all settings. However, the type-I errors of the non-robust tests are very high.

 \begin{table}[!ht]
    \centering
    \caption{
Empirical type-I error and power (averaged over $500$ runs) of the robust and non-robust tests with target level $0.05$.
    }
\label{tab:simple}
    \resizebox{0.4\linewidth}{!}{
    \begin{tabular}{c|cccc}
    \toprule
    \addlinespace
  &\multicolumn{2}{c}{Type-I error} &
\multicolumn{2}{c}{Power} \\
\cmidrule(lr){2-3} \cmidrule(lr){4-5} &
Robust  &  Non-robust  & Robust  &  Non-robust \\
\midrule
\addlinespace 
%100 & 10000&   0.97  &  0.96  & 15.38 & \textbf{12.21} & 2.97 & 17.63\\
Experiment 1 &    0.048 &  0.992    & 1 & 1\\
%500 & 10000&   0.97  & 0.95 & 15.96 & \textbf{12.88} & 2.59 & 17.26\\
Experiment 2  & 0.044  & 1  & 1 & 0.998\\
Experiment 3  & 0.047  & 0.986 & 1 & 1\\
\bottomrule
 \end{tabular}}
 \end{table}
 }

\bibliography{ref}

\end{document}